\documentclass[10pt, oneside]{article}
\usepackage[font=small,labelfont=bf]{caption}
\usepackage{hyperref}
\usepackage{bm}
\usepackage{fouriernc}
\usepackage{url}
\usepackage{longtable}
\usepackage{mathrsfs}
\usepackage{multirow}
\usepackage{bigstrut}
\usepackage{amssymb}
\usepackage{graphicx}
\usepackage[final]{epsfig}
\usepackage[fleqn]{amsmath}
\usepackage{subfigure}
\usepackage{amssymb}
\usepackage[amsmath,amsthm,thmmarks]{ntheorem}
\usepackage{latexsym}
\usepackage[allnumber,warning,easyold,fleqn,leqno,math]{easyeqn}
\usepackage[strict]{changepage}
\newcommand{\bs}{\boldsymbol}

\newcommand{\eps}{{\varepsilon}}

\newcommand{\beq}{\begin{equation}}
\newcommand{\eeq}{\end{equation}}
\newcommand{\beqs}{\begin{eqnarray}}
\newcommand{\eeqs}{\end{eqnarray}}
\newcommand{\beql}{\begin{equation} \label}

\newcommand{\lel}{ \ {\scriptstyle \le} \ }

\def\Xint#1{\mathchoice
{\XXint\displaystyle\textstyle{#1}}%
{\XXint\textstyle\scriptstyle{#1}}%
{\XXint\scriptstyle\scriptscriptstyle{#1}}%
{\XXint\scriptscriptstyle\scriptscriptstyle{#1}}%
\!\int}
\def\XXint#1#2#3{{\setbox0=\hbox{$#1{#2#3}{\int}$ }
\vcenter{\hbox{$#2#3$ }}\kern-.55\wd0}}

\def\dashint{\Xint-}

\DeclareMathSizes{10}{11}{8}{7}
\numberwithin{equation}{section}

\begin{document}

\begin{center}
{ \Large Mathematical modeling of magnetostrictive nanowires for sensors/devices with application to Galfenol}  \\
\vspace{0.2in}
{\it Krishnan Shankar \\Aerospace Engineering and Mechanics, University of Minnesota 
}
\vspace*{2mm}
\end{center}

\begin{abstract}
Magnetostrictive wires of diameter in the nanometer scale have been proposed for application as acoustic sensors \cite{downey2008}, \cite{yang2006distant}. The sensing mechanism is expected to operate in the bending regime. In this work we derive a variational theory for the bending of magnetostrictive nanowires starting from a full 3-dimensional continuum theory of magnetostriction. We recover a theory which looks like a typical Euler-Bernoulli bending model but includes an extra term contributed by the magnetic part of the energy. The solution of this variational theory for an important, newly developed magnetostricitve alloy called Galfenol $ \big( \mbox{cf. } \cite{clark2000} \big)$ is compared with the result of experiments on actual nanowires $ \big( \mbox{cf. } \cite{downey2008thesis} \big)$ which shows agreement.   
\end{abstract}
{\centering \section{Introduction} }
\label{Sec-Introduction}

Magnetostrictive solids are those in which reversible elastic deformations are caused by changes in the magnetization. These materials have a coupling of  ferromagnetic energies with elastic energies. Typically magnetostriction is a small effect in the range of 20-200 ppm for commonly occurring ferromagnetic materials like Fe, Co and Ni alloys. In the 1970's giant magnetostrictive alloys like $ Tb_{0.3}Dy_{0.7}Fe_2 $ were developed. This alloy called Terfenol has high magnetostriction of the order $ \sim 2000 $ ppm, but is very brittle, and has low tensile strength of the order $ \sim 100 $ MPa. For this reason in most sensor/actuator applications it is used under compressive strain. Recent research by Clark et al. \cite{clark2000} has led to the development of a new alloy called Galfenol with formula $ Fe_{100-x}Ga_{x} $ where $ x $ ranges from $ 10\%-30\% $. These alloys have relatively high magnetostriction $ \sim 400 $ ppm and high tensile strengths $ \sim 400 MPa$.
 
In recent years a lot of new experimental techniques have been developed to manufacture ferromagnetic wires of nanometer diameter such as electron-beam lithography, step growth electro-deposition, and template-assisted electro-deposition. A possible application of these nanosize wires is in making acoustic sensors. The inspiration for this application comes from the structure of the human ear. The inner ear has fine cilia like hair whose response to impinging acoustic waves is transmitted by the nervous system to the brain. Such biologically inspired devices have been proposed to detect acoustic, fluid flow and tactile inputs (cf. \cite{yang2006distant} ). One possible arrangement of galfenol nanowires is in the form of an array depicted in Fig~\ref{sensor} . Here impinging acoustic waves are expected to change the magnetization of the wire array by inducing bending deformation.
\begin{figure}[!h]
\begin{center}
\includegraphics[scale=0.40]{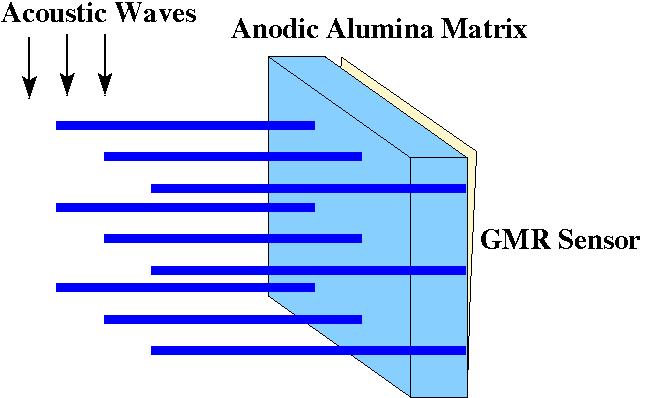}
\caption{Proposed model device using nanowires of Galfenol}
\label{sensor}
\end{center}
\end{figure}

The models of a vibrating string and the bending of a beam are important models in elasticity which are known to approximate the full 3-D behavior of a deformable body in the linear strain regime. Starting in the 80's rigorous mathematical methods based on the theory of $\Gamma$-convergence were used to justify these 1-D models as the correct approximation of 3-D elasticity, loosely speaking under asymptotic conditions as the diameter of the 3-D body approaches zero. The basic references for these results are \cite{acerbi1991variational} and \cite{anzellotti1994dimension}, while reference for $\Gamma$ -convergence can be found in \cite{braides2002gamma}.

Meanwhile in the micromagnetics literature there has been extensive use of $\Gamma$-convergence based methods to derive reduced dimension models for ferromagnetic thin films. The earliest results in this direction are  \cite{gioia1997micromagnetics} and \cite{carbou2001thin}.
Since our nanowires are expected to be used for the proposed sensor application in the bending deformation regime, the main goal of this paper is to combine the ideas of the references cited above from the elasticity and micromagnetics literature to derive similar asymptotic models for magnetostrictive nanowires in bending. The nanowires we are modeling have diameters in the 10-100nm range with lengths in the range 2-5$\mu$m. We will show that the bending behavior of a magnetostrictive nanowire resembles the classical Euler-Bernoulli bending model with an extra term which comes from the magnetic part of the energy.

\S \ \ref{Sec-micromagnetics} gives a brief review of the continuum theory of magnetostriction and defines the classical energy $ \mathcal E(\widetilde{\bs m \,}, \widetilde{\bs u \,}) $ as a function of the magnetization-deformation pair $ (\widetilde{\bs m \,}, \widetilde{\bs u \,})$. The section \S \ \ref{Sec-Scaling} gives a simple heuristic argument to show the various scales of elastic and magnetic energy relevant to the final result. In \S \ \ref{Rescaling} we start with the energy $ \mathcal E(\widetilde{\bs m \,}, \widetilde{\bs u \,}) $ defined on a wire of diameter $ \eps $ and on rescaling the wire to have unit diameter, recover a new energy $ \mathcal I^\eps(\bs m,\bs u) $ which equals the energy $ \mathcal E(\widetilde{\bs m \,}, \widetilde{\bs u \, }) $ per unit wire cross-sectional area, and depends on a rescaled magnetization-deformation pair $(\bs m,\bs u)$ now defined on the wire with unit diameter.  
Starting with minimizers $ (\bs m^\eps, \bs u^\eps)$ of the energy $ \mathcal I^\eps(\bs m,\bs u) $ in \S \ \ref{Secfirstgammalimit} we derive the first variational limit problem which physically represents the magnetoelastic equivalent to the elastic theory of an extensible string. \S \ \ref{Secsecondgammalimit}  gives the next order correction to the first variational problem which only involves magnetic terms. \S \ \ref{Secthirdgammalimit} gives the following order variational problem which is the main result of this paper and describes the bending behavior of the magnetostrictive nanowires. Here we show that we can extract a deformation $ \bs w^\eps \ \big(\mbox{cf.} \ \eqref{chi} \big) $ from the energy minimizing pair $ (\bs m^\eps,\bs u^\eps) $ which itself minimizes an energy $ \mathcal I^o_2 \ \big(\mbox{cf.} \ \eqref{I0defn.} \big) $ where $ \mathcal I^o_2 $ is an energy which resembles the classical Euler-Bernoulli bending energy with some correction terms depending on the magnetization. The method of proof involves the idea of convergence of minimizers, and we do not use the more abstract $ \Gamma$-convergence method. The Appendix \ref{Sec-Linear-gamma} treats the magnetostatic energy separately.

Basic notation: $ \alpha,\beta,\gamma,\cdots $ are scalars; $ \bs a, \bs u,\bs m,\cdots $ denote vectors in $ \mathbb R^3$; $\bs A, \bs B, \bs E, \cdots$ are tensors in $ \mathbb R^{3 \times 3}$ and $ S^2 \subset \mathbb R^3$ represents the surface of the unit ball in $ \mathbb R^3$. Components of any vector $ \bs m $ are denoted by either $ m_1,m_2,m_3 $ or $ m_x,m_y,m_y$. For any matrix $ \bs A $, $ \bs A^{T} $ denotes the transpose of the matrix. We use standard function space notation of $L^2(\Omega,\mathbb R^3)$, $H^1(\mathbb R^3,\mathbb R^3)$, $H_0^1(\Omega,\mathbb R^3)$; for details refer \cite{adams2009sobolev}. By Young's inequality we mean $ 2 a b \le \delta^{-1} a^2 + \delta b^2 \ $ for $ \ \mathbb R \ni \delta > 0 $, \ a variation of the classical Young's inequality.

{\centering \section{Micromagnetics} \label{Sec-micromagnetics}  }

The initial model for ferromagnetic solids was proposed in \cite{landau1935} where they also derived a model for magnetization dynamics. The continuum theory of ferromagnetic materials was  developed in the work of Brown \cite{brown1963} which was subsequently expanded to a theory for magnetostriction in \cite{brown1966}, where a variational model for magnetostriction with small strain is developed. We give a brief presentation of Brown's work relevant to magnetostriction in this section.

Let $ \Omega_\eps $ be a smooth bounded reference configuration in $ \mathbb R^3 $ depending on a parameter $\eps$. In the following sections we fill specify this dependence. Let $ \widetilde{\bs m \,}( \bs y) $ be the magnetization vector at a point $ \bs y \in \Omega_\eps $. Below the Curie temperature, the magnetization is constrained to have constant euclidean norm i.e., 
\begin{align} |\widetilde{\bs m \,}( \bs y )| = m_s \quad  a.e. \qquad \bs y \in \Omega_\eps. \nonumber \end{align}
For a bounded domain, this constraint implies $ \widetilde{\bs m \,} \in L^{ p }(\Omega_\eps,m_s S^2), \ \forall \ 1 \le p \le \infty$. We extend $ \widetilde{\bs m \,} $ by $ 0 $ outside $ \Omega_\eps $ whenever necessary and denote it by $ \widetilde{\bs m \,} \chi_{\Omega_\eps} = \widetilde{\bs m \,}(\bs y) \,\chi_{\Omega_\eps}(\bs y) $ which as a result gives $ \widetilde{\bs m \,} \chi_{\Omega_\eps} \in L^{ p }(\mathbb{R}^3,\mathbb{R}^3), \ \forall \ 1 \le p \le \infty $. We denote by $ \widetilde{\bs u \,} \in H^1( \Omega_\eps, \mathbb R^3 ) $ the displacement map. The infinitesimal strain corresponding to $ \widetilde{\bs u \,}( \bs y) $ is, $\big( \nabla^{\bm y} \mbox{ is gradient w.r.t. } \bs y \big)$
\begin{align} \widetilde{\bs E}[\widetilde{\bs u \,}](\bs y) & = \frac{1}{2} \big( \nabla^{\bm y} \widetilde{\bs u \,}( \bs y) + \nabla^{\bm y} \widetilde{\bs u \,}( \bs y)^T \big). \end{align}

Interaction of the magnetization with the crystalline structure of a magnetic solid generates an interaction energy modeled by a function, $ \varphi : m_s S^2 \rightarrow [0, \infty) $. This energy has a finite number of wells (say $N$) along a set of constant magnetization vectors $ \big\{ \widetilde{\bs m \,}^{(k)} \big\} \in m_s S^2$ where the index $ k \in \big\{ \, 1,2,\cdots N \, \big\} $ and on which without loss of generality we can set $ \varphi \big( \widetilde{\bs m \,}^{(k)} \big)= 0 $. 
The anisotropy energy thus becomes,
\begin{align*} E_{anis} &=  \int_{\Omega_\eps} \varphi \big( \, \widetilde{\bs m \,}( \bs y) \big) \ \bs d  \bs y. \end{align*}
For cubic materials $ \varphi( \widetilde{\bs m \,} ) = \frac{\Pi_1}{m_s^4}( \widetilde{m}_1 ^2\widetilde{m}_2 ^2 + \widetilde{m}_1 ^2 \widetilde{m}_3 ^2 + \widetilde{m}_2 ^2\widetilde{m}_3 ^2 ) + \frac{\Pi_2}{m_s^6} (\widetilde{m}_1^2 \widetilde{m}_2 ^2 \widetilde{m}_3 ^2) $, which along with the constraint $ | \, \widetilde{\bs m \,} \, |=m_s $ gives that $0 \le \varphi ( \widetilde{\bs m \,} ) \le K_1$. Thus
\begin{align}
0 \le E_{anis}= \int_{\Omega} \varphi \big( \widetilde{\bs m \,}( \bs y) \big) \ \bs d  \bs y \, \le \, K_1 |\,\Omega_\eps|. \label{Anisotroybound}
\end{align}
The exchange energy penalizes variations in the magnetization in a body and thus tends to prefer constant magnetizations. It is modeled as follows,
\begin{align*} E_{exc} = d \int_{\Omega_\eps} \big| \nabla^{\bm y} \widetilde{\bs m \,} \big|^2 \, \bs d \bs y. \end{align*}
Here $ d $ is called the exchange constant. Magnetized bodies generate a magnetic self field in all of $ \mathbb R^3$. This field $ \widetilde{\bs h}^\eps_{\widetilde{\bs m}}(\bs y) $ is  given by the following equation,
\begin{align*} & \nabla^{\bm y} \cdot \big( -\nabla^{\bm y}  \widetilde \phi^\eps(\bs y) + 4\pi \widetilde{\bs m \,}(\bs y) \, \big) = 0 \qquad \forall  \bs y \in \mathbb R^3, \\
& \widetilde{\bs h}^\eps_{\widetilde{\bs m}}(\bs y) = -\nabla^{\bm y} \widetilde \phi^\eps(\bs y), \\
& \big[ \big| \nabla^{\bm y} \widetilde \phi^\eps \cdot \widetilde{\bs n} \big| \big] = \big[ \big| -\widetilde{\bs h}^\eps_{\widetilde{\bs m}} \cdot \widetilde{\bs n} \big| \big] = 4 \pi \ \widetilde{\bs m \,} \cdot \widetilde{\bs n} \qquad \mbox{on } \partial \Omega_\eps. \end{align*}
$ [| \cdot |] $ represents the jump of a quantity across any oriented surface with unit normal $ \tilde{\bs n} $. The demagnetization energy is generated by the interaction of the magnetization $\widetilde{\bs m \,} $ with $ \widetilde{\bs h}^\eps_{\widetilde{\bs m}} $ and equals
\begin{align} E_{demag}(\widetilde{\bs m \,}) = \frac{1}{8 \pi} \int_{ \mathbb{R}^3 } \big| \widetilde{\bs h}^\eps_{\widetilde{\bs m}} (\bs y )  \big|^2  \, \bs d \bs y = - \frac{1}{2} \int_{\Omega_\eps} \widetilde{\bs h}^\eps_{\widetilde{\bs m}}( \bs y ) \cdot \widetilde{\bs m \,}( \bs y) \, \bs d \bs y. \label{demagdefnition}\end{align}
A standard upper and lower bound for $ E_{demag} $ is given by
\begin{align} 0 \le  \ \ E_{demag}(\widetilde{\bs m \,}) =\frac{1}{8 \pi} \int_{ \mathbb{R}^3 } \big| \widetilde{\bs h}^\eps_{\widetilde{\bs m}} (\bs y )  \big|^2  \, \bs d \bs y \ \ \le \frac{1}{2} \int_{ \Omega_\eps} |\widetilde{\bs m \,}( \bs y)|^2 \, \bs d \bs y = \frac{1}{2} \, |\, \Omega_\eps| \, m_s^2, \label{demagstandardbound} \end{align}
since $|\widetilde{\bs m \,}|=m_s$. The energy of interaction between an external applied field $ \widetilde{\bs h_a} \in L^2(\Omega,\mathbb R^3)$ and the magnetization over the body is modeled by the following,
\begin{align*} E_{app}(\widetilde{\bs m \,}) = -\int_{\Omega_\eps} \widetilde{\bs h_a}(\bs y) \cdot \widetilde{\bs m \,}(\bs y) \,\bs d \bs y. \end{align*}
which along with H$\ddot{\mbox{o}}$lder's inequality gives
\begin{align}
& -K_2 \le E_{app}(\widetilde{\bs m \,}) \le K_2, &  K_2= \big\Vert \widetilde{\bs h_a} \big\Vert_{L^2(\Omega_\eps)} \big\Vert \widetilde{\bs m \,} \big\Vert_{L^2(\Omega_\eps)}. \label{Zeemanbound}
\end{align} 
The elastic energy for the magnetoelastic solid for small strains is given by,
\begin{align*} E_{el} = \int_{\Omega_\eps} \frac{1}{2} \big( \, \widetilde{\bs E}[\widetilde{\bs u}] - \widetilde{\bs E}_s( \widetilde{\bs m \,} ) \, \big) : \mathbb C \big[ \widetilde{\bs E}[\widetilde{\bs u}]- \widetilde{\bs E}_s( \widetilde{\bs m \,} ) \big] \, \bs d \bs y. \end{align*}
In this paper by an abuse of notation, we write the above integrand as 
\begin{align*} \big( \, \widetilde{\bs E}[\widetilde{\bs u}] - \widetilde{\bs E}_s( \widetilde{\bs m \,} ) \, \big) : \mathbb C \big[ \widetilde{\bs E}[\widetilde{\bs u}] - \widetilde{\bs E}_s( \widetilde{\bs m \,} ) \big] = \mathbb C \big[ \widetilde{\bs E}[\widetilde{\bs u}] - \widetilde{\bs E}_s( \widetilde{\bs m \,} ) \big]^2. \end{align*} 
$\mathbb C$ is a positive definite fourth order tensor. Here $ {\widetilde{\bs E}_s( \widetilde{\bs m \,} ) } $ is the spontaneous strain due to magnetization
\begin{align} \widetilde{\bs m \,} \mapsto {\widetilde{\bs E}_s( \widetilde{\bs m \,} ) } \in \bs M^{3 \times 3 }_{sym}, \nonumber \end{align}
where $ \bs M^{3 \times 3 }_{sym} $ denotes the set of symmetric matrices of $ 3 \times 3 $ dimension. For cubic materials it's form is
\[ \widetilde{\bs E}_s (\widetilde{\bs m \,}) = \frac{3}{2 m_s^2} \displaystyle{ \begin{bmatrix} \lambda_{100} \, \widetilde{m_1}^2  &  \lambda_{111} \, \widetilde{m_1} \widetilde{m_2} & \lambda_{111} \, \widetilde{m_1}\widetilde{m_3} \\ \lambda_{111} \, \widetilde{m_1}\widetilde{m_2} & \lambda_{100} \, \widetilde{m_2}^2 & \lambda_{111} \, \widetilde{m_2}\widetilde{m_3} \\ \lambda_{111} \, \widetilde{m_1}\widetilde{m_3} & \lambda_{111} \, \widetilde{m_2}\widetilde{m_3} & \lambda_{100} \, \widetilde{m_3}^2 \end{bmatrix} } - \frac{\lambda_{100}}{2} \, I  \] 
where $ I $ is the Identity matrix in $ \mathbb R^{3 \times 3} $. The form for $ \widetilde{\bs E}_s (\widetilde{\bs m \,})$ and $|\widetilde{\bs m \,}|=m_s$ gives
\begin{align}
&\big|\widetilde{\bs E}_s (\widetilde{\bs m \,})\big| \le K_3, && \mbox{where } \big|\widetilde{\bs E}_s (\widetilde{\bs m \,})\big|^2=tr\big(\widetilde{\bs E}_s (\widetilde{\bs m \,})^T \widetilde{\bs E}_s (\widetilde{\bs m \,})\big)\label{elasticbound1}. \end{align}
The norm defined used above is the standard Frobenius norm, i.e. for any matrix $\bs M$, $\big| \bs M \big|^2$ is defined as the trace of $\bs M^T \bs M$. Also $ \mathbb C $ being symmetric positive definite $4^{th} $ order tensor gives for some $ \gamma, \Gamma > 0$,
\begin{align}
& \Gamma \, \big| \bs M^2 \big| \, \ge \, \mathbb C \big[ \bs M \big]^2 \, \ge \, \gamma \, \big| \bs M^2 \big| \, , && \forall \bs M \in \bs M_{sym}^{3 \times 3 }.  \label{elasticbound2}
\end{align}
In addition to these, energy due to external force acting on the body in the form of body force or surface traction is included in the general energy. However since these terms are lower order in deformation $\widetilde{\bs u}$, they do not affect the final form of the limit problem. For our investigation in this paper, we neglect this term to reduce the length of the computation. Thus the full energy functional for magnetostriction is,
\begin{align} \hspace{-6mm} \mathcal E^\eps( \widetilde{\bs m \,},\widetilde{\bs u}) &= E_{exc} + E_{anis} + E_{app} + E_{el} + E_{demag} \nonumber \\
&= \int_{\Omega_\eps} \Big\{ d \big| \nabla^{\bm y} \widetilde{\bs m \,} \big|^2 +\varphi( \widetilde{\bs m \,} ) - \widetilde{\bs h_a} \cdot \widetilde{\bs m \,} + \frac{\mathbb C}{2} \big[ \widetilde{\bs E}[\widetilde{\bs u}] - \widetilde{\bs E}_s( \widetilde{\bs m \,} ) \big]^2 \, \Big\} \bs d \bs y +\frac{1}{8 \pi}\int_{ \mathbb R^3 } \big| \widetilde{\bs h}^\eps_{\widetilde{\bs m}} \big|^2 \bs d \bs y.  \label{mathcale}
\end{align}
{ \centering
\section{Heuristic Scaling of energy} \label{Sec-Scaling}
}
In \S \ \ref{Sec-Scaling1} and \S \ \ref{Sec-Scaling2}, we start with a cylindrical domain with radius $ \eps $ and length 1. We then show how both an isotropic linear elastic energy and the magnetostatic energy defined in equation \eqref{demagdefnition} scale with respect to $ \eps$. The scaling of the linear elastic energy has been know for long in the engineering literature, but a rigorous derivation starting from a three-dimensional linear elastic theory is relatively recent. 
\subsection{Linear Elastic Energy} \label{Sec-Scaling1}
Let $ \Theta=\big\{ (y_1,y_2) \in B_\eps(0) \ , \ y_3 \in (0,1) \big\} $ be a cylindrical domain of radius $ \eps $ centered at the origin and length $ 1 $ with axis aligned along the $ y_3 $ axis. Let $ Y $ be the Young's Modulus, $ A =\pi\eps^2$ is the cross-sectional area,  and $ I =\frac{\pi}{4} \eps^4$ be the second moment of area of the cross section. Let $ (\widetilde{u}_1,\widetilde{u}_2,\widetilde{u}_3) $ be the displacements in $ (y_1,y_2,y_3) $ directions. From the engineering literature we know that the extensional energy of a rod along its axis is given as
\begin{align}
\int^1_0 Y \, A \, \big| \partial_3 \widetilde{u}_3 \big|^2 d y_3= Y \pi \eps^2 \int^1_0 \big| \partial_3 \widetilde{u}_3 \big|^2 d y_3 \approx O(\eps^2),
\end{align}  
where $ \widetilde{u}_3 $ is the extension of the rod along its axis. From the Euler-Bernoulli model for a beam bending in the direction of the $ x_1 $ axis, the bending energy is
\begin{align}
\int^1_0 Y \, I \, \big| \partial_{33} \widetilde{u}_1 \big|^2 d y_3= Y \frac{\pi\eps^4}{4} \int^1_0 \big| \partial_{33} \widetilde{u}_1 \big|^2 d y_3 \approx O(\eps^4).
\end{align}  
The different scaling of the two energies with respect to $ \eps$ suggests to us that a linear elastic isotropic energy of the form 
\begin{align}
\mathcal W^\eps(\widetilde{\bs u})=\int_{\Theta} \Big\{ \, \mu \big| \bs E(\widetilde{\bs u}) \big|^2+\frac{\lambda}{2} \big| \mbox{ tr}\big(\bs E( \widetilde{\bs u}) \, \big) \, \big|^2 \ \Big\} \bs d \bs y
\end{align}
should factor into terms which are of different orders in powers of $\eps$. Using $\Gamma$-convergence this factorization into orders of powers of $\eps$ has been proven in \cite{anzellotti1994dimension}. They have shown that,
\begin{align} 
\mathcal W^\eps(\widetilde{\bs u}) = \eps^2 \mathcal W_1(\widehat{u}_3) + \eps^4 \mathcal W_2(\widehat{u}_1,\widehat{u}_2,\widehat{u}_4)+\mbox{ higher order terms}
\end{align}
where $\widehat{\bs u \,}(y_3) \equiv (\widehat{u}_1,\widehat{u}_2,\widehat{u}_3)(y_3)$ and $ \ \widehat{\bs u \,}(y_3)= \mbox{\fontsize{9}{8}\selectfont $\displaystyle{ \frac{1}{|B_\eps(0)|} \int_{B_\eps(0)}}$} \widetilde{\bs u}(\bs y) \ d y_1 d y_2 \ $ is the averaged cross-sectional displacement and $ \ \widehat{u}_4(y_3)=\mbox{\fontsize{9}{8}\selectfont $\displaystyle{ \frac{2}{\eps^2|B_\eps(0)|} \ \int_{B_\eps(0)}}$} \big( \, y_2 \widetilde{u}_1-y_1 \widetilde{u}_2 \, \big) d y_1 d y_2 \ $ gives the torsional component.
\subsection{Magnetostatic energy} \label{Sec-Scaling2}
For an ellipsoidal body it is well known cf. \cite{james2007} that the demagnetization field $\widetilde{\bs h}^\eps_{\widetilde{\bs m}} $ and the corresponding demagnetization $ E_{demag} $ for a constant magnetization $ \widetilde{\bs m \,} $ are,  
\begin{align} 
& \widetilde{\bs h}^\eps_{\widetilde{\bs m}}=-4\pi\bs D \widetilde{\bs m \,}, & E_{demag}=2\pi \ \big( \mbox{ Volume of body} \big) \times \bs D \widetilde{\bs m \,} \cdot \widetilde{\bs m \,}
\end{align} 
where $ \bs D \in \mathbb R^{3\times3} $ is called demagnetization tensor. $\bs D$ is independent of position $ \bs y $, and has trace $1$. For non-ellipsoidal bodies supporting a constant magnetization $ \widetilde{\bs m \,} $, it still is true that $\widetilde{\bs h}^\eps_{\widetilde{\bs m}}=-4\pi\bs D \widetilde{\bs m \,}$. However the demagnetization tensor $\bs D $ (with trace still 1) now depends on position $ \bs y $. The magnetostatic energy is now given by $ E_{demag}= \ 2 \pi \big( \mbox{ Volume of body} \big) \times \widehat{\bs D} \widetilde{\bs m \,} \cdot \widetilde{\bs m \,},
$ where $ \widehat{\bs D} $ is the volumetric average of $ \bs D $. For our cylindrical domain $ \Theta=B_\eps(0) \times (0,1) $, $ \widehat{\bs D} $ is a diagonal matrix with entries
\begin{align}
&\widehat{D}_{33}=\frac{8\eps}{3\pi}-\frac{\eps^2}{2}+O(\eps^4), & \widehat{D}_{11}=\widehat{D}_{22}=\frac{1}{2}-\frac{4\eps}{3\pi}+\frac{\eps^2}{4}+O(\eps^4). \nonumber
\end{align}
See \cite{joseph1966ballistic} for a simple derivation of this result. The demagnetization energy for a constant magnetization $ \widetilde{\bs m \,}=(\widetilde{m \,}_1,\widetilde{m \,}_2,\widetilde{m \,}_3) $ is given by
\begin{align} \hspace{-5mm}
E_{demag}&=2\pi^2 \Big[ \eps^2 \frac{\widetilde{m \,}_1^2+\widetilde{m \,}_2^2}{2} - \eps^3 \frac{4}{3\pi} \Big( \ \big( \, \widetilde{m \,}_1^2+\widetilde{m \,}_2^2 \, \big) -2 \widetilde{m \,}_3^2  \Big) + \frac{\eps^4}{2} \Big( \frac{\widetilde{m \,}_1^2+\widetilde{m \,}_2^2}{2} - \widetilde{m \,}_3^2 \Big) \Big] \nonumber \\
&=\pi^2 \, \big( \, \widetilde{m \,}_1^2+\widetilde{m \,}_2^2\, \big) \ \eps^2+\eps^3 Q_1(\widetilde{\bs m})+\eps^4 Q_2(\widetilde{\bs m}), \label{magnetostaticM}
\end{align}
where $Q_1(\widetilde{\bs m}):=\frac{8\pi}{3}\Big( \ \big( \, \widetilde{m \,}_1^2+\widetilde{m \,}_2^2 \, \big) -2 \widetilde{m \,}_3^2 \Big)$ and $Q_2(\widetilde{\bs m}):=\frac{\pi^2}{2}\Big( \ \big( \, \widetilde{m \,}_1^2+\widetilde{m \,}_2^2 \, \big) -2 \widetilde{m \,}_3^2 \Big)$.
Thus for a cylindrical domain $\Theta$ with constant magnetization we can already see the presence of various orders of scales in the magnetostatic and elastic energy. 


{\centering \section{Rescaling} \label{Rescaling} }

In this section we rescale the domain $\Omega_\eps $ depending on a parameter $\eps$ to a fixed domain $\Omega $. The space variable in the original domain $\Omega_\eps $ is either denoted by $\bs y $ or $\bs z$ and in the rescaled domain by $\bs x$. The gradient operators w.r.t. $\bs y$ and $\bs z$ are denoted by $\nabla^{\bm y} $ and $\nabla^{\bm z} $ respectively and gradient w.r.t. $ \bs x$ is denoted as just $\nabla$. All variables in the original domain $\Omega_\eps$ come with the tilde notation, for e.g. $\widetilde{\bs m}$ while variables in the rescaled domain are plain e.g. $\bs m$. For any vector $\bs v \in \mathbb R^3$, we will write $ \bs v =(v_1,v_2,v_3) = (\bs v_p,v_3)$ where $p=1,2$ and $\bs v_p$ denotes the planar component of $\bs v$. Analogously the gradient operator may be denoted by $\nabla = (\nabla_p, \partial_3)$.

 Let $ \Omega_\eps := \big[ \, \bs y_p \in \omega_\eps, \, y_3 \in (0,1) \, \big] \ \ $ be a domain with cross-section $ \omega_\eps \subset \mathbb R^2$ where $\omega_\eps $ is any Lipschitz domain in 2-dimensions. While the results of all the subsequent sections in this paper hold for any arbitrary cross-section $ \omega_\eps$, however for the sake of simplicity we set 
\begin{align}\omega_\eps = B_\eps(0) \subset \mathbb R^2  \end{align} 
a ball of radius $ \eps $ in 2-D. We rescale the domain $\Omega_\eps $ to $ \Omega$ by the following one-to-one map  
\begin{align}
x_1=\frac{y_1}{\eps} & & x_2=\frac{y_2}{\eps} & & x_3=y_3. 
\end{align}
By the rescaling $ \Omega= \big[ \, \bs x_p \in \omega, \, x_3 \in (0,1) \, \big]$ where $\omega $ is now a ball with unit radius in 2 dimensions. We rescale the fields $ \widetilde{\bs m}(\bs y)$, $ \widetilde{\bs u}(\bs y) $, $ \widetilde{\bs h_a}(\bs y)$, and $ \widetilde{\bs h}^\eps_{\widetilde{\bs m}}(\bs y)$ using the one-to-one maps
\begin{align}
\bs m(\bs x) = \widetilde{\bs m}(\bs y), && \bs u(\bs x) = \widetilde{\bs u}(\bs y), && \bs h_a(\bs x)=\widetilde{\bs h_a}(\bs y), && \bs h^\eps_{\bs m}(\bs x) = \widetilde{\bs h}^\eps_{\widetilde{\bs m}}(\bs y). \label{vectorscale}
\end{align}

The map $\bs m(\bs x) = \widetilde{\bs m}(\bs y)$ being one-to-one means that we can invert the rescaled magnetization $\bs m(\bs x)$ back to the unscaled magnetization $\widetilde{\bs m}(\bs y)$. Also while the pair $\big( \widetilde{\bs h}^\eps_{\widetilde{\bs m}},\widetilde{\bs m} \big)$ satisfies Maxwell's equation on $\Omega_\eps$, the rescaled pair $\big(\bs h^\eps_{\bs m},\bs m \big)$ does not satisfy Maxwell's equation on $\Omega$. However unscaling the pair $\big(\bs h^\eps_{\bs m},\bs m \big)$ to $\big( \widetilde{\bs h}^\eps_{\widetilde{\bs m}},\widetilde{\bs m} \big)$, solves Maxwell's equation on $\Omega_\eps$. Hence the $\eps$ superscript on $\bs h^\eps_{\bs m}$.

The gradient operator $\nabla^{\bm y} = (\nabla^{\bm y}_p, \partial^{\bm y}_3)$ operating on $ \widetilde{\bs m}(\bs y) $ or $ \widetilde{\bs u}(\bs y) $ correspondingly scales as,
\begin{align}
\nabla^{\bm y}_p \,\widetilde{\bs m}(\bs y) = \frac{1}{\eps} \nabla_p \, \bs m(\bs x), && \partial_3^{\bm y} \,\widetilde{\bs m}(\bs y) = \partial_3 \, \bs m(\bs x). \nonumber
\end{align}
Using the scaling of gradients, we rescale the strain $\widetilde{\bs E}[\widetilde{\bs u}](\bs y)$ to get a new field $ \bs \kappa^\eps(\bs x) $ as
\begin{align} \widetilde{\bs E}[\widetilde{\bs u}](\bs y) &= \frac{1}{2} \, \big[ \ \nabla^{\bm y} \widetilde{\bs u}(\bs y) + \nabla^{\bm y} \widetilde{\bs u}(\bs y)^T \ \big] \nonumber \\
&= \begin{bmatrix} \frac{1}{\eps}\partial_1  u_1(\bs x) &  \frac{1}{2 \eps} \big( \, \partial_1 u_2 + \partial_2  u_1\, \big)(\bs x) & \frac{1}{2} \big( \, \frac{1}{\eps} \partial_1 u_3 + \partial_3 u_1\, \big)(\bs x)  \\ \frac{1}{2 \eps}\big( \,\partial_1 u_2 + \partial_2  u_1\, \big)(\bs x) &  \frac{1}{\eps} \partial_2 \ u_2(\bs x) &  \frac{1}{2}\big( \, \frac{1}{\eps} \partial_2 u_3 + \partial_3 u_2\, \big)(\bs x)  \\ \frac{1}{2}\big( \, \frac{1}{\eps}  \partial_1 u_3 + \partial_3 u_1\, \big)(\bs x)  & \frac{1}{2}\big( \, \frac{1}{\eps} \partial_2 u_3 + \partial_3 u_2 \, \big)(\bs x) & \partial_3 u_3(\bs x) \end{bmatrix} \nonumber \\
& =: \bs \kappa^\eps[\bs u](\bs x). \label{kappa}
\end{align} 
Substituting the above transformations into equation \eqref{mathcale} we get
\begin{align} \mathcal E^\eps ( \widetilde{\bs m},\widetilde{\bs u} ) &=  \eps^2 \int_\Omega \Big[ \ \frac{d}{\eps^2} \big| {\nabla_p \bs m} \big|^2 + d \, \big|\partial_3 \bs m\big|^2 +  \varphi(\bs m ) - \bs h_a \cdot \bs m + \frac{\mathbb C}{2} \big[ \bs \kappa^\eps[\bs u] - \bs E_s( \bs m ) \big]^2 \, \Big] \bs d \bs x \nonumber \\
& \qquad + \frac{\eps^2}{8 \pi}\int_{ \mathbb R^3 } \big| \bs h^\eps_{\bs m}(\bs x) \big|^2 \, \bs d \bs x. \nonumber
\end{align} 
Dividing above by $\eps^2$ and defining a new energy $ \mathcal I ^\eps ( \bs m, \bs u ) := \eps^{-2} \mathcal E ( \widetilde{\bs m},\widetilde{\bs u} ) $ we get, 
\begin{align} 
\mathcal I ^\eps ( \bs m, \bs u ) = \int_\Omega & \Big\{ \frac{d}{\eps^2} \big| {\nabla_p \bs m} \big|^2 + d \ \big|\partial_3 \bs m\big|^2 +  \varphi(\bs m ) - \bs h_a \cdot \bs m + \frac{\mathbb C}{2} \big[ \bs \kappa^\eps[\bs u] - \bs E_s( \bs m ) \big]^2 \, \Big\} \bs d \bs x \nonumber \\
& +\mathcal E^\eps_d(\bs m),
\nonumber \end{align}
where $\mathcal E^\eps_d(\bs m)$ is defined and bounded by rescaling the standard demag bound in equation  \eqref{demagstandardbound}
\begin{align}
& 0 \ \le \ \mathcal E^\eps_d(\bs m) := \frac{1}{8 \pi} \int_{ \mathbb R^3 } \big| \bs h^\eps_{\bs m}(\bs x) \big|^2 \bs d \bs x \ \le \ \frac{1}{2} \int_{\Omega} |\bs m(\bs x) |^2 \bs d \bs x  = \frac{1}{2} \big|\Omega\big| m_s^2. \label{demagstandardbound1}
\end{align} 
We investigate the asymptotic nature of the problem
\begin{align} ( \mathcal{P}^\eps ) \ \ \inf_{ \mathcal A_\eps} \mathcal I ^\eps ( \bs m, \bs u ), \quad \mathcal A_\eps=\Big\{ \ ( \bs m, \bs u ) \in H^1(\Omega,m_s S^2) \times H_{\sharp}^1(\Omega,\mathbb{R}^3) \ \Big\}  \end{align}
where $ H_{\sharp}^1(\Omega,\mathbb{R}^3) = \Big\{ \bs u(\bs x) \in H^1(\Omega,\mathbb{R}^3) \ | \ \bs u(x_1,x_2,0) = \bs 0, \ \ \forall (x_1,x_2) \in \omega \Big\} $ enforces Dirichlet boundary conditions at the base $x_3=0$. For the subsequent sections we also use the notation $H_{\sharp}^1 \big((0,1),\mathbb{R}\big)$ defined as
\begin{align}
H_{\sharp}^1 \big((0,1),\mathbb{R}\big) = \Big\{ \ w(x_3) \in H^1((0,1),\mathbb{R}^3) \ | \ w(x_3=0) = 0 \ \Big\}.  \end{align}

In the next section, we will start with a sequence of minimizers $(\bs m^\eps,\bs u^\eps)$ of $ \mathcal I ^\eps ( \bs m, \bs u )$ and show that we can extract a subsequence whose limit relates to the minimizers of a simpler lower dimensional problem $\mathcal{I}^o$.

{ \centering
\section{First variational limiting problem}
\label{Secfirstgammalimit} }
Let $( \widehat{ \ \ \cdot \ \ } )$ denote the cross-sectional average of any scalar/vector, i.e. for any field $\bs a(\bs x)$ set
\begin{align} \widehat{\bs a \,}(x_3) =\dashint_{\omega}\bs a(\bs x_p,x_3) \bs d \bs x_p. \label{csavg} \end{align}

For $ \eps $ fixed, let $ (\bs  m^\eps, \bs u^\eps)$ be a minimizer of $ \mathcal I ^\eps ( \bs m, \bs u )$. We look at the behavior of $ \mathcal I ^\eps(\bs m^\eps, \bs u^\eps)$ as $ \eps \rightarrow 0 $. For that we will first show that $\mathcal I ^\eps(\bs m^\eps, \bs u^\eps)$ is bounded above and below independent of $\eps$. Then we will show that from the sequence $(\bs m^\eps,\bs u^\eps)$, we can extract a subsequence(unrelabeled) such that $(\bs m^\eps,\widehat{u_3}^\eps) $ on the subsequence converges weakly to some $(\bs m^o,v^o)$ in an appropriate space. This convergence will be improved to strong and the limit $(\bs m^o,v^o) $ will be shown to minimize a new functional $\mathcal I^o$ in Theorem \ref{theorem2.1.1} . 
\subsection{Boundedness of \texorpdfstring{$ \mathcal I ^\eps(\bs m^\eps,\bs u^\eps) $}{boundedness}}
For an upper bound on $ \mathcal I ^\eps ( \bs m^\eps, \bs u^\eps ) $ we compare its energy with a test function $ (\bs m,\bs 0)$ with $ \bs m $ any constant vector on $ m_s S^2 $ and $ \bs u = 0 $ to get,
\begin{align} \mathcal I ^\eps ( \bs m^\eps, \bs u^\eps ) \le \mathcal I ^\eps ( \bs m, \bs 0 ) &=\int_{\Omega} \Big[ \phi(\bs m) + \frac{1}{2} \mathbb C\big[\bs E_s(\bs m)\big]^2 -\bs h_a \cdot \bs m \Big] \, \bs d \bs x + \mathcal E^\eps_d(\bs m)\nonumber \\
& \le  K_4 + \frac{m_s^2}{2} \, \big| \, \Omega \, \big|, \label{sequenceboundedenergy}
\end{align} 
where the anisotropy, elastic, Zeeman and magnetostatic terms are bounded using equations \eqref{Anisotroybound}, \eqref{elasticbound1}, \eqref{elasticbound2}, \eqref{Zeemanbound} and \eqref{demagstandardbound1}.
The positivity of all terms in $\mathcal I ^\eps ( \bs m^\eps, \bs u^\eps )$ except possibly of the Zeeman energy along with
equation \eqref{Zeemanbound} gives the lower bound,
\begin{align}
\mathcal I ^\eps ( \bs m^\eps, \bs u^\eps ) &\ge -\int_\Omega \bs h_a \cdot \bs m \, \bs d \bs x  \ge -K_2. \nonumber \end{align}
\subsection{Weak compactness of minimizers \texorpdfstring{$(\bs m^\eps,\bs u^\eps)$ as $ \eps \rightarrow 0 $}{weak convergence} }
The upper and lower bound on $ \mathcal I ^\eps ( \bs m^\eps, \bs u^\eps ) $ gives, 
\begin{align} K_5 & \ > \ \int_{\Omega} \Big\{ \ \frac{d}{\eps^2} \big| {\nabla_p \bs m^\eps} \big|^2 + d \ \big| \ \partial_3 { \bs m^\eps} \big|^2 \ \Big\} \bs d \bs x \ \ge \ d \int_{\Omega} | \nabla \bs m^\eps |^2 \bs d \bs x \label{ge1}.
\end{align}
Then for some unrelabeled subsequence 
\begin{align}
\big\Vert \nabla_p \bs m^\eps (\bs x) \big\Vert^2_{L^2(\Omega)} \le \frac{K_5}{d} \eps^2, & & \bs m^\eps \rightarrow \bs m^o \mbox{ in } L^2(\Omega), & & \nabla \bs m^\eps \rightharpoonup \nabla\bs m^o \mbox{ in } L^2(\Omega).
\label{magconvergence}\end{align}
By the weak convergence of $ \nabla \bs m^\eps (\bs x) $ to $ \nabla \bs m^o (\bs x) $ and the lower semi-continuity of norm operator $ \big\Vert ( \ \cdot \ ) \big\Vert_{L^2(\Omega)}$ w.r.t. weak convergence we have using equation \eqref{magconvergence}
\begin{align} & \big\Vert \nabla_p \bs m^o (\bs x) \big\Vert_{L^2(\Omega)} \le \ \liminf_{\eps \rightarrow 0} \ \big\Vert \nabla_p \bs m^\eps (\bs x) \big\Vert_{L^2(\Omega)} \le \ \liminf_{\eps \rightarrow 0} \sqrt{\frac{K_5}{d}} \eps = 0, \nonumber \end{align}
which implies
\begin{align}
& \lim_{\eps \rightarrow 0} \ \bs m^\eps(\bs x)=\bs m^o(\bs x) = \bs m^o(x_3) \qquad \mbox{ in } L^2(\Omega). \label{magconvergence1} \end{align}
Strong convergence of $ \bs m^\eps $ to $ \bs m^o $ in $L^2(\Omega) $ gives convergence pointwise $a.e.$ for a (unrelabelled) subsequence. The cross-sectional average of this subsequence $\widehat{\bs m \,}^\eps(x_3)=\mbox{\fontsize{8}{8}\selectfont $\displaystyle{ \dashint_{\omega}} $ } \bs m^\eps(\bs x) \bs d \bs x_p \ $ then converges pointwise $a.e.$ to $ \mbox{\fontsize{8}{8}\selectfont $\displaystyle{ \dashint_{\omega}} $ } \bs m^o(x_3) \bs d \bs x_p = \bs m^o(x_3)$. Since from Jensen's inequality we have $\ |\widehat{\bs m \,}^\eps(x_3)| \, \le \, |\bs m^\eps| =m_s $, the pointwise a.e. convergence of the unrelabeled subsequence $\widehat{\bs m \,}^\eps(x_3)$ to $\bs m^o(x_3)$ gives on using $L^p$ Dominated convergence theorem
\begin{align}
& \widehat{\bs m \,}^\eps(x_3) \rightarrow \bs m^o(x_3) \qquad \mbox{ in } L^p(0,1) \mbox{ as } \eps \rightarrow 0, \quad \  \forall 1 \le p \le \infty. \label{magconvergence2} \end{align}
Also the strong convergence of $\bs m^\eps$ to $\bs m^o$ in $L^2(\Omega)$ gives convergence of $|m_i^\eps|^2$ to $|m_i^o|^2$ in $L^2(\Omega)$ because of the fact that $\big| \, | m_i^\eps|^2 - |m_i^o|^2 \, \big|^2=|m_i^\eps-m_i^o|^2 \ | m_i^\eps+m_i^o|^2$ and domination of $|m_i^\eps|^2$ and $| m_i^o|^2$ by $m_s^2$ and $i \in \big\{ 1,2,3 \big\}$ denoting any of the 3 components of $\bs m^\eps$. Then using the same argument as above we can derive an unrelabeled subsequence such that 
\begin{align}
& \widehat{|m_i^\eps|^2}(x_3) \rightarrow \widehat{ |m_i^o|^2}(x_3)=|m_i^o|^2(x_3) \qquad \mbox{ in } L^2(0,1) \mbox{ as } \eps \rightarrow 0. \label{magconvergence3} \end{align} 
We will need equation \eqref{magconvergence3} for showing convergence of the elastic energy in Theorem \ref{theorem2.1.1}.
Next we prove a proposition which we need for extracting weak compactness of the elastic terms.
\newtheorem{prop}{Proposition}[section]
\begin{prop}
\label{proposition2.1.1}
Given $ \widehat{\bs m \,}^\eps \in H^1(\Omega)$ and $ \widehat{\bs u \,}^\eps \in H_\sharp^1(\Omega)$ using eqn. \eqref{csavg}, we have the following,
\begin{align}
& \big\Vert \widehat{\bs m \,}^\eps \big\Vert^2_{L^2(\Omega)} =|\omega| \ \big\Vert \widehat{\bs m \,}^\eps \big\Vert^2_{L^2(0,1)} \le \big\Vert \bs m^\eps \big\Vert^2_{L^2(\Omega)} , \nonumber \\
& \big\Vert \partial_3 \widehat{\bs m \,}^\eps \big\Vert^2_{L^2(\Omega)} =|\omega| \ \big\Vert \partial_3 \widehat{\bs m \,}^\eps \big\Vert^2_{L^2(0,1)} \le \big\Vert \partial_3 \bs m^\eps \big\Vert^2_{L^2(\Omega)}, \nonumber \\
& \big\Vert \partial_3 \widehat{\bs u \,}^\eps \big\Vert^2_{L^2(\Omega)} =|\omega| \ \big\Vert \partial_3 \widehat{\bs u \,}^\eps \big\Vert^2_{L^2(0,1)} \le \big\Vert \partial_3 \bs u^\eps \big\Vert^2_{L^2(\Omega)} , \nonumber 
\end{align}
and for $ i=\{ 1,2,3 \}$
\begin{align}
\big\Vert \widehat{u_i}^\eps \big\Vert^2_{L^2(\Omega)} \le K_6 \big\Vert \partial_3 \widehat{u_i}^\eps \big\Vert^2_{L^2(\Omega)} \le K_6 \big\Vert \partial_3 u^\eps_i \big\Vert^2_{L^2(\Omega)}.  \nonumber
\end{align}
\end{prop}
\begin{proof} The first result is easily seen using Jensen's inequality
\begin{align} \int^1_0 \big| \widehat{\bs m \,}^\eps \big|^2 dx_3 &= \int^1_0 \Big| \, \frac{1}{|\omega|} \int_{\omega} \bs m^\eps \bs d \bs x_p \Big|^2 dx_3 \le \int^1_0 \frac{1}{|\omega|} \int_{\omega} | \bs m^\eps |^2  \bs d \bs x_p dx_3 = \frac{1}{|\omega|} \int_{\Omega} \big| \bs m^\eps \big|^2 \, \bs d \bs x. \nonumber
\end{align}
To see the second result, note for $i \in \{ 1,2,3 \}$ using Jensen's inequality
\begin{align} \int^1_0 \big| \partial_3 \widehat{m_i}^{\eps} \big|^2 dx_3 &= \int^1_0 \Big| \ \partial_3 \, \Big\{ \frac{1}{|\omega|} \int_{\omega} m_i^{\eps} \, \bs d \bs x_p \Big\} \ \Big|^2 dx_3 =\int^1_0 \frac{1}{|\omega|^2} \ \, \Bigg| \int_{\omega} \partial_3 m_i^{\eps} \, \bs d \bs x_p \Bigg|^2 dx_3 \nonumber \\
& \le \int^1_0 \frac{1}{|\omega|} \int_{\omega} \big| \partial_3 m_i^{\eps} \big|^2 \ \bs d \bs x_p \ dx_3 = \frac{1}{|\omega|} \int_{\Omega} \big| \partial_3 m_i^{\eps} \big|^2 \, \bs d \bs x. \nonumber
\end{align}
Integrating over $ \omega $ and summing over $ i $ gives us the first result. Similar calculation with $ \bs u^\eps $ replacing $ \bs m^\eps $ gives the third result. Noting the Dirichlet Boundary conditions on $ \widehat{\bs u \,}^\eps(x_3) $ at $ x_3=0 $ we get using 1-D Poincar\'e inequality over $ (0,1) $, 
\begin{align} \int^1_0 |\widehat{u_i}^\eps|^2 dx_3 \ \le \ K_6 \int^1_0 |\partial_3 \widehat{u_i}^{\eps} |^2 dx_3 \ \le \ K_6\int_{\Omega} \frac{1}{|\omega|} \ \big| \partial_3 u_i^{\eps} \big|^2 \bs d \bs x \nonumber
\end{align}
where $ K_6 $ is the Poincar\'e constant on $(0,1)$. Integrating over $ \omega $ gives the result. 
\end{proof}
Using positive definiteness of $\mathbb C$ in \eqref{elasticbound2}, \eqref{elasticbound1} and bounds on $\mathcal I^\eps(\bs m^\eps,\bs u^\eps)$, we get
\begin{align} \hspace{-5mm} \int_{\Omega} \big|\bs \kappa^{\eps}[\bs u^\eps]\big|^2 &= \int_{\Omega} \big|\bs \kappa^{\eps}[\bs u^\eps] - \bs E_s( \bs m^\eps ) + \bs E_s( \bs m^\eps )\big|^2 \le 2 \int_{\Omega} \Big\{ \big|\bs \kappa^{\eps}[\bs u^\eps] - \bs E_s( \bs m^\eps )\big|^2 + \big|\bs E_s( \bs m^\eps )\big|^2 \Big\} \nonumber \\
& \le \frac{2}{\gamma} \int_{\Omega} \mathbb C  \big[\bs \kappa^{\eps}[\bs u^\eps] - \bs E_s( \bs m^\eps )  \big]^2 + 2 \int_{\Omega} \big|\bs E_s( \bs m^\eps )\big|^2 \ \le \ K_7. \nonumber
\end{align}
Combining this with the fourth result in Proposition \ref{proposition2.1.1} we have
\begin{align}  \int_{\Omega} \big|\widehat{u_3}^\eps\big|^2 dx_3 \, \le \, K_6 \int_{\Omega} \big|\partial_3 \widehat{u_3}^\eps\big|^2 dx_3 \, \le \, K_6 \int_{\Omega} \big|\partial_3 u^\eps_3 \big|^2 \bs d \bs x \, \le \, K_6 \int_{\Omega} \big| \bs \kappa^{\eps}[\bs u^\eps]\big|^2 \bs d \bs x \, < \, K_8. \nonumber
\end{align}
Thus $ \big\Vert \widehat{u_3}^\eps \big\Vert_{H^1(0,1)} \le \infty $ and due to Dirichlet conditions on $ \bs u^\eps $ we get $ \widehat{u_3}^\eps \in H_{\sharp}^1(0,1) $. For an unrelabeled subsequence we have,
\begin{align}
& \widehat{u_3}^\eps(x_3) \rightarrow v^o(x_3) \mbox{ in } L^2(0,1), & \partial_3 \widehat{u_3}^\eps(x_3) \rightharpoonup \partial_3 v^o(x_3) \mbox{ in } L^2(0,1) \label{elconvergence} \end{align}
Already from the fact that $ \bs m^o $ and $ v^o $ depends only on the $ x_3 $ space variable the 1-D nature of the limit problem becomes evident. 
The magnetostatic estimate in equation \eqref{gammafirstdemag1} gives 
\begin{align}
\mathcal E^\eps_d(\bs m^\eps)- \pi |\omega| \int^1_0 \big|\widehat{\bs m \,}^\eps_p(y_3) \big|^2 dy_3=O(\eps) +O(\eps^{3/4}) \nonumber
\end{align}
which implies on using strong convergence of $\widehat{\bs m \,}^\eps$ to $\bs m^o$ in $L^2(0,1)$ in equation \eqref{magconvergence2},
\begin{align} \lim_{\eps \rightarrow 0} \ \mathcal E^\eps_d(\bs m^\eps) =\lim_{\eps \rightarrow 0} \ \pi |\omega| \int^1_0 \big| \widehat{\bs m \,}_p^\eps(x_3) \big|^2 d x_3 =\pi |\omega| \int^1_0 \big| \bs m^o_p(x_3) \big|^2 d x_3. \
\end{align}  
The magnetostatic estimate in equation \eqref{gammafirstdemag1} and Remark \ref{RemA3} also gives 
\begin{align} \mathcal E^\eps_d(\bs m^o) = \pi |\omega| \int^1_0 \big| \bs m_p^o(x_3) \big|^2 d x_3 +O(\eps) +O(\eps^{3/4}) =\pi |\omega| \int^1_0 \big| \bs m^o_p(x_3) \big|^2 d x_3  \label{moedemag}
\end{align}
and thus,
\begin{align}
\lim_{\eps \rightarrow 0} \ \mathcal E^\eps_d(\bs m^\eps) = \lim_{\eps \rightarrow 0} \ \mathcal E^\eps_d(\bs m^o)=\pi |\omega| \int^1_0 \big| \bs m^o_p(x_3) \big|^2 d x_3. \label{edepsconvergence}
\end{align}
\subsection{Strong compactness of \texorpdfstring{ $(\bs m^\eps,\bs u^\eps)$}{Strong convergence} and variational problem}
\begin{align}
\hspace{-8mm} \mbox{Set \quad } f_0(s) \ := \ \min \Big[ \mathbb C [\bs E]: \bs E \ ; \ \bs E \in \bs M_{sym}^{3 \times 3} \mbox{ and } E_{33}=s \Big].  \label{f0definition}
\end{align}
Note that $ f_0 $ defined above in \eqref{f0definition} can be evaluated as
\begin{align}  f_0(s) &= c_{11}|s|^2 - 2\sigma c_{12} |s|^2 := Y |s|^2, &  \int^1_0 f_0(s(x_3)) dx_3=Y \big\Vert s(x_3)\big\Vert^2_{L^2(0,1)} \label{youngs-modulus} \end{align}
where $\displaystyle{ \sigma = \Big( \frac{c_{12}}{c_{11}+c_{12}} \Big) } \ $ is the Poisson's ratio and $ Y =\big( \, c_{11}- 2\sigma c_{12} \, \big)$ is the Young's modulus.
We now state the main result of this section.
\newtheorem{thm1}{Theorem}[section]
\begin{thm1}
\label{theorem2.1.1}
There exists a subsequence $(\bs m^\eps,\bs u^\eps)$ not relabeled such that $ \bs m^\eps \rightarrow \bs m^o$ strongly in $ H^1(\Omega,\mathbb R^3)$, $\widehat{u_3}^\eps \rightarrow v^o \mbox{ strongly in } H_\sharp^1((0,1),\mathbb R) $ and $(\bs m^o,v^o)$ minimizes $\mathcal{I}^o(\bs m,v) $ in $ \mathcal A_o =\Big\{ (\bs m(x_3),$ $ v(x_3)) \in H^1((0,1),m_sS^2) \times H^1_{\sharp}((0,1),\mathbb R) \Big\} $ where $ \mathcal{I}^o(\bs m,v) $ is defined as 
\begin{align} \mathcal{I}^o(\bs m,v) &= \int^1_0 d \big|\partial_3 \bs m \big|^2 + \varphi(\bs m )+\pi| \bs m_p |^2 + \frac{1}{2} f_0 \big( \, \partial_3 v- E_{s_{33}} (\bs m) \, \big) - \bs h_a \cdot \bs m.  \label{firstsgamma} \end{align} 
\end{thm1}
\begin{proof}
Comparing energy of $\mathcal I^\eps $ at its minimizer $(\bs m^\eps,\bs u^\eps) $ with the test function $(\bs m^o,\bs u^\eps) $ we get $\mathcal I^\eps(\bs m^\eps,\bs u^\eps) \le \mathcal I^\eps(\bs m^o,\bs u^\eps)$ which expands out as  
\begin{align} 
& \int_\Omega \Big[ \, \frac{d}{\eps^2} \big| {\nabla_p \bs m^\eps} \big|^2 + d \, \big| \partial_3 \bs m^\eps \big|^2 +  \varphi(\bs m^\eps ) - \bs h_a \cdot \bs m^\eps + \frac{\mathbb C}{2} \big[ \bs \kappa^\eps[\bs u^\eps] - \bs E_s( \bs m^\eps ) \big]^2 \, \Big] \bs d \bs x + \mathcal E_d^\eps(\bs m^\eps) \nonumber \\
&  \le \int_\Omega \Big[  d \, \big| \, \partial_3 \bs m^o\big|^2 + \varphi(\bs m^o ) - \bs h_a \cdot \bs m^o +\frac{\mathbb C}{2} \big[ \bs \kappa^\eps[\bs u^\eps] - \bs E_s( \bs m^o ) \big]^2 \Big] + \mathcal E_d^\eps(\bs m^o). \nonumber \end{align}  
Equation \eqref{edepsconvergence} gives that $\lim_{\eps \rightarrow 0} \ \mathcal E^\eps_d(\bs m^\eps) = \lim_{\eps \rightarrow 0} \ \mathcal E^\eps_d(\bs m^o)$. Then taking lim-sup of both sides w.r.t. $\eps$, canceling common terms, and noting that $ \bs m^\eps(\bs x) \rightarrow \bs m^o(x_3) $ strongly in $L^2(\Omega)$, we can simplify the above equation to get
\begin{align} 
 \limsup_{\eps \rightarrow 0} \int_\Omega & \Big[ \, \frac{d}{\eps^2} \big| {\nabla_p \bs m^\eps} \big|^2 + d \, \big| \, \partial_3 \bs m^\eps\big|^2 \, \Big]\, \bs d \bs x \ \le \ \int_\Omega  d \, \big| \, \partial_3 \bs m^o \big|^2  \bs d \bs x. \nonumber \end{align} 
But weak convergence of $\nabla \bs m^\eps $ in equation \eqref{magconvergence} implies $ \liminf_{\eps \rightarrow 0} \mbox{\fontsize{9}{8}\selectfont $\displaystyle{ \int_\Omega}$} \big| \, \partial_3 \bs m^o\big|^2 \le \mbox{\fontsize{9}{8}\selectfont $\displaystyle{ \int_\Omega}$} \big|\, \partial_3 \bs m^\eps\big|^2 $ which combined with the $\lim\sup $ condition above gives the strong convergence,
\begin{align}
& \partial_3 \bs m^\eps \rightarrow \partial_3 \bs m^o \mbox{ in } L^2(\Omega),
&\frac{1}{\eps} {\nabla_p \bs m^\eps} \rightarrow \bs 0 \mbox{ in } L^2(\Omega).  \label{magstrongconvergence}
\end{align}    
Now we show strong convergence of the elastic terms. Set $ s^\eps(\bs x) $ and $ \widehat{s}^{\ \eps}(x_3) $ as
\begin{align} \hspace{-4mm} & s^\eps(\bs x):= \partial_3 u_3^\eps(\bs x) - E_{s_{33}}(\bs m^\eps),& & & \widehat{s}^{\ \eps}(x_3) =  [\widehat{\partial_3 u_3^\eps - E_{s_{33}}(\bs m^\eps)}](x_3)=\widehat{\partial_3 u_3^\eps}(x_3) - \widehat{E_{s_{33}}(\bs m^\eps)}
\nonumber \end{align}
where $ \widehat{s}^{\ \eps}(x_3) $ is defined using equation \eqref{csavg}. Noting $f_0(s)=Y|s|^2$, using Jensen's inequality
\begin{align}
\hspace{-7mm} \int_{\Omega} f_0(\partial_3 \widehat{u_3}^\eps-\widehat{E_{s_{33}}(\bs m^\eps)}) \bs d \bs x =\int^1_0\int_{\omega} Y |\widehat{s}^{\ \eps}|^2 \bs d \bs x &\le \int^1_0 Y \Big[ \int_{\omega} |s^\eps|^2 \bs d \bs x_p \Big] d x_3 = \int_{\Omega} f_0(s^\eps(\bs x))\bs d \bs x \nonumber \\ 
&\le \int_{\Omega} \mathbb C \big[ \bs \kappa^\eps[\bs u^\eps] - \bs E_s( \bs m^\eps ) \big]^2 \bs d \bs x \label{Jensen2} \end{align}
where in the last step we have used the definition of $ f_0$ from equation \eqref{f0definition}. \\
By definition $E_{s_{33}}(\bs m^\eps)=\frac{3}{2m^2_s} \big( |m^\eps_3|^2 - \frac{1}{3}\big) $ which using $\widehat{|m_3^\eps|^2} \rightarrow |m_3^o|^2 $ in $ L^2(0,1) $ from equation \eqref{magconvergence3} gives 
\begin{align}
\widehat{ E_{s_{33}}(\bs m^\eps)} \rightarrow \bs E_{s_{33}}(\bs m^o) \quad \mbox{ in  } L^2(0,1). \nonumber
\end{align}
The above combined with weak convergence $\partial_3 \widehat{u_3}^\eps \rightharpoonup \partial_3 v^o $ in $ L^2(0,1) $ in eqn. \eqref{elconvergence}  gives
\begin{align}
\partial_3 \widehat{u_3}^\eps-\widehat{E_{s_{33}}(\bs m^\eps)} \rightharpoonup \partial_3 v^o-E_{s_{33}}(\bs m^o) \mbox{ in $ L^2(0,1) $ }. \nonumber
\end{align}
Then noting from eqn. \eqref{youngs-modulus} that $\mbox{\fontsize{9}{8}\selectfont $\displaystyle{ \int_{\Omega} } $ } f_0(s(x_3)) \bs d \bs x=Y |\omega| \ \big\Vert s(x_3) \big\Vert^2_{L^2(0,1)} $ and weak lower semi-continuity of norm in $ L^2(0,1) $ gives
\begin{align}
\int_{\Omega} f_0\big( \partial_3 v^o-E_{s_{33}}(\bs m^o) \big) \bs d \bs x & \le \liminf_{\eps \rightarrow 0} \int_{\Omega} f_0(\partial_3 \widehat{u_3}^\eps-\widehat{E_{s_{33}}(\bs m^\eps)}) \bs d \bs x \nonumber \\ 
&\le \liminf_{\eps \rightarrow 0} \int_{\Omega} \mathbb C \big[ \bs \kappa^\eps[\bs u^\eps] - \bs E_s( \bs m^\eps ) \big]^2 \bs d \bs x \label{Jensen22} \end{align}
where in the last step we use eqn. \eqref{Jensen2}. 

To get strong convergence we will show the converse inequality of equation \eqref{Jensen22}. For that we need to compare the energy $ \mathcal I^\eps(\bs m^\eps,\bs u^\eps)$ with some test function based on $ \bs m^o$ and $ v^o $. But the lack of regularity of $ v^o \in H^1(0,1) $ requires a mollification procedure. Let  $ v^h(x_3) \in \mathcal D(0,1) $ and $ v^h(x_3) \rightarrow v^o(x_3)$ in $ H^1(0,1)$ as $ h \rightarrow 0$. Set $ s^h(x_3):=\big(\partial_3 v^h(x_3)-E_{s_{33}}(\bs m^o)\big)$. Note $ \lim_{h \rightarrow 0} s^h = \big(\partial_3 v^o(x_3)-E_{s_{33}}(\bs m^o)\big)$. Define   
\begin{align} \bs v_h^\eps(\bs x) := \begin{bmatrix} \eps E_{s_{11}} (\bs m^o)x_1 + \eps E_{s_{12}}(\bs m^o) x_2 - \eps \sigma \ s^h(x_3) x_1 \\ \eps E_{s_{22}}(\bs m^o) x_2 + \eps E_{s_{12}}(\bs m^o) x_1- \eps \sigma \ s^h(x_3) x_2  \\ v^h(x_3) + 2 \eps \big( E_{s_{13}}(\bs m)x_1 + 2E_{s_{23}}(\bs m)x_2 \big) \end{bmatrix}. \label{bmvepsdef} 
\end{align}
For $ \bs v_h^\eps $ defined above, $\bs \kappa^\eps[\bs v_h^\eps] -\bs E_s(\bs m^o) $ is given by
\begin{align}
\bs \kappa^\eps[\bs v_h^\eps] -\bs E_s(\bs m^o) =\begin{bmatrix} -\sigma s^h & 0 & -\frac{\eps}{2} \big( \sigma x_1 \ \partial_3 s^h +x_1 \ \partial_3 E_{s_{11}}(\bs m^o) + x_2 \ \partial_3 E_{s_{12}}(\bs m^o) \big) \\ \cdot & -\sigma s^h & -\frac{\eps}{2} \big( \sigma x_2 \ \partial_3 s^h +x_2 \ \partial_3 E_{s_{22}}(\bs m^o) +x_1 \ \partial_3 E_{s_{12}}(\bs m^o) \big) \\ \cdot & \cdot  & s^h+2\eps\big( x_1 \ \partial_3 E_{s_{13}}(\bs m^o) +x_2 \ \partial_3 E_{s_{23}}(\bs m^o) \big)  \nonumber \end{bmatrix},
\end{align}
where we have left out terms below the diagonal due to symmetry. A straight forward computation gives, $\big($ Recall $Y=c_{11}-\sigma c_{12} $ from eqn. \eqref{youngs-modulus} $\big)$ 
\begin{align}
\int_{\Omega} \mathbb C \big[ \bs \kappa^\eps[\bs v_h^\eps] -\bs E_s(\bs m^o) \big]^2 \bs d \bs x &= \int_{\Omega} f_0(s^h) \ \bs d \bs x + O(\eps^2). \label{Calpha}
\end{align}
Then comparing energy of the test function $(\bs m^o,\bs v_h^\eps)$ with $ (\bs m^\eps,\bs u^\eps) $ gives 
\begin{align} 
\mathcal I^\eps(\bs m^\eps,\bs u^\eps) \ \le \ \mathcal I^\eps(\bs m^o,\bs v_h^\eps). \nonumber \end{align}
Fixing $ h $ and taking lim-sup of both sides w.r.t. $\eps$, 
using strong convergence in \eqref{magstrongconvergence}, and equation \eqref{Calpha}, the above simplifies to 
\begin{align} \limsup_{\eps \rightarrow 0} \int_\Omega \frac{\mathbb C}{2} \big[ \bs \kappa^\eps[\bs u^\eps] - \bs E_s( \bs m^\eps ) \big]^2 \bs d \bs x & \ \le \ \limsup_{\eps \rightarrow 0} \int_\Omega \frac{\mathbb C}{2} \big[ \bs \kappa^\eps(\bs v_h^\eps) - \bs E_s( \bs m^o ) \big]^2 \, \bs d \bs x \nonumber \\
& \ = 
 \int_\Omega \frac{1}{2} \, f_0( s^h )\, \bs d \bs x. \nonumber 
\end{align} 
Now taking $\lim_{h \rightarrow 0}$ of L.H.S. and noting that $ \lim_{h \rightarrow 0} s^h = (\partial_3 v^o(x_3)-E_{s_{33}}(\bs m^o))$ gives
\begin{align} \limsup_{\eps \rightarrow 0} \int_\Omega \frac{\mathbb C}{2} \big[ \bs \kappa^\eps[\bs u^\eps] - \bs E_s( \bs m^\eps ) \big]^2 \bs d \bs x & \ \le \ \int_\Omega \frac{1}{2} \, f_0 \big( \partial_3 v^o(x_3)-E_{s_{33}}(\bs m^o) \big)\, \bs d \bs x.\label{C-f0Vtest2}
\end{align} 
Then \eqref{Jensen22} and \eqref{C-f0Vtest2} together give along with eqn. \eqref{f0definition}, the strong convergence
\begin{subequations}\label{elstrongconvergence}
\begin{align}
& \lim_{\eps \rightarrow 0 } \ \partial_3 \widehat{u_3}^\eps \rightarrow \partial_3 v^o \mbox{ in } L^2(0,1) \label{elstrongconvergence1}, \\ 
& \lim_{\eps \rightarrow 0 } \int_\Omega \frac{\mathbb C}{2} \big[ \bs \kappa^\eps[\bs u^\eps] - \bs E_s( \bs m^\eps ) \big]^2 \, \bs d \bs x \rightarrow
\int_\Omega \frac{1}{2} f_0\big( \, \partial_3 v^o-E_{s_{33}} (\bs m^o) \, \big) \, \bs d \bs x. \label{elstrongconvergence2}
\end{align}
\end{subequations}
Finally its easy to see that the strong convergence from \eqref{magstrongconvergence} and \eqref{elstrongconvergence} together gives,
\begin{align}
\lim_{\eps \rightarrow 0} \, \mathcal I^\eps(\bs m^\eps,\bs u^\eps) = |\omega|\mathcal I^o(\bs m^o,v^o).\nonumber
\end{align} 
Now we show that $\big( \bs m^o,v^o \big)$ minimizes $\mathcal I^o(\bs m,v)$ in $\mathcal A_o $. For any $\big( \bs m,v \big) $ with $\bs m \in H^1((0,1),m_sS^2)$ and $v \in C^\infty(0,1) \cap \{ v(0)=0 \}$, let us define as in eqn. \eqref{bmvepsdef} 
\begin{align} \bs V^\eps(\bs x) := \begin{bmatrix} \eps E_{s_{11}} (\bs m)x_1 + \eps E_{s_{12}}(\bs m) x_2 - \eps \sigma \ \big( \partial_3 v -E_{s_{33}}(\bs m)\big) x_1 \\ \eps E_{s_{22}}(\bs m) x_2 + \eps E_{s_{12}}(\bs m) x_1- \eps \sigma \ \big( \partial_3 v -E_{s_{33}}(\bs m)\big) x_2  \\ v(x_3) + 2 \eps \big( E_{s_{13}}(\bs m)x_1 + 2E_{s_{23}}(\bs m)x_2 \big) \end{bmatrix}. \nonumber
\end{align}
Then comparing energy of $\mathcal I^\eps(\bs m^\eps,\bs u^\eps)$ with $\mathcal I^\eps(\bs m,\bs V^\eps)$ we get $\mathcal I^\eps(\bs m^\eps,\bs u^\eps) \le \mathcal I^\eps(\bs m,\bs V^\eps)$. Taking $\lim_{\eps}$ of the inequality we note that L.H.S converges from above $\lim_{\eps \rightarrow 0} \, \mathcal I^\eps(\bs m^\eps,\bs u^\eps) = |\omega|\mathcal I^o(\bs m^o,v^o)$. Using equations \eqref{moedemag} and \eqref{Calpha} with $\bs m$ and $\bs V^\eps$ replacing $\bs m^o$ and $\bs v^\eps_h$ in the respective equations it is easy to check that the R.H.S converges as $\lim_{\eps \rightarrow 0} \, \mathcal I^\eps(\bs m,\bs V^\eps) = |\omega|\mathcal I^o(\bs m,v)$. We thus get our minimizing principle on noticing that $v \in C^\infty(0,1) \cap \{ v(0)=0 \}$ is dense in $H^1_\sharp(0,1)$
\end{proof}

\subsection{Minimization of limit problem} 
\label{Sec-example}
The minimization of $\mathcal I^o(\bs m,v)$ is a substantially simpler problem than the original one. One can see that if the applied field is a constant over the domain, the terms $ \varphi(\bs m )+\pi| \bs m_p |^2 - \bs h_a \cdot \bs m$ behaves like an ``effective anisotropy". If this is minimized over constant vector $ \bs m \in m_sS^2 $ to give $ \bs m^o $, then its easy to see that $ (\bs m^o,E_{s_{33}}(\bs m^o)x_3) $ minimizes $ \mathcal I^o(\bs m,v)$. 

For a large class of ferromagnetic materials, the largest energy in the ``effective anisotropy" for typical applied fields is the demagnetization term $\pi| \bs m_p |^2 $ which finds its minimum if $ \bs m^o $ is an axial magnetization $(0,0,m_s)$. In particular for our nanowires of Galfenol this is true. Experimentally produced nanowires of Galfenol of 30-100 nanometer diameter show strong alignment of magnetization along the axis in the absence of applied fields and need large applied fields in transverse direction to alter this state. Experimental verification of these results for Galfenol wires can be seen from Magnetic Force Microscopy (MFM) scans in Figures \ref{fig12} taken from \cite{downey2008thesis}. 
\begin{figure}
\centering
\mbox{ \subfigure{\includegraphics[scale=0.6]{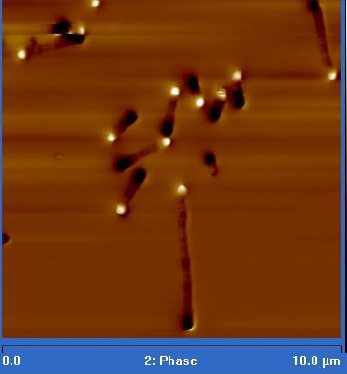} }\qquad \subfigure{
\includegraphics[scale=0.43]{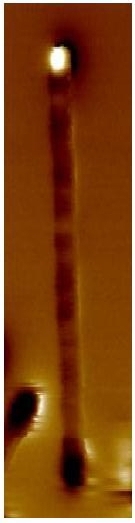} } }
\caption{Left: MFM scan for several Galfenol nanowires, Right: Magnified scan of single nanowire, scale of the wires shown in bottom of left figure (Scans courtesy of Downey \cite{downey2008thesis} ). }  \label{fig12}
\end{figure}

These scans are done for wires with 100 nanometer diameter and $< \hspace{-1mm}110\hspace{-1mm}>$ crystallographic orientation with no applied field. For cubic anisotropy, $< \hspace{-1mm}110\hspace{-1mm}>$ is a local minimum of the anisotropy energy and gives zero magnetostatic energy contribution  making it a global minimum of the ``effective anisotropy". The uniformity of the scan along the wire length depicts a uniform state of magnetization and the bright and dark spots at the two ends are interpreted to be the field lines due to an axial magnetization producing net positive and negative poles at the ends.

With these observations in mind, for the following sections we will assume that the field $\bs h_a$ is constant. 
This assumption simplifies the calculation in the following sections without effecting the main presentation of the asymptotic limiting problem. Let us then set
\begin{align}
Q_0 = \int_{\Omega} \varphi(\bs m^o )+\pi| \bs m^o_p |^2 - \bs h_a \cdot \bs m^o
\end{align}
where $\bs m^o$ minimizes $ \varphi(\bs m )+\pi| \bs m_p |^2 - \bs h_a \cdot \bs m$ in $m_sS^2$. Then $(\bs m^o,v^o)$ minimizes $\mathcal I^o$ where $v^o:=E_{s_{33}}(\bs m^o)x_3$. Set
\begin{align} \bs u^o(\bs x) := \begin{bmatrix} \eps E_{s_{11}} (\bs m^o)x_1 + \eps E_{s_{12}}(\bs m^o) x_2 \\ \eps E_{s_{22}}(\bs m^o) x_2 + \eps E_{s_{12}}(\bs m^o) x_1 \\ E_{s_{33}}(\bs m^o)x_3 + + 2 \eps \big( x_1 E_{s_{13}}(\bs m^o) + x_2 E_{s_{23}}(\bs m^o) \big) \end{bmatrix} \nonumber 
\end{align}
where we have abused notation a little as $\bs u^o(\bs x)$ depends on $ \eps$ but does not reflect that. Then it is easy to check that $ \bs \kappa^\eps[\bs u^o] = \bs E_s( \bs m^o )$ since $\bs m^o $ is constant. Using eqn. \eqref{magnetostaticM} ,
\begin{subequations}  \label{C-f0Vtest}
\begin{align}
& \int_\Omega \frac{\mathbb C}{2} \big[ \bs \kappa^\eps[\bs u^o] - \bs E_s( \bs m^o ) \big]^2 \, \bs d \bs x =
\int_\Omega \frac{1}{2} f_0\big( \, \partial_3 v^o-E_{s_{33}} (\bs m^o) \, \big) \, \bs d \bs x =0  \label{C-f0Vtesta} \\
& \mathcal I^\eps(\bs m^o,\bs u^o)=Q_0+\mathcal E^\eps_d(\bs m^o)-\int_{\Omega} \pi| \bs m^o_p |^2=Q_0+\eps Q_1(\bs m^o)+\eps^2Q_2(\bs m^o)  \label{C-f0Vtestb} \\
& |\omega| \ \inf_{\mathcal A_o} \mathcal I^o(\bs m,v^o) = Q_0.  \label{C-f0Vtestc}
\end{align}\end{subequations}

{\centering
\section{Second order variational limit problem}
\label{Secsecondgammalimit} } 
\S \ \ref{Secfirstgammalimit} gives a rigorous derivation of the first order variational approximation  $ \mathcal I^o(\bs m,v) $ in the sense that for a sequence of minimizers $(\bs m^\eps,\bs u^\eps)$ of $ \mathcal I^\eps(\bs m,\bs u) $, $ \lim_{\eps \rightarrow 0} \mathcal I^\eps(\bs m^\eps,\bs u^\eps) = |\omega| \mathcal I^o(\bs m^o,v^o)+o(\eps)$ with $(\bs m^o,v^o)$ minimizing $\mathcal I^o(\bs m,v)$ in an appropriate space. Correctors to this approximation come up as higher order theories which involve an expansion of the $o(\eps)$ term. These higher terms can be understood as an asymptotic $\Gamma$- series of variational problems in the sense of \cite{anzellotti1993asymptotic}. 

With this in mind we define $ \mathcal{I}_1^{ \eps}(\bs m^\eps,\bs u^\eps) := \eps^{-1} ( \mathcal{I}^\eps(\bs m^\eps,\bs u^\eps) - Q_0) $. We look at the limit minimization problem corresponding to $ \mathcal{I}_1^{ \eps}(\bs m^\eps,\bs u^\eps) $. For this we first show that $ \mathcal{I}_1^{ \eps}(\bs m^\eps,\bs u^\eps) $ is bounded above and below independently of $ \eps $ so that its limit $ \eps \rightarrow 0 $ makes sense. We then show that a limit exists as $ \eps \rightarrow 0 $ for the quantity $ \mathcal{I}_1^{ \eps}(\bs m^\eps,\bs u^\eps) $. Note that
\begin{align}
\mathcal{I}^{\eps}(\bs m^\eps,\bs u^\eps) - Q_0 &= \Bigg[ \int_\Omega \frac{d}{\eps^2} | {\nabla_p \bs m^\eps} |^2 \bs d \bs x \Bigg] + \Big[ \mathcal E^\eps_d(\bs m^\eps)-\int_\Omega \pi |\bs m^\eps_p |^2 \bs d \bs x \Big] + \Bigg[ \int_\Omega \Big\{ \, d \, |\partial_3 {\bs m}^\eps|^2 \nonumber \\ 
& \hspace{6mm} + \varphi(\bs m^\eps) - \bs h_a \cdot \bs m^\eps + \pi| \bs m^\eps_p |^2 + \frac{\mathbb C}{2} \big[ \bs \kappa^{\eps}[\bs u^\eps] - \bs E_s(\bs m^\eps) \big]^2 \ \Big\} \bs d \bs x- Q_0  \Bigg]\nonumber \\ &= \mathfrak{A}^\eps+\mathfrak{B}^\eps+\mathfrak{C}^\eps \label{A+B+C}
\end{align}
where $ \mathfrak{A}^\eps,\mathfrak{B}^\eps$ and $\mathfrak{C}^\eps $ are the terms in the big square brackets. 
\subsection{Bounds for \texorpdfstring{$\mathcal{I}_1^{ \eps}(\bs m^\eps,\bs u^\eps)$}{I2bound}}
Since $(\bs m^\eps,\bs u^\eps)$ minimizes $ \mathcal{I}^\eps(\bs m,\bs u) $, we have $\mathcal{I}^{ \eps}( \bs m^\eps,\bs u^\eps) \lel \mathcal{I}^{ \eps}( \bs m^o,\bs u^o)$ which on using equation \eqref{C-f0Vtestb} along with the definition of $\mathcal{I}_1^{ \eps}( \bs m,\bs u) $ gives us the inequality 
\begin{align}
\mathcal{I}_1^{ \eps}( \bs m^\eps,\bs u^\eps) &\le \mathcal{I}_1^{ \eps}( \bs m^o,\bs u^o) =\frac{1}{\eps} \big(\eps Q_1(\bs m^o)+\eps^2 Q_2(\bs m^o) \big) \le \ K_{9}.\label{I2upper}
\end{align}
The lower bound for $ \mathcal{I}_1^{ \eps}(\bs m^\eps,\bs u^\eps) $ requires the following technical condition. See Result 8.2 in \cite{bhattacharya1999theory} and Definition 5.2 from \cite{le2005modeling} to see other contexts where such a condition is necessary to get lower bound estimates.
\newtheorem{mydef}{Definition}[section]
\begin{mydef}
We say that a minimizer $(\bs m^o,v^o)$ of $\mathcal I^o(\bs m,v) \ \big( \mbox{cf.} \,\mbox{Eqn.} \, \eqref{firstsgamma} \big)$,  satisfies the strong second variation condition if for any $(\bs m(x_3),v(x_3)) \in \mathcal A_o $ there exists $\Lambda > 0$ such that,
\begin{align} \mathcal I^o(\bs m,v) - \mathcal I^o(\bs m^o,v^o)=\mathcal I^o(\bs m,v) - Q_0 &\ge  \Lambda  \int^1_0 \Big\{ \ \big| \partial_3 \bs m(x_3) -\partial_3 \bs m^o (x_3)\big|^2 + \big| \bs m-\bs m^o\big|^2 \nonumber\\
& \qquad \qquad +\big| \partial_3 v (x_3) - \partial_3 v^o(x_3)\big|^2  \Big\} d x_3. \end{align} 
provided $ \big\Vert \bs m - \bs m^o \big\Vert_{H^1(0,1)} < K_{10} \eps$ and $ \big\Vert v - v^o \big\Vert_{H^1(0,1)} < K_{11} \eps$ for some $\eps >0 $ sufficiently small and $ K_{10},K_{11}$ arbitrary constants independent of $\eps$.
\end{mydef}
\begin{align} \hspace{-8mm} \mbox{Set } \hspace{8mm} \bs M^\eps := \bs m^\eps - \bs m^o. \label{defnMeps}
\end{align}
Using the hypothesis that $ \mathcal I^o(\bs m,v^o) $ satisfies strong second variation condition let us show the following Lemma,
\newtheorem{lem2}{Lemma}[section]
\begin{lem2}
\label{lemma2.2.1}
For $ \mathfrak{C}^\eps $ defined as in equation \eqref{A+B+C},
\begin{align} \mathfrak{C}^\eps \ \ge \ \int_{\omega} \mathcal I^o(\bs m^\eps,u_3^\eps) \bs d \bs x - Q_0 \ \ge \ \Lambda \Big( \  \big\Vert \partial_3 \bs M^\eps \big\Vert^2_{L^2(\Omega)} + \big\Vert \bs M^\eps \big\Vert^2_{L^2(\Omega)} \ \Big). \nonumber
\end{align}
\end{lem2}
\begin{proof}
For fixed $ \bs x_p =(x_1,x_2) \in \omega $ we define $ \bm{\mathfrak{M}}^\eps(x_3) := \bs m^\eps(\bs x_p,x_3) $, $ \bm{\mathfrak{V}}^\eps(x_3) := \bs u^\eps(\bs x_p,x_3) $. Then using the strong second variation condition on $ \mathcal I^o $ we get
\begin{align}
\mathcal I^o(\bm{\mathfrak{M}}^\eps(x_3),\mathfrak V_3^\eps(x_3))- \mathcal I^o(\bs m^o,v^o) \ge \Lambda \int^1_0 \Big\{ \, \big|\partial_3\big(\mathfrak{M}^\eps-\bs m^o \big)\big|^2+\big|\mathfrak{M}^\eps-\bs m^o\big|^2 \, \Big\} dx_3. \nonumber
\end{align} 
For fixed $\bs x_p \in \omega, \ \bs M^\eps(\bs x_p,x_3)=\mathfrak{M}^\eps(x_3)-\bs m^o(x_3) $. Integrating above result over $ \bs x_p \in \omega $ gives
\begin{align} \int_{\omega} \mathcal I^o(\bm{\mathfrak{M}}^\eps,\mathfrak V_3^\eps) \, \bs d \bs x_p- Q_0 =\int_{\omega} \mathcal I^o(\bs m^\eps,u_3^\eps) \, \bs d \bs x- Q_0= \Lambda \big( \big\Vert \partial_3 \bs M^\eps \big\Vert^2_{L^2(\Omega)}+ \big\Vert \bs M^\eps \big\Vert^2_{L^2(\Omega)} \big) \nonumber
\end{align}
noting $Q_0=\mbox{\fontsize{9}{8}\selectfont $\displaystyle{ \int_{\omega} } $ } \hspace{-1mm}\mathcal I^o(\bs m^o,v^o) \bs d \bs x_p $. From eqn. \eqref{f0definition}, $\mathbb C \big[ \bs \kappa^{\eps}[\bs u^\eps] - \bs E_s(\bs m^\eps) \big]^2 \ge f_0\big( \partial_3 u_3^\eps - \bs E_{s_{33}}(\bs m^\eps) \big)$ which gives $ \mathfrak{C}^\eps \ \ge \ \mbox{\fontsize{9}{8}\selectfont $\displaystyle{ \int_{\omega} } $ } \mathcal I^o(\bs m^\eps,u_3^\eps) \bs d \bs x - Q_0 \ $ and our final result. 
\end{proof}
We use the result of Lemma \ref{lemma2.2.1} above and Proposition \ref{proposition3.1.8} to get
\begin{align} \hspace{-3mm}
\mathcal{I} & ^{\eps}( \bs m^\eps ,\bs u^\eps) - Q_0 = \mathfrak{A}^\eps+\mathfrak{B}^\eps+\mathfrak{C}^\eps \nonumber \\
&\ge \frac{d}{\eps^2} \big\Vert \nabla_p \bs m^\eps \big\Vert^2_{L^2(\Omega)}+\big[ \mathcal E_d^\eps(\bs m^\eps) - \int_{\Omega} \pi \big| \bs m^\eps_p \big|^2 \bs d \bs x \big]+\Lambda \Big( \  \big\Vert \partial_3 \bs M^\eps \big\Vert^2_{L^2(\Omega)} + \big\Vert \bs M^\eps \big\Vert^2_{L^2(\Omega)} \ \Big) \nonumber \\
&\ge \frac{d}{2\eps^2} \big\Vert \nabla_p \bs M^\eps \big\Vert^2_{L^2(\Omega)}+\frac{\Lambda}{2} \Big( \  \big\Vert \partial_3 \bs M^\eps \big\Vert^2_{L^2(\Omega)} + \big\Vert \bs M^\eps \big\Vert^2_{L^2(\Omega)} \ \Big) +\eps Q_1(\bs m^o) + \eps^2 Q_2(\bs m^o) -D_{18}\eps^2 \label{2lb} \\
&\ge -K_{12}\eps \nonumber
\end{align}
where $Q_1(\bs m^o)=\frac{8\pi}{3} \Big( \ |\bs m^o_p|^2 -2 |m^o_3|^2 \Big)$ and $Q_2(\bs m^o)=\pi^2 \Big( \frac{|\bs m^o_p|^2}{2} - |m^o_3|^2 \Big)$ as in equation \eqref{magnetostaticM}.
\subsection{Convergence of \texorpdfstring{$\mathcal{I}_1^{ \eps}(\bs m^\eps,\bs u^\eps)$}{I2convergence}}
\begin{thm1}
We have the following convergence,
\begin{align}
\lim_{\eps \rightarrow 0}\mathcal{I}_1^{ \eps}(\bs m^\eps,\bs u^\eps)=Q_1=\frac{16 \pi}{3} |m^o_3|^2 - \frac{8 \pi}{3} |\bs m^o_p|^2.
\nonumber \end{align}
\end{thm1}
\begin{proof}
Dividing equation \eqref{2lb} by $ \eps $ gives
\begin{align}
\mathcal{I}_1 ^{\eps}(\bs m^\eps,\bs u^\eps) &\ge \frac{1}{\eps} \big( \eps Q_1 + \eps^2 Q_2 \big) =\frac{8\pi}{3} \Big( \ |\bs m^o_p|^2 -2 |m^o_3|^2 \Big)+\eps \pi^2 \Big( \frac{|\bs m^o_p|^2}{2} - |m^o_3|^2 \Big)\nonumber 
\end{align}
Taking $ \liminf_{\eps \rightarrow 0} $ above to get,
\begin{align}
\liminf_{\eps \rightarrow 0} \mathcal{I}_1 ^{\eps}(\bs m^\eps,\bs u^\eps) &\ge\frac{8\pi}{3} \Big( \ |\bs m^o_p|^2 -2 |m^o_3|^2 \Big). \nonumber \end{align} 
To get the reverse inequality we take divide eqn. \eqref{I2upper} by $\eps$ and then take $ \limsup$ to get,
\begin{align}
\limsup_{\eps \rightarrow 0} \mathcal{I}_1(\bs m^\eps,\bs u^\eps) &\le \limsup_{\eps \rightarrow 0} \frac{1}{\eps} \Big( \eps Q_1 + \eps^2 Q_2 \Big)=Q_1=\frac{8\pi}{3} \Big( \ |\bs m^o_p|^2 -2 |m^o_3|^2 \Big). \nonumber
\end{align}
The $ \limsup $ and $ \liminf $ inequality together gives our result.
\end{proof} 
\newtheorem{Rem7}[theorem]{Remark}
\begin{Rem7} \label{remark6.1}
In the limit we get that $\mathcal{I}_1^{ \eps}( \bs m^\eps,\bs u^\eps)$ converges to a fixed quantity $Q_1(\bs m^o)$ depending only on $ \bs m^o$. $ Q_1(\bs m^o) $ consists of the mutual interaction of the poles generated by $ \bs m^o$ on one end $\omega(0) $ of the wire domain $\Omega $ with the other end $\omega(1) $ giving the term $\frac{16}{3} \pi |m^o_3|^2 $, and the self-interaction of the poles created by $ \bs m^o$ on the curved surface $ \partial \omega \times (0,1) $ giving the term $- \frac{8}{3} \pi |\bs m^o_p|^2$. 
\end{Rem7}


{ \centering\section{Third variational limit problem} \label{Secthirdgammalimit}  }
As in the previous section we first define $ \mathcal{I}_2^{ \eps}(\bs m^\eps,\bs u^\eps) :=\eps^{-2} \big( \mathcal{I}^\eps(\bs m^\eps,\bs u^\eps) - Q_0 -\eps Q_1(\bs m^o) \big)$. We will show that $ \mathcal{I}_2^{ \eps}(\bs m^\eps,\bs u^\eps)$ is bounded above and below independent of $ \eps$. Then we define $\bs w^\eps$ in \eqref{chi} and prove a weak compactness result for it. The convergence is improved to strong in Theorem \ref{theorem5.1} where we also define a new variational problem $ \mathcal I^o_2$ and show its relation with $\mathcal I^\eps_2(\bs m^\eps,\bs u^\eps)$. \\
Recalling $\mathfrak{A}^\eps,\mathfrak{B}^\eps$ and $\mathfrak{C}^\eps$ from equation \eqref{A+B+C} in \S \ \ref{Secsecondgammalimit} \,, we note 
\begin{align}
& \mathcal{I} ^{\eps}(\bs m^\eps, \bs u^\eps ) - Q_0 -\eps Q_1(\bs m^o) \nonumber \\
& \ = \Big[ \int_\Omega  \frac{d}{\eps^2} | {\nabla_p \bs m^\eps} |^2 \bs d \bs x \Big] + \Big[ \mathcal{E}^\eps_d(\bs m^\eps)-\pi \int_\Omega |\bs m^\eps_p|^2 \bs d \bs x -\eps Q_1 \Big] + \Big[ \int_\Omega \Big\{  d \ |\partial_3 {\bs m}^\eps|^2  + \varphi(\bs m^\eps) \nonumber \\ 
& \qquad - \bs h_a \cdot \bs m^\eps + 2| \bs m^\eps_p |^2 + \frac{1}{2} \mathbb C [ \bs \kappa^{\eps}[\bs u^\eps] - \bs E_s(\bs m^\eps) ]^2 \Big\} \bs d \bs x - Q_0 \Big] \nonumber \\
& \ = \mathfrak{A}^\eps+\big( \, \mathfrak{B}^\eps-\eps Q_1(\bs m^o)\big)+\mathfrak{C}^\eps.\nonumber
\end{align}
\subsection{ Boundedness of \texorpdfstring{$ \mathcal I_2^\eps(\bs m^\eps, \bs u^\eps) $}{bound on I2} }
To get an upper bound on $ \mathcal{I}_2^{ \eps}(\bs m^\eps,\bs u^\eps) $ we use the upper bound on $ \mathcal{I}_1^{ \eps}(\bs m^\eps,\bs u^\eps) $ from eqn. \eqref{I2upper} and subtract $ Q_1(\bs m^o) $ from both sides to get
\begin{align}
\mathcal{I}_2^{ \eps} ( \bs m^\eps ,\bs u^\eps) =\frac{1}{\eps} \big[ \mathcal{I}_1^{ \eps} ( \bs m^\eps ,\bs u^\eps) -\eps Q_1 \big] &\le \frac{1}{\eps} \big[ \mathcal{I}_1^{ \eps} ( \bs m^o ,\bs u^o) -\eps Q_1 \big]= \ Q_2(\bs m^o) \ \le \ K_{12}. \label{I2ub}\end{align}
To get a lower bound on $\mathcal{I}_2^{\eps}( \bs m^\eps ,\bs u^\eps)$, we subtract $\eps Q_1(\bs m^o)$ from the lower bound in the previous section in equation \eqref{2lb} to get
\begin{align}
\mathcal{I}^{\eps}( & \bs m^\eps ,\bs u^\eps) - Q_0 -\eps Q_1(\bs m^o) = \mathfrak{A}^\eps+\big( \, \mathfrak{B}^\eps-\eps Q_1(\bs m^o)\big)+\mathfrak{C}^\eps \nonumber \\
&\ge \frac{d}{2\eps^2} \big\Vert \nabla_p \bs M^\eps \big\Vert^2_{L^2(\Omega)}+\frac{\Lambda}{2} \Big( \  \Vert \partial_3 \bs M^\eps \Vert^2_{L^2(\Omega)} + \Vert \bs M^\eps \Vert^2_{L^2(\Omega)} \ \Big) +\eps^2 Q_2(\bs m^o) -D_{18}\eps^2 \label{3lb} \\
&\ge -K_{13}\eps^2. \nonumber
\end{align}
The upper bound \ref{I2ub} and lower bound in \ref{3lb} together give with Sobolev inequality on $\Omega$ 
\begin{align}
K_{14} \eps^2 &\ge \frac{d}{2\eps^2} \big\Vert \nabla_p \bs M^\eps \big\Vert^2_{L^2(\Omega)}+\frac{\Lambda}{2} \Big( \  \Vert \partial_3 \bs M^\eps \Vert^2_{L^2(\Omega)} + \Vert \bs M^\eps \Vert^2_{L^2(\Omega)} \ \Big) \nonumber \\&\ge \frac{\Lambda}{2} \big\Vert \bs M^\eps \big\Vert^2_{H^1(\Omega)} \ge C_q \big\Vert \bs M^\eps \big\Vert^2_{L^q(\Omega)}, \qquad \forall 1 \le q \le 6 \label{boundonMeps}
\end{align}
with $ C_q $ being the appropriate Sobolev constant.
\subsection{Weak convergence of \texorpdfstring{$ \bs w^\eps$}{weak conv I2}}
In this subsection we will extract some energy terms from the elastic energy and define a new variable $\bs w^\eps$ from the extracted terms. For this we need an improvement on Lemma \ref{lemma2.2.1} . For that first note that using a truncated Taylor Expansion we write $\bs E_s(\bs m^\eps)$ as
\begin{align}
\bs E_s(\bs m^\eps)-\bs E_s(\bs m^o)=\bs E_s^\prime(\bs m^o) \cdot \bs M^\eps+\frac{1}{2}\bs E_s^{\prime\prime}(\bs m^o)\bs M^\eps \cdot \bs M^\eps+o(|\bs M^\eps|^2):=\Delta(\bs M^\eps) \label{trunctaylor}
\end{align}
where we recall $\bs m^\eps=\bs m^o+\bs M^\eps$ and $\bs E_s^{\prime}(\bs m)$ and $\bs E_s^{\prime\prime}(\bs m) $ are the $1^{\mbox{st}}$ and $2^{\mbox{nd}}$ derivatives of $ \bs E_s(\bs m) $ w.r.t. $ \bs m $. Since $ \bs E_s(\bs m) $ is a polynomial function of $ \bs m$,  both $\bs E_s^{\prime}$ and $\bs E_s^{\prime\prime} $ are bounded in $ L^\infty$ for $ |\bs m| =m_s$. Then using \eqref{boundonMeps} we get
\begin{align}
& \Delta(\bs M^\eps) \lel K_{15}|\bs M^\eps|+K_{16}|\bs M^\eps|^2 & \mbox{ and } & &  \big\Vert \Delta(\bs M^\eps) \big\Vert^2_{L^2(\Omega)} \lel K_{17} \eps^2. \label{Deltabound}  
\end{align}
Set $ \bs u^\eps=\bs u^o+\bs U^\eps$. Note $\bs \kappa^\eps[\bs u^\eps]=\bs \kappa^\eps[\bs u^o]+\bs \kappa^\eps[\bs U^\eps]$.  Eqn. \eqref{elasticbound2} and Young's inequality gives
\begin{align}
\Big| \mathbb C\big[\Delta(\bs M^\eps)\big]:\bs \kappa^\eps[\bs U^\eps] \Big| \ \le \ \Gamma \, \big|\Delta(\bs M^\eps)\big| \ \big|\bs \kappa^\eps[\bs U^\eps]\big| \ \le \ \frac{4\Gamma^2}{\gamma} \big|\Delta(\bs M^\eps)\big|^2 + \frac{\gamma}{4} \big|\bs \kappa^\eps[\bs U^\eps]\big|^2. \label{yng1}
\end{align}
From \eqref{C-f0Vtest} note that $ \bs \kappa^\eps[\bs u^o] -\bs E_s(\bs m^o) = \bs 0$. Then equations \eqref{elasticbound2}, \eqref{f0definition} and \eqref{yng1} gives
\begin{align}
& \hspace{-9mm} \frac{\mathbb C}{2} \big[\bs \kappa^\eps[\bs u^\eps]-\bs E_s(\bs m^\eps)\big]^2 =\frac{\mathbb C}{2} \big[\bs \kappa^\eps[\bs u^o]+\bs \kappa^\eps[\bs U^\eps]-\bs E_s(\bs m^\eps)\big]^2  \nonumber \\
& \hspace{-10mm} =\frac{\mathbb C}{2} \big[\bs \kappa^\eps[\bs u^o] -\bs E_s(\bs m^\eps)\big]^2 + \frac{\mathbb C}{2} \big[\bs \kappa^\eps[\bs U^\eps] \big]^2 + \mathbb C \big[\bs \kappa^\eps[\bs u^o] -\bs E_s(\bs m^o)\big]:\bs \kappa^\eps[\bs U^\eps\big]+\mathbb C \big[\Delta(\bs M^\eps)\big]:\bs \kappa^\eps[\bs U^\eps\big] \nonumber \\
& \hspace{-10mm} \ge \frac{1}{2}f_0\big(\partial_3 v^o-E_{s_{33}}(\bs m^\eps)\big)+ \frac{\gamma}{2} \big|\bs \kappa^\eps[\bs U^\eps]\big|^2 -\frac{4\Gamma^2}{\gamma} \big|\Delta(\bs M^\eps)\big|^2-\frac{\gamma}{4} \big|\bs \kappa^\eps[\bs U^\eps]\big|^2 \nonumber  \\
&\hspace{-10mm} = \frac{1}{2}f_0\big(\partial_3 v^o-E_{s_{33}}(\bs m^\eps)\big)+ \frac{\gamma}{4} \big|\bs \kappa^\eps[\bs U^\eps]\big|^2 -\frac{4\Gamma^2}{\gamma} \big|\Delta(\bs M^\eps)\big|^2.
\label{mathfrakcbound} \end{align}
Then using equation \eqref{Deltabound}, and above result \eqref{mathfrakcbound}, we improve Lemma \ref{lemma2.2.1} to get
\begin{align}
\mathfrak C^\eps &=\int_\Omega \Big\{  d \ |\partial_3 {\bs m}^\eps|^2  + \varphi(\bs m^\eps) - \bs h_a \cdot \bs m^\eps + \pi| \bs m^\eps_p |^2 + \frac{1}{2} \mathbb C \big[ \bs \kappa^{\eps}[\bs u^\eps] - \bs E_s(\bs m^\eps) \big]^2 \Big\}\bs d \bs x - Q_0  \nonumber \\
&\ge \int_\Omega \Big\{  d \ |\partial_3 {\bs m}^\eps|^2  + \varphi(\bs m^\eps) - \bs h_a \cdot \bs m^\eps + \pi| \bs m^\eps_p |^2 + \frac{1}{2} f_0\big(\partial_3 v^o - \bs E_{s_{33}}(\bs m^\eps)\big) \Big\}\bs d \bs x  \nonumber \\
& \hspace{56mm}  + \int_{\Omega} \Big\{ \, \frac{\gamma}{4} \big|\bs \kappa^\eps[\bs U^\eps]\big|^2 -\frac{4\Gamma^2}{\gamma} |\Delta(\bs M^\eps)|^2 \, \Big\}\bs d \bs x - Q_0 \nonumber \\
&\ge \Big[ \int_\omega \mathcal I^o(\bs m^\eps,v^o)\bs d \bs x_p -Q_0 \Big] + \frac{\gamma}{4} \big\Vert \bs \kappa^\eps[\bs U^\eps]\big\Vert^2_{L^2(\Omega)} -4\Gamma^2 \gamma^{-1} \big\Vert \Delta(\bs M^\eps) \big\Vert_{L^2(\Omega)}^2 \nonumber \\
&\ge \Lambda \Big( \, \big\Vert \partial_3 \bs M^\eps \big\Vert^2_{L^2(\Omega)}+ \big\Vert \bs M^\eps \big\Vert^2_{L^2(\Omega)} \Big) +\frac{\gamma}{4} \big\Vert \bs \kappa^\eps[\bs U^\eps]\big\Vert^2_{L^2(\Omega)}-K_{18} \eps^2
\end{align}
where we have used the strong second variation condition on $\mathcal I^o$ in the last step. Using this we revisit the lower bound equation \eqref{3lb} using Proposition \ref{proposition3.1.8} to give
\begin{align}
\mathcal{I} ^{\eps}(\bs m^\eps &,\bs u^\eps) - Q_0 -\eps Q_1(\bs m^o) = \mathfrak{A}^\eps+\big(\mathfrak{B}^\eps-\eps Q_1(\bs m^o)\big)+\mathfrak{C}^\eps \nonumber \\
&\ge \int_{\Omega} \Big\{ \, \frac{d}{\eps^2} |\nabla _p \bs M^\eps|^2 +\Lambda|\partial_3 \bs M^\eps|^2+\Lambda| \bs M^\eps|^2+\frac{\gamma}{4} \big|\bs \kappa^\eps[\bs U^\eps]\big|^2 \, \Big\}-K_{18} \eps^2 +\big(\mathfrak{B}^\eps-\eps Q_1\big)
\nonumber \\
&\ge \frac{\gamma}{4} \big\Vert \bs \kappa^\eps[\bs U^\eps]\big\Vert^2_{L^2(\Omega)} +\eps^2 Q_2(\bs m^o)-D_{18}\eps^2-K_{18} \eps^2.
\end{align}
The upper bound \eqref{I2ub} and the above equation then gives
\begin{align}
K_{19} \eps^2 \ge \int_{\Omega} \big|\bs \kappa^{\eps}(\bs U^\eps) \big|^2. \label{basicboundelastic} \end{align} 
Set $ \bs w^\eps = ( \bs u_p^\eps - \bs u^o_p, \eps^{-1}\big( u_3^\eps - u^o_3 \big) ) = ( \bs U_p^\eps, \eps^{-1} U_3^\eps) $ and note
\begin{align} \bs \kappa^{\eps}(\bs U^\eps) & =\begin{bmatrix} \frac{1}{\eps}\partial_1  w^\eps_1 &  \frac{1}{2 \eps}( \partial_1 w^\eps_2 + \partial_2  w^\eps_1) & \frac{1}{2}( \partial_1 w^\eps_3 + \partial_3 w^\eps_1) \\ \frac{1}{2 \eps}(\partial_1 w^\eps_2 + \partial_2  w^\eps_1) &  \frac{1}{\eps} \partial_2  w^\eps_2 &  \frac{1}{2} ( \partial_2 w^\eps_3 + \partial_3 w^\eps_2)  \\ \frac{1}{2}( \partial_1 w^\eps_3 + \partial_3 w^\eps_1)  & \frac{1}{2}( \partial_2 w^\eps_3 + \partial_3 w^\eps_2) & \eps \partial_3 w^\eps_3 \end{bmatrix} \nonumber \\
&=: \bs \chi (\bs w^\eps) . \label{chi} \end{align}
Note $ \displaystyle{ \Big| \frac{\bs \chi (\bs w^\eps)}{\eps} \Big| \ge \big| \bs E(\bs w^\eps) \big| }$ where $ \bs E(\bs w^\eps)$ is the elastic strain of field $\bs w^\eps$. Korn's inequality in \eqref{basicboundelastic} gives,
\begin{align}
K_{19} \ge \int_{\Omega} \Big| \frac{\bs \kappa^{\eps}(\bs U^\eps)}{\eps} \Big|^2 \bs d \bs x= \int_{\Omega} \Big| \frac{\bs \chi (\bs w^\eps) }{\eps} \Big|^2 \bs d \bs x & \ge \int_{\Omega} \big| \bs E(\bs w^\eps) \big|^2 \bs d \bs x \ge \alpha \int_{\Omega} \big( \, \big| \nabla \bs w^\eps \big|^2 + | \bs w^\eps |^2 \big) \bs d \bs x \nonumber 
\end{align}
where $ \alpha(\Omega) > 0 $ is the Korn's constant. These results together imply for some unrelabeled subsequence 
\begin{subequations}
\begin{align} &\bs w^\eps \rightharpoonup \bs w^o \mbox{ in } H^1(\Omega;\mathbb R^3)  & \bs w^\eps \rightarrow \bs w^o \mbox{ in }L^2(\Omega;\mathbb R^3) \label{wepsweak1} \\
& \bs E(\bs w^\eps) \rightharpoonup \bs E(\bs w^o) \mbox{ in } L^2(\Omega;\mathbb R^3) &  \frac{\bs \chi^\eps }{\eps} \rightharpoonup \bs \nu^0  \mbox{ in } L^2(\Omega;\mathbb R^3). \label{wepsweak2} 
\end{align}
\end{subequations}
Note from \eqref{chi},
\begin{align}
& \frac{ \chi_{ij}^\eps }{\eps} = \frac{1}{\eps^2} E_{ij}(\bs w^\eps) \mbox{ for } (i,j) \in \{1,2\},
& \frac{ \chi_{i3}^\eps }{\eps} = \frac{1}{\eps} E_{i3}(\bs w^\eps) \mbox{ for } i \in \{1,2\} \nonumber
\end{align}
which together imply after using lower semi-continuity of norm w.r.t weak convergence
\begin{align}
\big\Vert E_{ij}(\bs w^o) \big\Vert_{L^2(\Omega)} \le \liminf_{\eps \rightarrow 0} \ \big\Vert E_{ij}(\bs w^\eps) \big\Vert_{L^2(\Omega)} = 0
\end{align}
when $ i \in \{ 1,2 \} $ and $ j \in \{ 1,2,3 \}$. 
\newtheorem{lem3}{Lemma}[section]
\begin{lem3}\label{lemma5.1}
Let the strain corresponding to a displacement field $ \bs w^o \in H_\sharp^1(\Omega) $, $ \bs E(\bs w^o) $ be such that $ E_{ij}(\bs w^o) = 0 $ for $ i \in \{ 1,2 \} $ and $ j \in \{ 1,2,3 \}$. Then $\exists \gamma^o(x_3) \in H^1_\sharp(0,1)$ so that $ \bs w^o $ is given by,
\begin{align}
&w^o_1(\bs x)=w^o_1(x_3) \in H_{\sharp\sharp}^2(0,1) \ , &w^o_2(\bs x)=w^o_2(x_3) \in H{\sharp\sharp}^2(0,1) \ , \nonumber \\
&w^o_3(\bs x) = - x_1 \ \partial_{3} w^o_1(x_3) - x_2 \ \partial_{3} w^o_2(x_3) + \gamma^o(x_3) \label{w03,3}
\end{align} 
and $ H_{\sharp\sharp}^2(0,1) = \big\{ w \in  H^2(0,1) : w(x_3=0)=\partial_3 w(x_3=0)=0 \ \big\}.$
\end{lem3}
\begin{proof}
Given $ E_{11}(\bs w^o) = \partial_1 w^o_1(\bs x)= 0 $ and $ E_{22}(\bs w^o) = \partial_2 w^o_2(\bs x)= 0 $ together gives $ w^o_1(\bs x) = \alpha_1(x_2,x_3) $ and $ w^o_2(\bs x) = \alpha_2(x_1,x_3) $. $E_{12}(\bs w^o) =0$ gives us,
\begin{align} \partial_2 w^o_1(\bs x) + \partial_1 w^o_2(\bs x)= \partial_2 \alpha_1(x_2,x_3) + \partial_1 \alpha_2(x_1,x_3) =0. \nonumber \end{align}
This implies $ \partial_2 \alpha_1(x_2,x_3) =- \partial_1 \alpha_2(x_1,x_3) =\beta(x_3) $. Thus, $ w^o_1(\bs x) = \gamma_1(x_3) + x_2\beta(x_3) $ and $ w^o_2(\bs x) = \gamma_2(x_3) -x_1\beta(x_3)$. Also given $ E_{13}(\bs w^o) = E_{23}(\bs w^o) =0$, we have
\begin{align} \partial_2 E_{13}(\bs w^o) - \partial_1 E_{23}(\bs w^o) = \frac{1}{2} (\partial_{32} w^o_1 + \partial_{12}w^o_3 -\partial_{31}w^o_2 - \partial_{12}w^o_3 ) =\partial_3 \beta(x_3) =0. \nonumber
\end{align}
This gives us $ \beta(x_3) =K_{20} $ is constant. Using the Dirichlet boundary conditions at the base $ x_3=0$, we have $ w^o_1(x_2,0)= \gamma_1(0) + x_2 K_{20} =0 $ which gives us $ K_{20}=0$. So $ w^o_1(\bs x)= \gamma_1(x_3) $ and $ w^o_2(\bs x)= \gamma_2(x_3) $. We finally have,
\begin{align}
& E_{13}(\bs w^o)=0 \Rightarrow \partial_1 w^o_3(\bs x) = - \partial_3 w^o_1(\bs x) = - \partial_3 \gamma_1(x_3), \nonumber \\
& E_{2,3}(\bs w^o)=0 \Rightarrow \partial_2 w^o_3(\bs x) = - \partial_3 w^o_2(\bs x) = - \partial_3\gamma_2(x_3). \nonumber \end{align}
which gives us for $w^o_3$ on integrating above equations
\begin{align}
& w^o_3(\bs x) = - x_1 \, \partial_3 \gamma_1(x_3) - x_2 \, \partial_3 \gamma_2(x_3) + \gamma^o(x_3) = - x_1 \ \partial_{3} w^o_1(x_3) - x_2 \ \partial_{3} w^o_2(x_3) + \gamma^o(x_3), \nonumber \\
& \partial_3 w^o_3(\bs x) = - x_1 \ \partial_{33} w^o_1(x_3) - x_2 \ \partial_{33} w^o_2(x_3) + \partial_3 \gamma^o(x_3). \nonumber \end{align}
Note also that $ \bs w^o \in H^1(\Omega)$ gives $ \partial_3 w^o_3 \in L^2(\Omega) $. Equation \eqref{w03,3} gives then that $ \partial_{33} w^o_i(x_3) \in L^2(\Omega) $ for $ i=1,2 $ and $ \partial_3 \gamma_3 \in L^2(\Omega)$ and thus $ w^o_i(x_3) \in H^2(\Omega) $ and $ \gamma \in H^1(\Omega)$. Note also that the Dirichlet boundary conditions at the base gives $w^o_3(\bs x_p,0) = - x_1 \ \partial_{3} w^o_1(0) - x_2 \ \partial_{3} w^o_2(0) + \gamma^o(0) \equiv 0$ which means $\partial_{3} w^o_1(0)=\partial_{3} w^o_2(0)=\gamma^o(0)=0$.
\end{proof}
The displacement solution in Lemma \ref{lemma5.1} are well known in literature as the Bernoulli-Navier displacements. See Theorem 4.3 in \cite{trabucho1996mathematical} for more details.
\subsection{ Strong convergence of \texorpdfstring{$ \mathcal I_2^\eps(\bs m^\eps, \bs u^\eps) $}{bound on I2} }
For the third limit variational problem we assume that $\bs m^o=(0,0,m_s)$. This assumption greatly simplifies our final limit problem while essentially describing the underlying physics. Refer to Remark \ref{Rem71} to see more regarding this assumption and the more general case.
From \eqref{firstsgamma}, $\bs M^\eps=\bs m^\eps-\bs m^o $ which gives $|\bs m^\eps|^2-|\bs m^o|^2=|\bs M^\eps|^2+2 \bs m^o \cdot \bs M^\eps=|\bs M^\eps|^2+2 m_3^o M_3^\eps \equiv 0$. Also note $ \kappa_{33}^{\eps}(\bs u^\eps)=\partial_3 u_3^o + \chi_{33}^\eps(\bs w^\eps)=\partial_3 v^o + \chi_{33}^\eps(\bs w^\eps)=E_{s_{33}}(\bs m^o) + \chi_{33}^\eps(\bs w^\eps) $. Then
\begin{align}
|\bs M^\eps|^2&=-2\bs m^o \cdot \bs M^\eps =-2m_sM_3^\eps, \label{m1z}  \\
E_{s_{33}}(\bs m^\eps)&=\frac{3\lambda_{100}}{2m_s^2} \big( \, | m_3^\eps|^2 -\frac{m_s^2}{3} \, \big) =E_{s_{33}}(\bs m^o) - \frac{3\lambda_{100}}{2m_s^2} \ |\bs M_p^\eps|^2, \label{es33m0} \\ 
\kappa_{33}^{\eps}(\bs u^o)-E_{s_{33}}(\bs m^\eps) &= \partial_3 v^o -E_{s_{33}}(\bs m^o)+ \frac{3\lambda_{100}}{2m_s^2} \ |\bs M_p^\eps|^2=\frac{3\lambda_{100}}{2m_s^2} \ |\bs M_p^\eps|^2, \nonumber
\end{align} and using H$\ddot{\mbox{o}}$lder's inequality and \eqref{boundonMeps}
\begin{align}
\hspace{-8mm} \int_\Omega f_0\big( \kappa_{33}^{\eps}( & \bs u^\eps) - E_{s_{33}}(\bs m^\eps)\big)=\int_\Omega f_0\big(\partial_3 v^o \hspace{-1mm}- E_{s_{33}}(\bs m^\eps)\big)+ f_0\big(\chi_{33}^\eps(\bs w^\eps)\big) +\frac{3 Y \lambda_{100}}{2m_s^2} \ |\bs M_p^\eps|^2 \, \chi_{33}^\eps(\bs w^\eps) \nonumber \\
&\ge \int_\Omega f_0\big(\partial_3 v^o \hspace{-1mm}- E_{s_{33}}(\bs m^\eps)\big)+\int_\Omega f_0\big(\chi_{33}^\eps(\bs w^\eps)\big)-Y \frac{3\lambda_{100}}{2m_s^2} \big\Vert \bs M^\eps \big\Vert^2_{L^4(\Omega)} \big\Vert \chi_{33}^\eps(\bs w^\eps) \big\Vert_{L^2(\Omega)} \nonumber \\
&\ge \int_\Omega f_0\big(\partial_3 v^o \hspace{-1mm} - E_{s_{33}}(\bs m^\eps)\big)+\int_\Omega f_0\big(\chi_{33}^\eps(\bs w^\eps)\big)-K_{21} \eps^2 \big\Vert \chi_{33}^\eps(\bs w^\eps) \big\Vert_{L^2(\Omega)}. \label{f0breakup} 
\end{align}
Using \eqref{f0definition} \ $\mathbb C \big[ \bs \kappa^{\eps}[\bs u^\eps] - \bs E_s(\bs m^\eps) \big] \ge f_0 \big( \kappa_{33}^{\eps}(\bs u^\eps) - E_{s_{33}}(\bs m^\eps)\big)$ which gives
\begin{align}
\mathfrak C^\eps & \ge \int_\Omega d \ |\partial_3 {\bs m}^\eps|^2 + \varphi(\bs m^\eps) - \bs h_a \cdot \bs m^\eps + 2| \bs m^\eps_p |^2 + \frac{1}{2} f_0\big(\kappa_{33}^{\eps}(\bs u^\eps) - E_{s_{33}}(\bs m^\eps)\big) -Q_0 \nonumber \\ 
& \ge \int_\Omega d \ |\partial_3 {\bs m}^\eps|^2 + \varphi(\bs m^\eps) - \bs h_a \cdot \bs m^\eps + 2| \bs m^\eps_p |^2 + \frac{1}{2} f_0\big(\partial_3 v^o - E_{s_{33}}(\bs m^\eps)\big) -Q_0 \nonumber \\
& \hspace{60mm} +\int_\Omega f_0(\chi_{33}^\eps(\bs w^\eps))-K_{21}\eps^2 \big\Vert \chi_{33}^\eps(\bs w^\eps)\big\Vert_{L^2(\Omega)}\nonumber \\ 
& = \Big[ \, \int_{\omega} \mathcal I^o(\bs m^\eps,v^o) \bs d \bs x_p -Q_0 \, \Big]+\int_\Omega f_0(\chi_{33}^\eps(\bs w^\eps))-K_{21}\eps^2 \big\Vert \chi_{33}^\eps(\bs w^\eps)\big\Vert_{L^2(\Omega)} \nonumber \\ 
& = \Lambda \Big( \,\big\Vert \partial_3 \bs M^\eps \big\Vert^2_{L^2(\Omega)}+ \big\Vert \bs M^\eps \big\Vert^2_{L^2(\Omega)} \Big)+\int_\Omega f_0(\chi_{33}^\eps(\bs w^\eps))-K_{21}\eps^2 \big\Vert \chi_{33}^\eps(\bs w^\eps)\big\Vert_{L^2(\Omega)} 
\label{7.18} \end{align}
where we have used strong second variation condition on $\mathcal I^o(\bs m^\eps,v^o)$ in the last step.

Define $\mathcal{I}_2^o(w_1(x_3),w_2(x_3),\nu(x_3))$ in function space $ \mathcal A_2 $ as
\begin{align}
\hspace{-2mm} \mathcal{I}_2^o(w_1,w_2,\gamma)&=\frac{1}{2} \int_\Omega \Big[ f_0\big(x_1 \partial_{33} w_1(x_3)\big)+f_0\big(x_2 \partial_{33} w_2(x_3)\big)+f_0\big(\partial_{3} \gamma(x_3)\big) \Big]\bs d \bs x+Q_2(\bs m^o) \label{I0defn.}
\end{align}
and $\mathcal A_2 := \big\{ (w_1(x_3),w_2(x_3),\gamma(x_3)) \in H_{\sharp\sharp}^2(0,1) \times H_{\sharp\sharp}^2(0,1) \times H_\sharp^1(0,1)\big\}$.
\begin{thm1}
\label{theorem5.1}
There exists a subsequence $\bs w^\eps$ not relabeled such that $ \bs w^\eps \rightarrow \bs w^o$ $ \mbox{ strongly in }$ $ H^1(\Omega,\mathbb R^3) $. $\bs w^o $ is given as in Lemma \ref{lemma5.1} and $ (w^o_1(x_3),w^o_2(x_3),\gamma^o(x_3)) $ minimizes $\mathcal{I}^0_2 $ in $ \mathcal A_2$ and $ \lim_{\eps \rightarrow 0} \mathcal I_2^\eps(\bs m^\eps,\bs u^\eps)=\mathcal I_2^o(w^o_1,w^o_2,\gamma^o)$ where $ (w^o_1,w^o_2,\gamma^o) $ minimizes $ \mathcal I^o_2(w_1(x_3),w_2(x_3),\nu(x_3)) $ in $ \mathcal A_2 $.
\end{thm1}
\begin{proof} 
Because $\bs m^o=(0,0,m_s)$ we use the relevant magnetostatic estimate from equation \eqref{finalremark} in the remark \ref{RemA2} \, following Proposition \ref{proposition3.1.8} and equation \eqref{7.18} to give
\begin{align}
\mathcal{I}^{\eps}(& \bs m^\eps,\bs u^\eps) - Q_0 -\eps Q_1(\bs m^o) = \mathfrak{A}^\eps+\big(\mathfrak{B}^\eps-\eps Q_1(\bs m^o) \big)+\mathfrak{C}^\eps \nonumber \\
&\ge \int_{\Omega} \Big\{ \, \frac{d}{\eps^2} |\nabla _p \bs M^\eps|^2 +\Lambda|\partial_3 \bs M^\eps|^2+\Lambda| \bs M^\eps|^2+\frac{1}{2} f_0\big(\chi_{33}^\eps(\bs w^\eps)\big) \, \Big\}- K_{21}\eps^3 \Vert \chi_{33}^\eps(\bs w^\eps)\Vert_{L^2(\Omega)}
\nonumber \\
& \qquad +\big(\mathcal E_d^\eps(\bs m^\eps)-\pi \int_{\Omega}\big| \bs m^\eps_p\big|^2-\eps Q_1(\bs m^o) \big) \nonumber \\
&\ge \int_{\Omega} \frac{1}{2} f_0\big(\chi_{33}^\eps(\bs w^\eps)\big) -K_{21} \eps^3 \Big\Vert \frac{\chi_{33}^\eps(\bs w^\eps)}{\eps}\Big\Vert_{L^2(\Omega)}+\eps^2 Q_2(\bs m^o).
\end{align}
Dividing by $ \eps^2 $, noting $\eps^{-1} \chi_{33}^\eps(\bs w^\eps) = \partial_3 w^\eps_3$ and $\partial_3 w^\eps_3 \rightharpoonup \partial_3 w^o_3$ in $L^2(\Omega)$, we get on taking $\liminf_{\eps \rightarrow 0}$,
\begin{align}
\liminf_{\eps \rightarrow 0} \mathcal{I}_2^{\eps}( \bs m^\eps,\bs u^\eps) & \ge \int_{\Omega} \frac{1}{2} f_0 \big(\partial_3 w_3^o \big)+Q_2(\bs m^o).
\end{align} 
The first term in the R.H.S comes from the fact that \eqref{f0definition} gives, $ \mbox{\fontsize{8}{8}\selectfont $\displaystyle{ \int_\Omega } $ } f_0(\partial_3 w_3^\eps)= Y \big\Vert \partial_3 w_3^\eps \big\Vert^2_{L^2(\Omega)} $ and $\partial_3 w_3^\eps \rightharpoonup \partial_3 w_3^o $ in $ L^2(\Omega) $ which implies $ \big\Vert \partial_3 w_3^o \big\Vert_{L^2(\Omega)} \le \liminf \big\Vert \partial_3 w_3^\eps \big\Vert_{L^2(\Omega)} $. \\
To get the $\limsup$ inequality, we compare energy of $\mathcal I^\eps_2 $ at its minimizer $(\bs m^\eps,\bs u^\eps)$ against a test function $(\bs m^o,\bs U=\bs u^o+(\bs W_p,\eps W_3))$ where 
\begin{align}
&W_1=w_1^o(x_3)-\eps^2 \sigma x_1 \partial_{3} \gamma(x_3) +\eps^2 \frac{\sigma}{2}\big(x_1^2\partial_{33}w_1^o-x_2^2\partial_{33}w_1^o+2x_1x_2\partial_{33}w_2^o\big) \label{W1} \\
&W_2=w_2^o(x_3)-\eps^2 \sigma x_2 \partial_{3} \gamma(x_3) +\eps^2 \frac{\sigma}{2}\big(x_2^2\partial_{33}w_2^o-x_1^2\partial_{33}w_2^o+2x_1x_2\partial_{33}w_1^o \big)\nonumber \\
&W_3=w_3^o(\bs x)=\gamma^o(x_3)-x_1\partial_3 w_1^o-x_2\partial_3 w_2^o 
\nonumber
\end{align}
Then $ \bs \kappa^\eps(\bs U)-\bs E_s(\bs m^o)= \bs \kappa^\eps(\bs u^o)-\bs E_s(\bs m^o)+\bs \chi^\eps(\bs W)=\bs \chi^\eps(\bs W) $ where $ \eps^{-1} \bs \chi^\eps(\bs W) $ converges as:
\begin{align}
& \eps^{-1}\bs \chi_{11}^\eps(\bs W) = \eps^{-1}\bs \chi_{22}^\eps(\bs W) =-\sigma\big(\partial_3\gamma^o(x_3)-x_1\partial_{33} w_1^o-x_2\partial_{33} w_2^o\big),  \qquad  \eps^{-1} \bs \chi_{12}^\eps(\bs W) = 0 \nonumber \\
& \eps^{-1}\bs \chi_{12}^\eps(\bs W) = \frac{\eps}{2} \Big\{ - \sigma x_1 \partial_{33} \gamma^o(x_3) +\frac{\sigma}{2} \big( x_1^2\partial_{333}w_1^o-x_2^2\partial_{333}w_1^o+ 2x_1x_2\partial_{333}w_2^o \big) \Big\} \approx O(\eps) \nonumber  \\
& \eps^{-1}\bs \chi_{13}^\eps(\bs W) = \frac{\eps}{2} \Big\{ - \sigma x_2 \partial_{33} \gamma^o(x_3) +\frac{\sigma}{2} \big( x_2^2\partial_{333}w_2^o-x_1^2\partial_{333}w_2^o+2x_1x_2\partial_{333}w_1^o \big) \Big\} \approx O(\eps)
\nonumber \\
& \eps^{-1}\bs \chi_{33}^\eps(\bs W) = \partial_3 W_3 = \partial_3 \gamma^o(x_3)-x_1\partial_{33} w_1^o-x_2\partial_{33} w_2^o. \nonumber
\end{align}
Its easy to check that $ \ \mathbb C \big[\bs \kappa^\eps(\bs U)-\bs E_s(\bs m^o)\big]^2=f_0 \big( \eps \partial_3 W_3 \big) +O(\eps^3) $ gives
\begin{align}
\mathcal I_2^\eps(\bs m^\eps,\bs u^\eps) & \le \mathcal I^\eps_2(\bs m^o,\bs U) = \frac{\mathcal I^\eps(\bs m^o,\bs U)-\eps Q_1(\bs m^o) -Q_0}{\eps^2} \nonumber \\
&=Q_2(\bs m^o)+\int_{\Omega}f_0 \big( \partial_3 \gamma^o(x_3)-x_1\partial_{33} w_1^o(x_3)-x_2\partial_{33} w_2^o(x_3) \big)+O(\eps). \nonumber
\end{align}
Taking $\limsup$ as $ \eps \rightarrow 0 $ we get our result. We finally need to show that $(w^o_1,w^o_2,\gamma^o)$ minimizes $\mathcal I^o(w_1,w_2,\gamma)$ in $\mathcal A_2$. Again as in Theorem \ref{theorem2.1.1} we start of with smooth $(w_1,w_2,\gamma)$ satisfying the boundary conditions. We set up a displacement $\bs W$ exactly as in eqn. \eqref{W1} with $(w_1,w_2,\gamma)$ replacing $(w^o_1,w^o_2,\gamma^o)$. Comparing energy of $\mathcal I_2^\eps(\bs m^\eps,\bs u^\eps)$ with the test function $\big(\bs m^o,\bs u^o+(\bs W_p^\eps,W^\eps_3) \big)$ we get $\mathcal I_2^\eps(\bs m^\eps,\bs u^\eps) \le \mathcal I_2^\eps\big(\bs m^o,\bs u^o+(\bs W_p^\eps,W^\eps_3) \big)$. We get our result on taking limit as $\eps \rightarrow 0 $ and noting that smooth functions $(w_1,w_2,\gamma)$ satisfying the appropriate boundary conditions are dense in $\mathcal A_2$. 
\end{proof}
\begin{Rem7} \label{Rem71}
The assumption $\bs m^o=(0,0,m_s)$ is not necessary, but it greatly simplifies the form of the third variational limit problem $\mathcal I_2^o$. If $ \bs m^o$ is a more general constant magnetization, then the magnetization $\bs m^o$ will have a non-trivial corrector. In fact the third variational limit problem then will be a more complicated problem involving both elastic and magnetic corrector terms. The elastic part of the problem will however still retain the Euler-Bernoulli type terms and the problem however will simplify to the limit problem of Theorem \ref{theorem5.1} if $ \bs m^o=(0,0,m_s)$. 

We however do not present that result here, as we are more interested in nanowires made of Galfenol. For these wires made of Galfenol, as expressed in \S \, \ref{Sec-example}, the demagnetization term $\pi| \bs m_p |^2 $ is the largest term in the ''effective anisotropy`` $ \varphi(\bs m )+\pi| \bs m_p |^2 - \bs h_a \cdot \bs m$ by an order of magnitude for typical applied fields. Thus the minimizer $ \bs m^o$ of this ''effective anisotropy`` is expected to be at or very close to $ (0,0,m_s) $. 
\end{Rem7}

{\centering \section{Summary and Discussion} \label{Sec-summary} }
\begin{figure}[h]
\begin{center}
\includegraphics[scale=1.0]{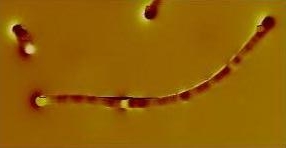}
\caption{Bent wires of Galfenol}
\label{nonlinearbending}
\end{center}
\end{figure}

We have presented in this paper the derivation of simple models for nanometer diameter wires to be used in sensors/devices using the physics of magnetostriction. Though the starting point for these problems is an infinite dimensional variational problem with a non-convex non-linear constraint and variational energy contains terms which are non-local, using the method of variational convergence we have derived much simpler 1-dimensional models which is expected to approximate the actual physics of the starting model. The Theorems \ref{theorem2.1.1} and \ref{theorem5.1} clearly set up these simpler models $\mathcal{I}^o(\bs m,v) $ and $\mathcal I^o_2(w_1,w_2,\nu) $ respectively.

The bending behavior of the nanowires is described by $\mathcal I^o_2(w_1,w_2,\nu) $ if $\bs m^o=(0,0,m_s)$ solves the first variational limit problem $\mathcal{I}^o(\bs m,v) $ as is expected for Galfenol wires. The form of second variational limit and the third variational limit suggest that the magnetization remains strongly stabilized at $\bs m^o$ and higher order theories do not add correctors to $\bs m^o$ within the framework of geometrically linear theory of magnetostriction. The displacement solution $\bs u^o$ corresponding to the first variational problem is however corrected due to the appearance of the bending energy terms in the third variational limit $\mathcal I^o_2(w_1,w_2,\nu)$. 

Although we have not included any external applied force in our analysis, it can be included with very minor changes to our presentation. The galfenol wires in bending behave like purely elastic beams with additional magnetic term which comes thorough the interaction of the positive and negative poles created at the two ends of the wire by the magnetization $ \bs m^o=(0,0,m_s)$. This contribution is a fixed energy at the order at which bending elastic terms appear. 

The strong stabilization of the magnetization is borne out by experiments where nanowires have been bend using an AFM tip. The Figure \ref{nonlinearbending} shows the MFM scan for a galfenol wire in bent shape. The details of the experiment are available from \cite{downey2008thesis}. The MFM scan shows the same bright and dark spots at the two ends of the wire characteristic of axially magnetized wires as seen in Figure \ref{fig12}. The bright spot in the middle was detected to be a topological defect. It is clear that even the large bending is unable to alter the axial magnetization, which can be interpreted as being equal to $\bs m^o$. 

The bending behavior of the nanowires will be more complicated if $\bs m^o \neq (0,0, \pm m_s)$ solves the first variational limit problem $\mathcal{I}^o(\bs m,v) $ as mentioned in Remark \ref{Rem71}. This case is however not very important for Galfenol nanowires with the geometry that we are interested in.

The highly nonlinear deformation of the nanowires in Figure \ref{nonlinearbending} also suggests to start of with a geometrically nonlinear theory for magnetostriction. For geometrically nonlinear deformations however, the problem is significantly harder as the magnetic energies in the starting energy \eqref{mathcale} will be defined on the deformed configuration, while typically in nonlinear elasticity, the free energy is defined over the reference configuration. \\
Recall the energy $\mathcal{I}_2^o(w_1(x_3),w_2(x_3),\nu(x_3))$ was defined in the previous section as 
\begin{align}
\mathcal{I}_2^o(w_1,w_2,\nu)&=\int_\Omega \frac{1}{2} \Big\{ f_0(x_1 \partial_{33} w_1)+f_0(x_2 \partial_{33} w_2)+f_0(\partial_{3} \gamma) \Big\} \bs d \bs x+Q_2(\bs m^o).  \nonumber
\end{align}
Note that the first and second term are exactly the bending energy that appears in classical Euler-Bernoulli theory. To see this note that from the definition of $ f_o $ in \eqref{youngs-modulus} we get,
\begin{align}
\hspace{-5mm} \int_\Omega f_0(x_1 \partial_{33} w_1(x_3)) \bs d \bs x = \int^1_0 \int_{\omega} Y x_1^2 \, \big| \partial_{33} w_1(x_3) \, \big|^2 &= \int^1_0 Y \Bigg\{ \int_{\omega} x_1^2 \bs d \bs x_p \Bigg\} \, \big| \partial_{33} w_1(x_3) \, \big|^2 dx_3 \nonumber \\
&=\int^1_0 Y I_{22} \, \big| \partial_{33} w_1(x_3) \, \big|^2 dx_3 \nonumber
\end{align}
where $ I_{22} $ is the polar moment of inertia.

From the point of view of using Galfenol as a potential material for sensor application, the strong stabilization of magnetization $\bs m^o=(0,0,m_s)$ is not encouraging, as a designer would hope that the magnetization would change drastically from $\bs m^o$ on imposing any bending deformation. Newer proposals for sensor design using Galfenol have been made which replace the wire array of Galfenol with an array where each wire is multi-layered with fine layers of magnetic Galfenol and non-magnetic Copper $ \big( \, \mbox{cf. } \cite{park2010characterization} \big)$. 
\newpage

\appendix
{\centering
\section{Magnetostatic calculations}
\label{Sec-Linear-gamma} }
\subsection{Introduction}
Recall in \eqref{demagstandardbound1} we defined $ \mathcal E^\eps_d(\bs m) = \mbox{\fontsize{9}{8}\selectfont $\displaystyle{ \frac{1}{8 \pi} \int_{\mathbb R^3}} $ } | \bs h^\eps_{\bs m}(\bs x) |^2 \bs d \bs x = \mbox{\fontsize{9}{8}\selectfont $\displaystyle{ \frac{1}{\eps^2} \frac{1}{8 \pi} \int_{\mathbb R^3}} $ } |\widetilde{\bs h}^\eps_{\widetilde{\bs m}}(\bs y) |^2  \bs d \bs y$. In this section we will work in the unrescaled magnetization $\widetilde{\bs m}$ and demag field $\widetilde{\bs h}^\eps_{\widetilde{\bs m}}$. We define
\begin{align}
\widehat{\bs m \,}^\eps(y_3) := \dashint_{\omega_\eps} \widetilde{\bs m \,}^\eps(\bs y_p,y_3) \ \bs d \bs  y_p. \label{unscaledaverage} 
\end{align}
Note on rescaling $\widetilde{\bs m \,}^\eps $ as in \eqref{vectorscale} $\widehat{\bs m \,}^\eps$ also corresponds to the cross-sectional average defined in \eqref{csavg} i.e., $ \  \
\widehat{\bs m \,}^\eps = \mbox{\fontsize{9}{8}\selectfont $\displaystyle{ \dashint_{\omega} } $ } \bs m^\eps(\bs x_p,x_3) \ \bs d \bs  x_p $. Thus proposition \ref{proposition2.1.1}  on $\widehat{\bs m \,}^\eps$ gives 
\begin{align} 
&\big\Vert \widehat{\bs m \,}^\eps \big\Vert^2_{L^2(0,1)} \le |\omega|^{-1} \big\Vert \bs m^\eps \big\Vert^2_{L^2(\Omega)}, \hspace{20mm} \big\Vert \partial^{\bm y}_3\widehat{\bs m \,}^\eps \big\Vert^2_{L^2(0,1)} \le |\omega|^{-1} \big\Vert \partial_3 \bs m^\eps \big\Vert^2_{L^2(\Omega)}, \nonumber \\
& \big\Vert \widehat{\bs m \,}^\eps \big\Vert^2_{H^1(0,1)} = \big\Vert \widehat{\bs m \,}^\eps \big\Vert^2_{L^2(0,1)}+\big\Vert \partial^{\bm y}_3\widehat{\bs m \,}^\eps \big\Vert^2_{L^2(0,1)}  \le \frac{1}{|\omega|} \big( \ \big\Vert \bs m^\eps \big\Vert^2_{L^2(\Omega)}+\big\Vert \partial_3 \bs m^\eps \big\Vert^2_{L^2(\Omega)}\big). \label{widetildemepsmeps1} \end{align}
Let $\widetilde{\bs h}^\eps_{\widehat{\bs m \,}^\eps} $ solve Maxwell equation for $\widehat{\bs m \,}^\eps $ on $\Omega_\eps$. We first prove the following Lemma to estimate the difference in magnetostatic energy between $\widetilde{\bs m \,}^\eps$ and $\widehat{\bs m \,}^\eps $.
\newtheorem{lemA}{Lemma}[section]
\begin{lemA}
\label{lemmaA.1.1}
The following inequality holds:
\beq \frac{1}{8 \pi \eps^2} \Big| \ \big\Vert \widetilde{\bs h}^\eps_{\widetilde{\bs m \,}^\eps} \big\Vert^2_{L^2(\mathbb R^3)} - \big\Vert \widetilde{\bs h}^\eps_{\widehat{\bs m \,}^\eps} \big\Vert^2_{L^2(\mathbb R^3)} \Big| \le 
D_2 \big\Vert \bs m^\eps \big\Vert_{L^2(\Omega)} \ \big\Vert \nabla_p \bs m^\eps \big\Vert_{L^2(\Omega)}.   \nonumber
\eeq
\end{lemA}
\begin{proof}
First we recall the basic demagnetization energy bound in equation \eqref{demagstandardbound},
\begin{align}
\frac{1}{8 \pi} \big\Vert \widetilde{\bs h}^\eps_{\widetilde{\bs m \,}} \big\Vert^2_{L^2(\mathbb R^3)} \le \frac{1}{2} \big\Vert \widetilde{\bs m \,} \big\Vert^2_{L^2(\Omega_\eps)}, \qquad \forall \ \ \widetilde{\bs m \,} \in L^2(\Omega_\eps,m_s S^2).  \label{demagstandardbound4}
\end{align}
We have $ \ \big\Vert \widetilde{\bs m \,}^\eps -\widehat{\bs m \, }^\eps \big\Vert^2_{L^2(\omega^\eps)} \le D_1 \eps^2 \big\Vert \nabla^{\bm y}_p \widetilde{\bs m \,}^\eps \big\Vert^2_{L^2(\omega^\eps)} $ using Poincar\'e inequality on a cross-section plane $ \omega_\eps(y_3) $, which on integrating on $ y_3 \in (0,1) $ gives 
\begin{align}
\big\Vert \widetilde{\bs m \,}^\eps - \widehat{\bs m \,}^\eps \big\Vert^2_{L^2(\Omega_\eps )} \ < \ \mathit D_1 \, \eps^2 \, \big\Vert \nabla^{\bm y}_p \widetilde{\bs m \,}^\eps(\bs y) \big\Vert^2_{L^2(\Omega_\eps )} \nonumber.
\end{align}
Since $ \widetilde{\bs h}^\eps_{\widehat{\bs m \,}^\eps} $ satisfies Maxwell equation for $\widehat{\bs m \,}^\eps $, by linearity $ ( \, \widetilde{\bs h}^\eps_{\widetilde{\bs m \,}^\eps}-\widetilde{\bs h}^\eps_{\widehat{\bs m \,}^\eps} \, ) $ satisfies Maxwell equation for $ ( \ \widetilde{\bs m \,}^\eps-\widehat{\bs m \,}^\eps \ ).$ Then using basic bound eqn. \eqref{demagstandardbound4} for Maxwell equation we have,
\begin{align} \frac{1}{8 \pi} \big\Vert \widetilde{\bs h}^\eps_{\widetilde{\bs m \,}^\eps}-\widetilde{\bs h}^\eps_{\widehat{\bs m \,}^\eps} \big\Vert^2_{L^2(\mathbb R^3)} \le \frac{1}{2} \big\Vert \widetilde{\bs m \,}^\eps - \widehat{\bs m \,}^\eps \big\Vert^2_{L^2(\Omega_\eps )} < \frac{\mathit D_1}{2} \, \eps^2 \, \big\Vert \nabla^{\bm y}_p \widetilde{\bs m \,}^\eps \big\Vert^2_{L^2(\Omega_\eps )}. \nonumber \end{align}
Using triangle inequality we also have, 
\beq \Big| \ \big\Vert \widetilde{\bs h}^\eps_{\widetilde{\bs m \,}^\eps} \big\Vert_{L^2(\mathbb R^3)} -\big\Vert \widetilde{\bs h}^\eps_{\widehat{\bs m \,}^\eps} \big\Vert_{L^2(\mathbb R^3)} \ \Big| \le \big\Vert \widetilde{\bs h}^\eps_{\widetilde{\bs m \,}^\eps}-\widetilde{\bs h}^\eps_{\widehat{\bs m \,}^\eps} \big\Vert_{L^2(\mathbb R^3)} < \mathit D_2 \, \eps \, \big\Vert \nabla^{\bm y}_p \widetilde{\bs m \,}^\eps \big\Vert_{L^2(\Omega_\eps )} . \nonumber \eeq
Jensen's inequality gives $ |\widehat{\bs m \,}^\eps| \le |\widetilde{\bs m \,}^\eps| $ and using basic demag bound in eqn. \eqref{demagstandardbound4} again for $ \widetilde{\bs m \,}^\eps $ and $ \widehat{\bs m \,}^\eps $ we have,
\begin{align}
\frac{1}{8\pi}\Big| \ \big\Vert \widetilde{\bs h}^\eps_{\widetilde{\bs m \,}^\eps} \big\Vert_{L^2(\mathbb R^3)} + \big\Vert \widetilde{\bs h}^\eps_{\widehat{\bs m \,}^\eps} \big\Vert_{L^2(\mathbb R^3)} \ \Big| \le \frac{1}{2} \Big( \Vert \widetilde{\bs m \,}^\eps \Vert_{L^2(\Omega_\eps )} + \Vert \widehat{\bs m \,}^\eps \Vert_{L^2(\Omega_\eps )} \Big) \le \Vert \widetilde{\bs m \,}^\eps \Vert_{L^2(\Omega_\eps )} . \nonumber
\end{align} 
Then on combining the two and rescaling $\widetilde{\bs m \,}^\eps $ to $\bs m^\eps $ we have 
\begin{align} 
\Big| \, \big\Vert \widetilde{\bs h}^\eps_{\widetilde{\bs m \,}^\eps} \big\Vert^2_{L^2(\mathbb R^3)} - \big\Vert \widetilde{\bs h}^\eps_{\widehat{\bs m \,}^\eps} \big\Vert^2_{L^2(\mathbb R^3)} \Big| &= \Big| \, \big\Vert \widetilde{\bs h}^\eps_{\widetilde{\bs m \,}^\eps} \big\Vert_{L^2} - \big\Vert \widetilde{\bs h}^\eps_{\widehat{\bs m \,}^\eps} \big\Vert_{L^2} \Big| \ . \ \Big| \, \big\Vert \widetilde{\bs h}^\eps_{\widetilde{\bs m \,}^\eps} \big\Vert_{L^2} + \big\Vert \widetilde{\bs h}^\eps_{\widehat{\bs m \,}^\eps} \big\Vert_{L^2} \Big| \nonumber \\
&\le D_2 \, \eps \, \big\Vert \nabla^{\bm y}_p \widetilde{\bs m \,}^\eps \big\Vert_{L^2(\Omega_\eps )} \cdot \ 8 \pi \ \big\Vert \widetilde{\bs m \,}^\eps \big\Vert_{L^2(\Omega_\eps )} \nonumber \\
&= 8 \pi \ D_2 \, \eps^2 \ \big\Vert \bs m^\eps \big\Vert_{L^2(\Omega)} \ \big\Vert \nabla_p \bs m^\eps \big\Vert_{L^2(\Omega)} \, , \nonumber 
\end{align}
by noting that $\Vert \widetilde{\bs m \,}^\eps \Vert_{L^2(\Omega_\eps )} =\eps \Vert \bs m^\eps \Vert_{L^2(\Omega)}$ and $\Vert \nabla^{\bm y}_p \widetilde{\bs m \,}^\eps \Vert_{L^2(\Omega_\eps )} = \Vert \nabla_p \bs m^\eps \Vert_{L^2(\Omega)}$.
\end{proof}
\newtheorem{RemA}[theorem]{Remark}
\begin{RemA} \label{RemA1}
From \eqref{ge1} we know that the exchange energy of magnetization $\bs m^\eps$ is bounded as $K_5 \ge \frac{d}{\eps^2} \Vert \nabla_p \bs m^\eps \Vert_{L^2(\Omega)}$. 
Since $\mathcal E^\eps_d(\bs m) = \frac{1}{8 \pi} \Vert \bs h^\eps_{\bs m} \Vert^2_{L^2(\mathbb R^3)}= \frac{1}{8 \pi \eps^2} \Vert \widetilde{\bs h}^\eps_{\widetilde{\bs m \,}} \Vert_{L^2(\mathbb R^3)}$, we then get from the Lemma
\begin{align} 
\hspace{-5mm} \big| \, \mathcal E^\eps_d(\bs m^\eps) - \mathcal E^\eps_d(\widehat{\bs m \,}^\eps) \, \big| = \frac{1}{8 \pi \eps^2} \Big| \ \big\Vert \widetilde{\bs h}^\eps_{\widetilde{\bs m \,}^\eps} \big\Vert^2_{L^2(\mathbb R^3)} - \big\Vert \widetilde{\bs h}^\eps_{\widehat{\bs m \,}^\eps} \big\Vert^2_{L^2(\mathbb R^3)} \Big| & \le D_2 \Vert \bs m^\eps \Vert_{L^2(\Omega)} \, \Vert \nabla_p \bs m^\eps \Vert_{L^2(\Omega)} \nonumber \\
& \le D_2 \, m_s |\Omega|^{1/2} \, \sqrt{\frac{K_5}{d}}\eps =O(\eps). \nonumber
\end{align}
Thus the difference in magnetostatic energy between $ \mathcal E^\eps_d(\bs m^\eps) $ and $\mathcal E^\eps_d(\widehat{\bs m \,}^\eps) $ is of order $ O(\eps)$. We will see that for the convergence arguments it is enough to estimate $ \mathcal E^\eps_d(\widehat{\bs m \,}^\eps) =\frac{1}{8 \pi \eps^2}\big\Vert \widetilde{\bs h}^\eps_{\widehat{\bs m \,}^\eps} \big\Vert^2_{L^2(\mathbb R^3)}$. 
\end{RemA}
For $\widehat{\bs m \,}(y_3) $ note $ \nabla^{\bm y} \cdot \widehat{\bs m \,}(\bs y) = \partial^{\bm y}_3 \widehat{m \,}_3(y_3) $. It is well know that for magnetization $ \widehat{\bs m \,}^\eps$, the energy $\frac{1}{8 \pi \eps^2}\big\Vert \widehat{\bs h}^\eps_{\widehat{\bs m \,}^\eps} \big\Vert^2_{L^2(\mathbb R^3)} $ can be written as a convolution of fundamental solutions with $ \widehat{\bs m \,}^\eps$,

\begin{align} 
\frac{\eps^{-2}}{8 \pi} \big\Vert \widetilde{\bs h}^\eps_{\widehat{\bs m \,}^\eps} \big\Vert^2_{L^2} 
&= \frac{1}{2 } \int_{\Omega_\eps} \int_{\Omega_\eps} \hspace{-2mm} \frac{\nabla^{\bm y} \cdot \widehat{\bs m \,}^\eps(\bs y) \ \nabla^{\bm y} \cdot \widehat{\bs m \,}^\eps(\bs z) }{ \eps^2|\bs y-\bs z| } + \frac{1}{2} \int_{\partial \Omega_\eps} \int_{\partial \Omega_\eps} \hspace{-3mm}\frac{\widehat{\bs m \,}^\eps(\bs y) \cdot \widetilde{\bs n}(\bs y) \ \widehat{\bs m \,}^\eps(\bs z) \cdot \widetilde{\bs n}(\bs z) }{ \eps^2|\bs y-\bs z| } \nonumber \\
&-\int_{\Omega_\eps} \int_{\partial \Omega_\eps} \frac{\nabla^{\bm y} \cdot \widehat{\bs m \,}^\eps(\bs y) \ \widehat{\bs m \,}^\eps(\bs z) \cdot \widetilde{\bs n}(\bs z) }{ \eps^2|\bs y-\bs z| }  \nonumber \\
&= \frac{1}{2 } \int_{\Omega_\eps} \int_{\Omega_\eps} \hspace{-2mm}\frac{\partial^{\bm y}_3 \widehat{m \,}^\eps_3(y_3) \ \partial^{\bm y}_3 \widehat{m \,}^\eps_3(z_3) }{ \eps^2|\bs y-\bs z| } + \frac{1}{2} \int_{\partial \Omega_\eps} \int_{\partial \Omega_\eps} \hspace{-3mm}\frac{\widehat{\bs m \,}^\eps(y_3) \cdot \widetilde{\bs n}(\bs y) \ \widehat{\bs m \,}^\eps(z_3) \cdot \widetilde{\bs n}(\bs z) }{ \eps^2|\bs y-\bs z| } \nonumber \\
&-\int_{\Omega_\eps} \int_{\partial \Omega_\eps} \frac{\partial^{\bm y}_3 \widehat{m \,}^\eps_3(y_3) \ \widehat{\bs m \,}^\eps(z_3) \cdot \widetilde{\bs n}(\bs z) }{ \eps^2|\bs y-\bs z| }  \nonumber \\
&= \frac{1}{2} J^\eps_1(\widehat{\bs m \,}^\eps)+\frac{1}{2}J^\eps_2(\widehat{\bs m \,}^\eps)+ J^\eps_3(\widehat{\bs m \,}^\eps). \label{J1J2J3}
\end{align}
Note that $ J^\eps_1, J^\eps_2 $ and $ J^\eps_3 $ respectively represent that ``Bulk-Bulk", the ``Boundary-Boundary" and the ``Bulk-Boundary" terms of the magnetostatic energy. The body $\Omega_\eps=\omega_\eps \times (0,1) $ and the boundary $\partial \Omega_\eps$ can be decomposed as as $\partial \Omega_\eps= \big\{ \, \partial \omega_\eps \times (0,1) \, \big\} \bigcup \omega_\eps(y_3=0) \bigcup \omega_\eps(y_3=1)$.
\subsection{Estimates of $J^\eps_1(\widehat{\bs m \,}^\eps) \, ,J^\eps_2(\widehat{\bs m \,}^\eps) \, ,$ and $J^\eps_3(\widehat{\bs m \,}^\eps)$} \label{sec-A2}

The magnetostatic estimates in this section are inspired by similiar estimates in other works like \cite{kohn2005effective} and \cite{carbou2001thin}. 
We use the following integral inequality in this section: for arbitrary $a \neq b \in \mathbb R $ and $ q, L \in \mathbb R$ using the fact that $ q (q^2+L^2)^{-1/2} \le 1$ we have
\begin{align}
\int^b_a \frac{dq}{ \big\{ L^2+q^2 \big\}^{3/2} }=\frac{1}{L^2} \frac{q}{ ( L^2+q^2 )^{1/2} } \bigg|^b_a
= \frac{1}{L^2} \Big( \frac{b}{ ( L^2+b^2 )^{1/2} } - \frac{a}{ ( L^2+a^2 )^{1/2} } \Big)\le \frac{2}{L^2}. \label{sinetan} \end{align} 
We also need an estimate of the following term, where we use the change of variable $ \bs w_p = \bs y_p - \bs z_p $, $\bs d \bs w_p = \bs d \bs y_p$ to get, $ \ \big( \, \mbox{Recall } \omega_\eps \mbox{ is a ball of radius } \eps \mbox{ in 2-d } \Big) $
\begin{align}
\int_{\omega_\eps} \int_{\omega_\eps} \frac{\bs d \bs y_p \ \bs d \bs z_p}{|\bs y_p-\bs z_p|} = \int_{\omega_\eps} \bs d \bs z_p \int_{\omega_\eps - \bs z_p} \frac{\bs d \bs w_p}{|\bs w_p|} & \le \int_{\omega_\eps} \bs d \bs z_p \int_{\omega_{3\eps}} \frac{\bs d \bs w_p}{|\bs w_p|}
\nonumber \\
&=\int_{\omega_\eps} \bs d \bs z_p \int^{2\pi}_0 \int^{3\eps}_0 \frac{ |\bs w_p| \ \bs d( |\bs w_p|) \ d \theta}{|\bs w_p|} \nonumber \\
&= ( \, \pi \eps^2 \,) \  ( \, 2 \pi \, ) \ ( \, 3 \eps \, ) = 6 \pi^2 \eps^3, \label{radialformula} \end{align}
where we have used the fact that $ (\omega_\eps - \bs z_p) \subset \omega_{3 \eps} $ for $ \bs z_p \in \omega_\eps $.
Henceforth we drop the $ \bm y $ superscript on the derivative operator. 

Also note that if $ \widehat{\bs m \,}^\eps \in H^1(0,1)$, Sobolev embedding gives along with \eqref{widetildemepsmeps1}
\begin{align}
\sup_{y_3} \ \big|\widehat{\bs m \,}^\eps(y_3) \, \big| \ \le \ D_3 \, \Vert \widehat{\bs m \,}^\eps \Vert_{H^1(0,1)} \ \le \ \frac{D_3}{|\omega|}\big( \ \Vert \bs m^\eps \Vert_{L^2(\Omega)}+\big\Vert \partial_3 \bs m^\eps \big\Vert_{L^2(\Omega)}\big).
\label{sobolevembed}\end{align} 
\newtheorem{prop1}{Proposition}[section]
\begin{prop1}
\label{proposition3.1.1}
\begin{align} \big| \, J^\eps_1(\widehat{\bs m \,}^\eps) \, \big| \ \le \ D_4 \eps \big( \, \big\Vert \bs m^\eps \big\Vert^2_{L^2(\Omega)}+ \big\Vert \partial_3 \bs m^\eps \big\Vert^2_{L^2(\Omega)} \, \big). \nonumber \end{align}  
\end{prop1}
\begin{proof} 
Recalling definition of $J^\eps_1 $ from Equation \eqref{J1J2J3} and noting $ |\bs y_p-\bs z_p| \le |\bs y-\bs z| $ we have
\begin{align}
\big|\eps^2 J^\eps_1(\widehat{\bs m \,})\, \big| &\le \int_{\Omega_\eps} \int_{\Omega_\eps} \hspace{-2mm}\frac{|\partial^{\bm y}_3 \widehat{m \,}^\eps_3(y_3) \ \partial^{\bm y}_3 \widehat{m \,}^\eps_3( z_3)| }{ |\bs y-\bs z| } \bs d \bs y \bs d \bs z \le \int^1_0 \int^1_0 \int_{\omega_\eps} \int_{\omega_\eps} \hspace{-2mm} \frac{ |\partial^{\bm y}_3 \widehat{m \,}^\eps_3(y_3) \ \partial^{\bm y}_3 \widehat{m \,}^\eps_3(z_3)| }{|\bs y_p-\bs z_p|} 
\nonumber \\ 
&=\int^1_0 \int^1_0 \big| \partial^{\bm y}_3 \widehat{m \,}^\eps_3(y_3) \ \partial^{\bm y}_3 \widehat{m \,}^\eps_3( z_3) \big| \int_{\omega_\eps} \int_{\omega_\eps} \frac{\bs d \bs y_p \bs d \bs z_p}{|\bs y_p-\bs z_p|} \le D_4 \eps^3 \big\Vert \partial^{\bm y}_3 \widehat{m \,}^\eps_3 \big\Vert^2_{L^2(0,1)} \nonumber  
\end{align}
where we have used H$\ddot{\mbox{o}}$lder's inequality on the term $\mbox{\fontsize{9}{8}\selectfont $\displaystyle{ \int^1_0 \int^1_0 } $ } \big| \partial^{\bm y}_3 \widehat{m \,}_3(y_3) \ \partial^{\bm y}_3 \widehat{m \,}_3( z_3) \big| dy_3 dz_3$ and equation \eqref{radialformula} in the last step. Using equation \eqref{widetildemepsmeps1} we get our result.
\end{proof}
\begin{prop1}
\label{proposition3.1.2}
\begin{align} \big| \, J^\eps_3(\widehat{\bs m \,}^\eps) \, \big| \le  D_5 \eps \big( \, \big\Vert \bs m^\eps \big\Vert^2_{L^2(\Omega)}+ \big\Vert \partial_3 \bs m^\eps \big\Vert^2_{L^2(\Omega)} \, \big). \nonumber\end{align}  
\end{prop1}
\begin{proof} Recalling definition of $J^\eps_1 $ from Equation \eqref{J1J2J3}, we split of $ J^\eps_3(\widehat{\bs m \,}^\eps) $ into 2 parts,
\begin{align}
-\eps^2J^\eps_3(\widehat{\bs m \,}) &= \int_{\omega_\eps} \int^1_0 \int_{\partial \omega_\eps} \int^1_0 \frac{ \widehat{\bs m \,}^\eps(z_3) \cdot \widetilde{\bs n}(\bs z_p) }{ |\bs y-\bs z| }\partial^{\bm y}_3 \widehat{m \,}^\eps_3(y_3) + \int_{\omega_\eps} \int^1_0 \partial^{\bm y}_3 \widehat{m \,}^\eps_3(y_3) \ \times  \nonumber \\
& \hspace{10mm} \Big[ \int_{\omega_\eps(0)} \frac{ \widehat{\bs m \,}^\eps(z_3=0) \cdot \widetilde{\bs n}(\bs z) }{ |\bs y-\bs z| } + \int_{\omega_\eps(1)} \frac{ \widehat{\bs m \,}^\eps(z_3=1) \cdot \widetilde{\bs n}(\bs z) }{ |\bs y-\bs z| } \Big] =: \eps^2 J^\eps_{31}+ \eps^2 J^\eps_{32}  \nonumber
\end{align}
with $\eps^2 J^\eps_{31}(\widehat{\bs m \,}^\eps)$ being first term and $\eps^2 J^\eps_{32}(\widehat{\bs m \,}^\eps)$ is the remaining term of the R.H.S. For $ J^\eps_{31} $ using divergence theorem on $ \partial \omega_\eps(z_3) $ gives,
\begin{align} \eps^2 J^\eps_{31}(\widehat{\bs m \,}^\eps) &= \int_{\omega_\eps} \int^1_0 \int_{\partial \omega_\eps} \int^1_0 \frac{ \widehat{\bs m \,}^\eps(z_3) \cdot \widetilde{\bs n}(\bs z_p) }{ |\bs y-\bs z| } \partial^{\bm y}_3 \widehat{m \,}^\eps_3(y_3) \, \bs d \bs y_p \bs d \bs \sigma(\bs z_p) d y_3 d z_3 \nonumber \\
&= \int_{\omega_\eps} \int^1_0 \int^1_0 \partial^{\bm y}_3 \widehat{m \,}^\eps_3(y_3) \ \bs d \bs y_p d y_3 d z_3 \int_{\omega_\eps} \nabla^{\bm z}_p \cdot \Big( \ \frac{ \widehat{\bs m \,}^\eps(z_3) }{ |\bs y-\bs z| } \ \Big) \ \bs d \bs z_p \nonumber \\
&= \int_{\omega_\eps} \int^1_0 \int^1_0 \int_{\omega_\eps} \partial^{\bm y}_3 \widehat{m \,}^\eps_3(y_3) \frac{\widehat{\bs m \,}^\eps(z_3) \cdot (\bs y_p-\bs z_p) }{ \big\{ \, |\bs y_p-\bs z_p|^2+(y_3-z_3)^2 \ \big\}^{3/2} } \bs d \bs y \bs d \bs z. \label{divergence}
\end{align}
Setting $ q =(z_3-y_3) $ and $ dz_3=dq$ gives,
\begin{align}
\big|\eps^2 J^\eps_{31}(\widehat{\bs m \,}^\eps)\big| &\le \sup_{z_3} |\widehat{\bs m \,}^\eps(z_3)| \, \int_{\omega_\eps} \int_{\omega_\eps} \int^1_0 \big| \partial^{\bm y}_3 \widehat{m \,}^\eps_3(y_3) \big| \ \int^{1-y_3}_{-y_3} \hspace{-2mm} \frac{|\, \bs y_p-\bs z_p \,| }{ \big\{ \, |\bs y_p-\bs z_p|^2+q^2 \ \big\}^{3/2} } dq. \nonumber
\end{align}
Using equation \eqref{sinetan} on the inner integral gives
\begin{align}
\big|\eps^2 J^\eps_{31}(\widehat{\bs m \,}^\eps)\big|&\le 2 \sup_{z_3} \big| \widehat{\bs m \,}^\eps(z_3) \big| \ \int_{\omega_\eps} \int_{\omega_\eps} \int^1_0  \frac{ | \partial^{\bm y}_3 \widehat{m \,}^\eps_3(y_3)| }{ \ |\bs y_p-\bs z_p| \ } \nonumber \\
&= 2 \sup_{z_3} \big| \widehat{\bs m \,}^\eps(z_3) \big| \ \Bigg\{ \, \int^1_0 | \partial^{\bm y}_3 \widehat{m \,}^\eps_3(y_3) | \ d y_3 \Bigg\} \ \Bigg\{ \ \int_{\omega_\eps} \int_{\omega_\eps} \frac{ \bs d \bs y_p \bs d \bs z_p }{ | \bs y_p-\bs z_p |} \ \Bigg\} \nonumber \\
&\le D_6 \eps^3 \ \sup_{z_3} \big| \widehat{\bs m \,}^\eps(z_3) \big| \ \big\Vert \partial^{\bm y}_3 \widehat{\bs m \,}^\eps \big\Vert_{L^2(0,1)} \le D_7 \eps^3 \big( \big\Vert \bs m^\eps \big\Vert^2_{L^2(\Omega)}+ \big\Vert \partial_3 \bs m^\eps \big\Vert^2_{L^2(\Omega)}\big), \nonumber \end{align} 
using equations \eqref{radialformula}, \eqref{sobolevembed} and \eqref{widetildemepsmeps1}. Also we estimate $J^\eps_{32}(\widehat{\bs m\,}^\eps)$ as
\begin{align} 
\eps^2 \big|J^\eps_{32}\big| &\le \big( |\widehat{m \,}^\eps_3(0)|+|\widehat{m \,}^\eps_3(1)| \big) \int_{\omega_\eps} \int^1_0 \int_{\omega_\eps} \Bigg[ \frac{ | \partial^{\bm y}_3 \widehat{m \,}^\eps_3(y_3) | }{ \sqrt{ |\bs y_p - \bs z_p|^2 + y_3 ^2 }} + \frac{ | \partial^{\bm y}_3 \widehat{m \,}^\eps_3(y_3) | }{ \sqrt{ |\bs y_p - \bs z_p|^2 + (1-y_3) ^2 }} \Bigg]  \nonumber \\
& \le 4 \ \sup_{z_3} \big| \widehat{\bs m \,}^\eps(z_3) \big| \ \Bigg\{ \ \int^1_0 | \partial^{\bm y}_3 \widehat{m \,}^\eps_3(y_3) | d y_3 \Bigg\} \ \Bigg\{ \ \int_{\omega_\eps}  \int_{\omega_\eps} \frac{\bs d \bs y_p \bs d \bs z_p }{ |\bs y_p - \bs z_p|} \ \Bigg\} \nonumber \\
&\le D_8 \eps^3 \sup_{z_3} \big| \widehat{\bs m \,}^\eps(z_3) \big|  \ \big\Vert \partial^{\bm y}_3 \widehat{\bs m \,}^\eps \big\Vert_{L^2(0,1)} \le D_9 \eps^3 \big( \big\Vert \bs m^\eps \big\Vert^2_{L^2(\Omega)}+ \big\Vert \partial_3 \bs m^\eps \big\Vert^2_{L^2(\Omega)}\big) \nonumber
\end{align} 
again using equations \eqref{radialformula}, \eqref{sobolevembed} and \eqref{widetildemepsmeps1}. Combining estimates for $ J^\eps_{31}(\widehat{\bs m \,}^\eps) $ and $ J^\eps_{32}(\widehat{\bs m \,}^\eps) $ we get our result.
\end{proof}
Recalling $ J^\eps_2(\widehat{\bs m \,}^\eps)$ from eqn. \eqref{J1J2J3} we write $ J^\eps_2 = J^\eps_{21} + J^\eps_{22} + J^\eps_{23} +J^\eps_{24} $ where,
\begin{align} \eps^2J^\eps_{21} &= \int_{\partial \omega_\eps} \int^1_0 \int_{\partial \omega_\eps} \int^1_0 \frac{\widehat{\bs m \,}^\eps(y_3) \cdot \widetilde{\bs n}(\bs y) \ \widehat{\bs m \,}^\eps(z_3) \cdot \widetilde{\bs n}(\bs z) }{ \big\{ \ | \bs y_p-\bs z_p |^2 + (y_3-z_3)^2 \ \big\}^{1/2} }, \nonumber \\
\eps^2J^\eps_{22} &= \int_{\omega_\eps(0)} \int_{\omega_\eps(0)} \hspace{-5mm}\frac{\widehat{\bs m \,}^\eps(0) \cdot \widetilde{\bs n}(\bs y_p) \ \widehat{\bs m \,}^\eps(0) \cdot \widetilde{\bs n}(\bs z_p) }{ | \bs y_p-\bs z_p | } + \int_{\omega_\eps(1)} \int_{\omega_\eps(1)} \hspace{-5mm} \frac{\widehat{\bs m \,}^\eps(1) \cdot \widetilde{\bs n}(\bs y_p) \ \widehat{\bs m \,}^\eps(1) \cdot \widetilde{\bs n}(\bs z_p) }{ | \bs y_p-\bs z_p |}, \nonumber \\
\eps^2J^\eps_{23} &= 2\int_{\omega_\eps(0)} \int_{\omega_\eps(1)} \frac{\widehat{\bs m \,}^\eps(0) \cdot \widetilde{\bs n}(\bs y_p) \ \widehat{\bs m \,}^\eps(1) \cdot \widetilde{\bs n}(\bs z_p) }{ \big\{ \ | \bs y_p-\bs z_p |^2+1 \ \big\}^{1/2} } \mbox{ \ \ and, }  \nonumber \\
\frac{\eps^2J^\eps_{24}}{2} &= \int_{\partial \omega_\eps} \int^1_0 \Bigg[ \int_{\omega_\eps(0)} \hspace{-5mm}\frac{\widehat{\bs m \,}^\eps(y_3) \cdot \widetilde{\bs n}(\bs y) \ \widehat{\bs m \,}^\eps(0) \cdot \widetilde{\bs n}(\bs z_p) }{ \big\{ \ | \bs y_p-\bs z_p |^2 + y_3^2 \ \big\}^{1/2} } + \int_{\omega_\eps(1)} \hspace{-1mm} \frac{\widehat{\bs m \,}^\eps(y_3) \cdot \widetilde{\bs n}(\bs y) \ \widehat{\bs m \,}^\eps(1) \cdot \widetilde{\bs n}(\bs z_p) }{ \big\{ \, | \bs y_p-\bs z_p |^2 + (1-y_3)^2 \, \big\}^{1/2} } \Bigg] \nonumber.
\end{align}
Noting that $ |\widehat{\bs m \,}^\eps(t) \cdot \widetilde{\bs n}(\bs z)|=|\widehat{m \,}_3^\eps(t)| \le \sup_{z_3} |\widehat{\bs m \,}^\eps(z_3)|$ for $t=0 $ and $t=1 $ we have using equations \eqref{radialformula} and \eqref{sobolevembed},
\begin{align} \hspace{-5mm}
\eps^2J^\eps_{22}(\widehat{\bs m \,}^\eps) &= \sup_{z_3} \big|\widehat{\bs m \,}^\eps(z_3) \big|^2 \int_{\omega_\eps} \int_{\omega_\eps} \hspace{-1mm}\frac{\bs d \bs y_p \ \bs d \bs z_p}{ |\bs y_p - \bs z_p |} \le D_{10}\eps^3 \big( \big\Vert \bs m^\eps \big\Vert^2_{L^2(\Omega)}+\big\Vert \partial_3 \bs m^\eps \big\Vert^2_{L^2(\Omega)} \big). \label{J22form} \end{align} \begin{align}
\eps^2J^\eps_{24}(\widehat{\bs m \,}^\eps) &= 2\int_{\partial \omega_\eps} \int^1_0 \int_{\omega_\eps} \frac{-\widehat{m \,}^\eps_3(0) \ \widehat{\bs m \,}^\eps(y_3) \cdot \widetilde{\bs n}(\bs y) }{ \big\{ \ | \bs y_p-\bs z_p |^2 + y_3^2 \ \big\}^{1/2} } + \frac{\widehat{m \,}^\eps_3(1) \ \widehat{\bs m \,}^\eps(y_3) \cdot \widetilde{\bs n}(\bs y) }{ \big\{ \, | \bs y_p-\bs z_p |^2 + (1-y_3)^2 \, \big\}^{1/2} } . \label{J24}
\end{align}
Note $ \big( \ | \bs y_p-\bs z_p |^2 + 1 \ \big)^{-\frac{1}{2}} \le 1$. Then eqn. \eqref{sobolevembed} gives
\begin{align} \Big| \eps^2\frac{J^\eps_{23}}{2}(\widehat{\bs m \,}^\eps) \Big| \ & \le \ \big| \, \widehat{m \,}^\eps_3(0) \widehat{m \,}^\eps_3(1) \, \big| \ \Bigg\{ \ \int_{\omega_\eps} \int_{\omega_\eps} \bs d \bs y_p \bs d \bs z_p \Bigg\} = \pi^2 \eps^4 \sup_{z_3} \big|\widehat{\bs m \,}^\eps(z_3) \big|^2\nonumber \\
& \le \ D_{11} \eps^4 \big( \ \big\Vert \bs m^\eps \big\Vert^2_{L^2(\Omega)}+\big\Vert \partial_3 \bs m^\eps \big\Vert^2_{L^2(\Omega)} \big). \label{J23form}
\end{align}   
\begin{prop1}
\label{proposition3.1.3}
\begin{align} \big| \, J^\eps_{24}(\widehat{\bs m \,}^\eps) \, \big| \le D_{12} \eps \big( \, \big\Vert \bs m^\eps \big\Vert^2_{L^2(\Omega)}+\big\Vert \partial_3 \bs m^\eps \big\Vert^2_{L^2(\Omega)} \, \big). \nonumber \end{align}
\end{prop1}
\begin{proof} As for term the $J^\eps_{31}$ in Proposition \ref{proposition3.1.2}, first using divergence theorem in $J^\eps_{24}$ from \eqref{J24} on $ \partial \omega_\eps(y_3) $ we get
\begin{align}
 \eps^2J^\eps_{24}(\widehat{\bs m \,}^\eps) &= 2 \int_{\omega_\eps} \int_{\omega_\eps} \int^1_0 \Bigg\{ \ \frac{\widehat{\bs m \,}^\eps(y_3) \cdot (\bs y_p-\bs z_p) \ \widehat{m \,}^\eps_3(0) }{ \big\{ \ | \bs y_p-\bs z_p |^2 + y_3^2 \big\}^{3/2} \ } - \frac{\widehat{\bs m \,}^\eps(y_3) \cdot (\bs y_p-\bs z_p) \ \widehat{m \,}^\eps_3(1) }{ \big\{ \ | \bs y_p-\bs z_p |^2 + (1-y_3)^2 \big\}^{3/2} \ } \Bigg\}  \nonumber.
\end{align}
Then using $ \widehat{\bs m \,}^\eps(y_3) \cdot (\bs y_p-\bs z_p) \le |\bs y_p-\bs z_p| \ \sup_{y_3} |\widehat{\bs m \,}^\eps(y_3)|$ and \eqref{sinetan} we get,
\begin{align}
\big| \eps^2 J^\eps_{24} \big| & \le 2 \sup_{y_3} |\widehat{\bs m \,}^\eps(y_3)|^2 \ \int_{\omega_\eps} \int_{\omega_\eps} \Bigg| \int^1_0 \hspace{-2mm} \frac{|\bs y_p-\bs z_p|\ d y_3}{ \big\{| \bs y_p-\bs z_p |^2 + y_3^2 \big\}^{\frac{3}{2}} } - \int^1_0 \hspace{-2mm} \frac{ |\bs y_p-\bs z_p| \ d (1-y_3) }{ \big\{| \bs y_p-\bs z_p |^2 + (1-y_3)^2 \big\}^{\frac{3}{2}} } \Bigg| \nonumber \\
&\le 8 \sup_{y_3} |\widehat{\bs m \,}^\eps(y_3)|^2 \ \int_{\omega_\eps} \int_{\omega_\eps} \frac{\bs d \bs y_p \bs d \bs z_p}{ |\bs y_p-\bs z_p| } = D_{12} \eps^3 \big( \big\Vert \bs m^\eps \big\Vert^2_{L^2(\Omega)}+\big\Vert \partial_3 \bs m^\eps \big\Vert^2_{L^2(\Omega)} \big).  \nonumber
\end{align}
and eqns. \eqref{sobolevembed} and \eqref{radialformula} in the last step.
\end{proof}
We will now show that $ J^\eps_{21}(\widehat{\bs m \,}^\eps)$ is the largest term in the magnetostatic terms. It contributes energy of $ O(1) $ which appears in the first limit problem $ \mathcal I_0 $. We split $ J^\eps_{21}(\widehat{\bs m \,}^\eps) $ as follows:
\begin{align}
 J^\eps_{21}(\widehat{\bs m \,}^\eps) &= \int_{\partial \omega_\eps} \int^1_0 \int_{\partial \omega_\eps} \int^1_0 \frac{\widehat{\bs m \,}^\eps(y_3) \cdot \widetilde{\bs n}(\bs y_p) \ \widehat{\bs m \,}^\eps(z_3) \cdot \widetilde{\bs n}(\bs z_p) }{\eps^2|\bs y-\bs z|} \nonumber \\
&= \int_{\partial \omega_\eps} \int^1_0 \int_{\partial \omega_\eps} \int^1_0 \frac{\widehat{\bs m \,}^\eps(y_3) \cdot \widetilde{\bs n}(\bs y_p)\widehat{\bs m \,}^\eps(y_3) \cdot \widetilde{\bs n}(\bs z_p)}{\eps^2|\bs y-\bs z|} \nonumber \\
& \hspace{14mm} -\int_{\partial \omega_\eps} \int^1_0 \int_{\partial \omega_\eps} \int^1_0 \frac{\widehat{\bs m \,}^\eps(y_3) \cdot \widetilde{\bs n}(\bs y_p) \ ( \widehat{\bs m \,}^\eps(y_3) - \widehat{\bs m \,}^\eps(z_3) ) \cdot \widetilde{\bs n}(\bs z_p) }{\eps^2 |\bs y-\bs z|}\nonumber \\
&= J^\eps_{211}(\widehat{\bs m \,}^\eps)+J^\eps_{212}(\widehat{\bs m \,}^\eps). \nonumber
\end{align}
Next we show the following proposition. 
\begin{prop1}
\label{proposition3.1.4} 
\begin{align} \big| \, J^\eps_{212}(\widehat{\bs m \,}^\eps) \, \big| \le D_{13} \eps^{3/4} \ \big( \ \big\Vert \bs m^\eps \big\Vert^2_{L^2(\Omega)}+\big\Vert \partial_3 \bs m^\eps \big\Vert^2_{L^2(\Omega)} \big). \nonumber \end{align}  
\end{prop1}
\begin{proof}
Using Divergence theorem in $\bs y_p$ variable as in \eqref{divergence} and Fubini's theorem we get,
\begin{align} \eps^2 J^\eps_{212} &=\int_{\partial \omega_\eps} \int^1_0 \int^1_0 \int_{\partial \omega_\eps} \frac{ \widehat{\bs m \,}^\eps(y_3) \cdot \widetilde{\bs n}(\bs y_p)  }{\sqrt{ |\bs y_p -\bs z_p|^2 + (y_3-z_3)^2}} \ ( \widehat{\bs m \,}^\eps(z_3) - \widehat{\bs m \,}^\eps(y_3))  \cdot \widetilde{\bs n}( \bs z_p) \bs d \bs \sigma(\bs y_p)  \nonumber \\
&= \int_{\partial \omega_\eps} \int^1_0 \int^1_0 \int_{\omega_\eps} \frac{\widehat{\bs m \,}^\eps(y_3) \cdot (\bs z_p-\bs y_p ) }{ \big\{ \, |\bs y_p -\bs z_p|^2 + (y_3-z_3)^2 \ \big\}^{3/2} } \ ( \widehat{\bs m \,}^\eps(z_3) - \widehat{\bs m \,}^\eps(y_3)) \cdot \widetilde{\bs n}( \bs z_p)  \, \bs d \bs y_p \nonumber \\
&= \int_{\partial \omega_\eps} \int_{\omega_\eps} \widetilde{\bs n}( \bs z_p) \cdot \int^1_0 \widehat{\bs m \,}^\eps(y_3) \cdot ( \bs z_p-\bs y_p ) \int^1_0  \frac{ \widehat{\bs m \,}^\eps(z_3) - \widehat{\bs m \,}^\eps(y_3) }{ \big\{ \, |\bs y_p -\bs z_p|^2 + (y_3-z_3)^2 \ \big\}^{3/2} } \, d z_3. \nonumber 
\end{align}
Now note that $\frac{|\bs y_p - \bs z_p|}{\big\{ \, |\bs y_p -\bs z_p|^2 + (y_3-z_3)^2 \ \big\}^{1/2}} \le 1 $ and $ |\widetilde{\bs n}( \bs z_p)|=1$ which gives
\begin{align}
\big|\eps^2 J^\eps_{212}(\widehat{\bs m \,}^\eps(y_3)) \big| & \le \sup_{y_3} \big|\widehat{\bs m \,}^\eps(y_3)\big| \ \ \Bigg\{ \int_{\partial \omega_\eps} \int_{\omega_\eps} \int^1_0 \int^1_0  \frac{ |\widehat{\bs m \,}^\eps(z_3) - \widehat{\bs m \,}^\eps(y_3)| }{ \big\{ \, |\bs y_p -\bs z_p|^2 + (y_3-z_3)^2 \ \big\} } \, d z_3  \Bigg\} . \label{sobolevpartial}
\end{align}
Note that,
\begin{align} \frac{1}{ \big\{ \, |\bs y_p -\bs z_p|^2 + (y_3-z_3)^2 \ \big\} } \le \frac{1}{|\bs y_p -\bs z_p|^{1/4} } \frac{1}{|y_3-z_3|^{7/4} } 
\end{align}
Then 
\begin{align}
\int^1_0 \int^1_0  \frac{ |\widehat{\bs m \,}^\eps(z_3) - \widehat{\bs m \,}^\eps(y_3)| d z_3 d y_3}{ \big\{ \, |\bs y_p -\bs z_p|^2 + (y_3-z_3)^2 \ \big\} } &\le \frac{1}{ |\bs y_p -\bs z_p|^{1/4} }  \int^1_0 \int^1_0  \frac{ |\widehat{\bs m \,}^\eps(z_3) - \widehat{\bs m \,}^\eps(y_3)| }{|y_3-z_3|^{7/4} } \, d z_3 d y_3 \nonumber \\
&\le \frac{1}{ |\bs y_p -\bs z_p|^{1/4} } \big\Vert \partial^{\bm y}_3 \widehat{\bs m \,}^\eps(y_3) \big\Vert_{L^1(0,1)} \label{sobolevpartial1}
\end{align}
because of the fact that $\mbox{\fontsize{9}{8}\selectfont $\displaystyle{ \int^1_0 \int^1_0 } $ } \frac{ |\widehat{\bs m \,}^\eps(z_3) - \widehat{\bs m \,}^\eps(y_3)| }{|y_3-z_3|^{7/4} } \, d z_3 d y_3 $ denotes the  seminorm in the fractional Sobolev space $W^{\frac{3}{4},1}(0,1)$ and by the continuous embedding of $W^{1,1}(0,1) \subset W^{\frac{3}{4},1}(0,1)$. The integral above cannot be bounded by norm in $W^{1,1}$ alone unless $\widehat{\bs m \,}^\eps$ is a constant, which is shown by the surprising result Proposition 1 in \cite{brézis2002recognize}. Also note using H$\ddot{\mbox{o}}$lder inequality
\begin{align}
\big\Vert \partial^{\bm y}_3 \widehat{\bs m \,}^\eps\big\Vert_{L^1(0,1)} = \int^1_0 \big| \partial^{\bm y}_3 \widehat{\bs m \,}^\eps \big| \chi_{(0,1)} dy_3 \le \big\Vert \chi_{(0,1)} \big\Vert_{L^2(0,1)} \big\Vert \partial^{\bm y}_3 \widehat{\bs m \,}^\eps\big\Vert_{L^2(0,1)} =\big\Vert \partial^{\bm y}_3 \widehat{\bs m \,}^\eps\big\Vert_{L^2(0,1)}. \nonumber
\end{align} 
Then using \eqref{sobolevpartial1} in eqn. \eqref{sobolevpartial} along with the above result we get,
\begin{align}
\big|\eps^2 J^\eps_{212}(\widehat{\bs m \,}^\eps(y_3)) & \le  \sup_{y_3} \big| \widehat{\bs m \,}^\eps(y_3) \big| \ \ \Bigg\{ \int_{\partial \omega_\eps} \int_{\omega_\eps} \int^1_0 \int^1_0  \frac{ |\widehat{\bs m \,}^\eps(z_3) - \widehat{\bs m \,}^\eps(y_3)| }{ \big\{ \, |\bs y_p -\bs z_p|^2 + (y_3-z_3)^2 \ \big\} } \, d z_3 \Bigg\} \nonumber \\
& \le \sup_{y_3} \big| \widehat{\bs m \,}^\eps(y_3) \big|  \ . \ \big\Vert \partial^{\bm y}_3 \widehat{\bs m \,}^\eps\big\Vert_{L^2(0,1)} \ \ \Bigg\{ \int_{\partial \omega_\eps} \int_{\omega_\eps}\frac{\bs d \bs y_p \bs d \bs \sigma(\bs z_p)}{ |\bs y_p -\bs z_p|^{1/4} } \Bigg\}  \nonumber \\
&= D_{14} \ \eps^2 \ \eps^{3/4} \ \, \big( \ \big\Vert \bs m^\eps \big\Vert^2_{L^2(\Omega)}+\big\Vert \partial_3 \bs m^\eps \big\Vert^2_{L^2(\Omega)} \big)\nonumber
\end{align}
using calculation like in eqn. \eqref{radialformula} to get the $\eps^2 \eps^{3/4}$ term and eqns. \eqref{sobolevembed} and \eqref{widetildemepsmeps1}. 
\end{proof}
In 2-dimensional micromagnetics on a domain $ \Psi \in \mathbb R^2 $ for a constant magnetization $ \bs m \in H^1(\Psi,m_sS^2),$ the demagnetization field is given by,
\begin{align}
\bs h_{\bs m}(\bs x) = \int_{\partial \Psi} \frac{\bs x - \bs y}{|\bs x - \bs y|^2} \bs m \cdot \bs n(\bs y) \, \bs d \bs y \label{2Dfield}
\end{align}   
and magnetostatic energy is given by,
\begin{align}
\mathcal E_{2d} = \int_{\Psi} \int_{\partial \Psi}  \ \bs m \cdot \frac{\bs x - \bs y}{|\bs x - \bs y|^2} \  \ \bs m \cdot \bs n(\bs y) \, \bs d \bs y \label{2Denergy}.
\end{align}   
\begin{prop1}
\label{proposition3.1.5} 
\begin{align} \Big| \ J^\eps_{211}(\widehat{\bs m \,}^\eps) -2\pi |\omega_\eps| \int^1_0 \big|\widehat{\bs m \,}^\eps_p(y_3) \big|^2 dy_3 \ \Big| \le D_{15} \ \eps \big( \big\Vert \bs m^\eps \big\Vert^2_{L^2(\Omega)}+\big\Vert \partial_3 \bs m^\eps \big\Vert^2_{L^2(\Omega)} \big). \nonumber \end{align}
\end{prop1}
\begin{proof}
Using the Divergence theorem on $ \bs z_p $ as in \eqref{divergence} and a subsequent change of variables $ q(z_3)=z_3-y_3 $, followed by \eqref{sinetan} (as in Proposition \ref{proposition3.1.2}) we get
\begin{align} \hspace{-10mm} \int_{\partial \omega_\eps} \int^1_0 \int^1_0 & \int_{\partial \omega_\eps} \hspace{-3mm} \frac{\widehat{\bs m \,}^\eps(y_3) \cdot \widetilde{\bs n}(\bs y_p)\widehat{\bs m \,}^\eps(y_3) \cdot \widetilde{\bs n}(\bs z_p)}{\sqrt{ |\bs y_p-\bs z_p|^2 + (y_3-z_3)^2}}=\int_{\partial \omega_\eps} \int^1_0 \int^1_0 \int_{\omega_\eps} \hspace{-2mm} \frac{\widehat{\bs m \,}^\eps(y_3) \cdot \widetilde{\bs n}(\bs y_p)\widehat{\bs m \,}^\eps(y_3) \cdot (\bs y_p-\bs z_p)}{\big\{ \ |\bs y_p-\bs z_p|^2 + (y_3-z_3)^2 \ \big\}^{3/2} } \nonumber \\
&= \int_{\partial \omega_\eps} \int^1_0 \int_{\omega_\eps} \widehat{\bs m \,}^\eps(y_3) \cdot \widetilde{\bs n}(\bs y_p) \int^{1-y_3}_{-y_3} \frac{ \widehat{\bs m \,}^\eps(y_3) \cdot (\bs y_p-\bs z_p) }{\big\{ \ |\bs y_p-\bs z_p|^2 + q^2 \ \big\}^{3/2}  } d q \nonumber \\
&= \int_{\partial \omega_\eps} \int^1_0 \int_{\omega_\eps} \frac{ \widehat{\bs m \,}^\eps(y_3) \cdot (\bs y_p-\bs z_p) }{ |\ \bs y_p-\bs z_p \ |^2 } \ \Bigg\{ \ \frac{ y_3 \ \widehat{\bs m \,}^\eps(y_3) \cdot \widetilde{\bs n}(\bs y_p) }{\sqrt{ y_3^2+|\bs y_p-\bs z_p|^2}} + \frac{ (1-y_3) \ \widehat{\bs m \,}^\eps(y_3) \cdot \widetilde{\bs n}(\bs y_p) }{\sqrt{(1-y_3)^2+|\bs y_p-\bs z_p|^2}} \ \Bigg\} \nonumber \\
& =: \eps^2 \big( \ J^\eps_{2111}(\widehat{\bs m \,}^\eps)+J^\eps_{2112}(\widehat{\bs m \,}^\eps) \ \big) \nonumber
\end{align}
where $\eps^2 J^\eps_{2111}(\widehat{\bs m \,}^\eps)=\mbox{\fontsize{9}{8}\selectfont $\displaystyle{ \int_{\partial \omega_\eps} \int^1_0 \int_{\omega_\eps} } $ } \frac{ \widehat{\bs m \,}^\eps(y_3) \cdot (\bs y_p-\bs z_p) }{ |\ \bs y_p-\bs z_p \ |^2 } \ \ \frac{ y_3 \ \widehat{\bs m \,}^\eps(y_3) \cdot \widetilde{\bs n}(\bs y_p) }{\sqrt{ y_3^2+|\bs y_p-\bs z_p|^2}} $ and $\eps^2 J^\eps_{2111}(\widehat{\bs m \,}^\eps)$ the remaining term. Set $ J^\eps_0(\widehat{\bs m \,}^\eps)$ as
\begin{align}
J^\eps_0(\widehat{\bs m \,}^\eps) &:= \int^1_0 \int_{\partial \omega_\eps} \int_{\omega_\eps}  \widehat{\bs m \,}^\eps(y_3) \cdot \widetilde{\bs n}(\bs y_p) \frac{ \widehat{\bs m \,}^\eps(y_3) \cdot (\bs y_p-\bs z_p) }{ |\bs y_p-\bs z_p|^2  }  \bs d \bs \sigma(\bs y_p) \bs d \bs z_p d y_3, \nonumber 
\end{align} 
Let $ \ R :=\max \ 2 \eps^{-1} \, |\bs x_p - \bs y_p|, \  \big( \, \bs z_p \in \omega_\eps, \bs y_p \in \partial \omega_\eps \, \big)$, and 
note 
\begin{align}
1-\frac{y_3}{ \sqrt{ \, y_3^2+|\bs y_p-\bs z_p|^2}} \ \le \ \  \begin{cases} \frac{|\bs y_p-\bs z_p|^2}{2y_3^2}, \quad &\mbox{for $y_3 \ge R \eps$ }, \\
1 & \mbox{for $y_3 \le R \eps$ }. \end{cases}
\end{align}
Noting that $|\widetilde{\bs n}|=1$, $\ |\widehat{\bs m \,}^\eps(y_3) \cdot \widetilde{\bs n}(\bs y_p)| \le |\widehat{\bs m \,}^\eps(y_3)|$ and $|\widehat{\bs m \,}^\eps(y_3) \cdot (\bs y_p-\bs z_p)| \le |\widehat{\bs m \,}^\eps(y_3)| \, |\bs y_p-\bs z_p|$,
\begin{align}
\int^1_0 & \frac{ \widehat{\bs m \,}^\eps(y_3) \cdot (\bs y_p-\bs z_p) }{ |\ \bs y_p-\bs z_p \ |^2 } \widehat{\bs m \,}^\eps(y_3) \cdot \widetilde{\bs n}(\bs y_p) \Big( 1-\frac{ y_3 }{\sqrt{ y_3^2+|\bs y_p-\bs z_p|^2}}\Big) \, d y_3 \nonumber \\
&\le \int^{R \eps}_0  \frac{ |\widehat{\bs m \,}^\eps(y_3)|^2 }{ |\ \bs y_p-\bs z_p \ | } \Big| 1-\frac{ y_3 }{\sqrt{ y_3^2+|\bs y_p-\bs z_p|^2}}\Big| + \int^1_{R \eps} \frac{ |\widehat{\bs m \,}^\eps(y_3)|^2 }{ |\ \bs y_p-\bs z_p \ | } \Big| 1-\frac{ y_3 }{\sqrt{ y_3^2+|\bs y_p-\bs z_p|^2}}\Big| \nonumber \\
&\le \int^{R \eps}_0 \frac{ |\widehat{\bs m \,}^\eps(y_3)|^2 }{ |\ \bs y_p-\bs z_p \ | } d y_3 + \int^1_{R \eps} \widehat{\bs m \,}^\eps(y_3)|^2 \ \Bigg( \ \frac{ |\bs y_p-\bs z_p|  }{2y_3^2} \ \Bigg) d y_3\nonumber \\
& \le\sup_{y_3 \in (0,1)} |\widehat{\bs m \,}^\eps(y_3)|^2 \ \ \Bigg\{ \ \int^{R \eps}_0 \frac{ d y_3}{ |\ \bs y_p-\bs z_p \ | } + \int^1_{R \eps} \frac{|\bs y_p-\bs z_p|}{2y_3^2} d y_3 \ \Bigg\}. \nonumber
\end{align}
Using above result and noting that $ \partial_3 \big( y_3^{-1} \big) = -y_3^{-2}$ we get
\begin{align}
\eps^2 |J^\eps_0-J^\eps_{2111}| &\le \Bigg|\int_{\partial \omega_\eps}  \int_{\omega_\eps} \int^1_0  \frac{ \widehat{\bs m \,}^\eps(y_3) \cdot (\bs y_p-\bs z_p) }{ |\ \bs y_p-\bs z_p \ |^2 } \widehat{\bs m \,}^\eps(y_3) \cdot \widetilde{\bs n}(\bs y_p) \Big( 1-\frac{ y_3 }{\sqrt{ y_3^2+|\bs y_p-\bs z_p|^2}}\Big) \Bigg| \nonumber \\
&\le \sup_{y_3} |\widehat{\bs m \,}^\eps(y_3)|^2 \int_{\partial \omega_\eps}  \int_{\omega_\eps} \Bigg\{ \int^{R \eps}_0 \frac{d y_3}{ |\ \bs y_p-\bs z_p \ | } + \int^1_{R \eps} \frac{|\bs y_p-\bs z_p|}{2y_3^2} d y_3 \Bigg\} \nonumber \\
&\le \sup_{y_3} |\widehat{\bs m \,}^\eps(y_3)|^2 \int_{\partial \omega_\eps}  \int_{\omega_\eps} \Bigg\{ \frac{R\eps }{|\bs y_p-\bs z_p|} - \frac{|\bs y_p-\bs z_p|}{2} \, \frac{1}{y_3}  \Bigg|^1_{R\eps} \ \Bigg\} \nonumber \\ 
&\le \sup_{y_3} |\widehat{\bs m \,}^\eps(y_3)|^2 \int_{\partial \omega_\eps}  \int_{\omega_\eps} \Big\{ \frac{R\eps }{|\bs y_p-\bs z_p|} - \frac{|\bs y_p-\bs z_p|}{2} +\frac{|\bs y_p-\bs z_p|}{2R\eps} \ \Big\}\bs d \bs \sigma(\bs y_p) \bs d \bs z_p \nonumber
\end{align}
Note from equation \eqref{radialformula}, the term $ \mbox{\fontsize{8}{8}\selectfont $\displaystyle{ \int_{\omega_\eps} \frac{1}{|\bs y_p-\bs z_p|} } $ } \bs d \bs z_p = D_{15}\eps$. So the first integral above is $ R\eps \mbox{\fontsize{8}{8}\selectfont $\displaystyle{ \int_{\partial \omega_\eps} \int_{\omega_\eps} \frac{1 }{|\bs y_p-\bs z_p|} \bs d \bs z_p } $ } \approx D_{16} \eps^3$. The second integrand is $O(\eps)$ and so its integral is of $O(\eps^4) $. The third integrand is bounded by $ 1$, since by definition $R\eps \ge |\bs y_p -\bs z_p|$. So the third integral $ \mbox{\fontsize{8}{8}\selectfont $\displaystyle{ \int_{\partial \omega_\eps} \int_{\omega_\eps} \frac{|\bs y_p-\bs z_p|}{R \eps} \bs d \bs z_p } $ } \approx D_{17} \eps^3$. So $ |J^\eps_0-J^\eps_{2111}| \le D_{18} \eps \ \sup_{y_3} |\widehat{\bs m \,}^\eps(y_3)|^2 $. $J^\eps_{2112}$ can be treated the same way to give the result on using eqn \eqref{widetildemepsmeps1}
\begin{align} \big| \ J^\eps_{211}(\widehat{\bs m \,}^\eps)-2J^\eps_0(\widehat{\bs m \,}^\eps) \ \big| \le D_{18} \eps \ \sup_{y_3} |\widehat{\bs m \,}^\eps(y_3)|^2 \le D_{19} \eps \big( \big\Vert \bs m^\eps \big\Vert^2_{L^2(\Omega)}+\big\Vert \partial_3 \bs m^\eps \big\Vert^2_{L^2(\Omega)} \big). \nonumber 
\end{align} 
We get our result noting that $J^\eps_0(\widehat{\bs m \,})$ is exactly the 2-D magnetostatic energy $\mathcal E_{2d} $ defined in \eqref{2Denergy} and for a circular cross-section $\omega_\eps$ it is well know that
\begin{align}
\mathcal E_{2d}(\widehat{\bs m \,}^\eps) = J^\eps_0(\widehat{\bs m \,}^\eps)= \pi |\omega_\eps| \int^1_0 \big|\widehat{\bs m \,}^\eps_p(y_3) \big|^2 dy_3 = \eps^2 \pi |\omega| \int^1_0  \big| \widehat{\bs m \,}^\eps_p(x_3) \big|^2 dx_3. \nonumber
\end{align}
 \end{proof}
\subsection{Final Estimate for $\mathcal E^\eps_d(\bs m^\eps)$}
The excange energy of $\bs m^\eps$ is bounded by equation \eqref{ge1}, $ \frac{K_5}{d} \ > \ \eps^{-2} \big\Vert {\nabla_p \bs m^\eps} \big\Vert^2_{L^2(\Omega)} + \big\Vert \partial_3 { \bs m^\eps} \big\Vert^2_{L^2(\Omega)} $.
Using Remark \ref{RemA1} we get first,
\begin{align}
\mathcal E^\eps_d(\bs m^\eps) - \mathcal E^\eps_d(\widehat{\bs m \,}^\eps) = O(\eps). \nonumber \end{align}
Combining Propositions \ref{proposition3.1.1} , \ref{proposition3.1.2} , \ref{proposition3.1.3} , \ref{proposition3.1.4} , \ref{proposition3.1.5} and and equations \eqref{J22form} and \eqref{J23form} we get,
\begin{align}
\mathcal E^\eps_d(\widehat{\bs m \,}^\eps) - \pi |\omega| \int^1_0 \big|\widehat{\bs m \,}^\eps_p(y_3) \big|^2 dy_3
&= \big( O(\eps) + O(\eps^{3/4}) \big) \Big( \big\Vert \bs m^\eps \big\Vert^2_{L^2(\Omega)}+\big\Vert \partial_3 \bs m^\eps \big\Vert^2_{L^2(\Omega)}\Big) \nonumber. \nonumber
\end{align} 
Combining the two we get 
\begin{align}
\mathcal E^\eps_d(\bs m^\eps)- \pi |\omega| \int^1_0 \big|\widehat{\bs m \,}^\eps_p(y_3) \big|^2 dy_3=O(\eps) +O(\eps^{3/4}). \label{gammafirstdemag1}
\end{align}
\begin{RemA} \label{RemA3}
The above result \eqref{gammafirstdemag1} is true for any magnetization $ \bs m$ as long as the magnetization satisfies the exchange bound $ K_5 \ge \eps^{-2} \big\Vert \nabla_p \bs m \big\Vert^2_{L^2(\Omega)} + \big\Vert \partial_3 \bs m \big\Vert^2_{L^2(\Omega)}$. 
\end{RemA}
Let $\widetilde{\bs m \,}^o$ be a constant vector in $m_sS^2$. If $\bs m^o$ is the rescaled version of $\widetilde{\bs m \,}^o$, recall the result in equation \eqref{magnetostaticM} gives,
\begin{align}
\mathcal E^\eps_d(\bs m^o) =\eps^{-2} E_{demag}&=\pi^2 |\bs m^o_p|^2 - \eps \frac{8\pi}{3} \Big( \ |\bs m^o_p|^2 -2 |m^o_3|^2 \Big) + \pi^2 \eps^2 \Big( \frac{|\bs m^o_p|^2}{2} - |m^o_3|^2 \Big) \nonumber \\ 
&=\pi^2 |\bs m^o_p|^2+ \eps Q_1 + \eps^2 Q_2, \nonumber
\end{align}
where we define $Q_1$ and $Q_2$ as in equation \eqref{magnetostaticM}. 

\begin{prop1}
\label{proposition3.1.8}
Let $ \widetilde{\bs m \,}^o$ be a constant on $m_sS^2$ and $ H^1(\Omega_\eps;m_sS^2) \ni \widetilde{\bs m \,}^\eps=\widetilde{\bs m \,}^o + \widetilde{\bs M}^\eps $. Then the following holds in terms of the rescaled magnetizations $ \big( \bs m^\eps, \bs m^o, \bs M^\eps\big)$, 
\begin{align} \frac{d}{\eps^2} \big\Vert \nabla_p & \bs m^\eps \big\Vert_{L^2(\Omega)}+\Lambda \big( \, \big\Vert \bs M^\eps \big\Vert^2_{L^2(\Omega)}+\big\Vert \partial_3 \bs M^\eps \big\Vert^2_{L^2(\Omega)} \, \big)+\mathcal E^\eps_d(\bs m^\eps) -\pi \int_{\Omega} \big| \bs m^\eps_p \big|^2 \nonumber \\
&\ge \frac{d}{2\eps^2} \big\Vert \nabla_p \bs M^\eps \big\Vert_{L^2(\Omega)}+\frac{\Lambda}{2} \big( \, \big\Vert \bs M^\eps \big\Vert^2_{L^2(\Omega)}+\big\Vert \partial_3 \bs M^\eps \big\Vert^2_{L^2(\Omega)} \, \big)+ \eps Q_1 + \eps^2 Q_2-D_{18} \eps^2. \nonumber\end{align} 
\end{prop1}
\begin{proof} First note $\widetilde{\bs h}^\eps_{\widetilde{\bs m \,}^o}$ as given in \eqref{magnetostaticM} on rescaling to $\bs h^\eps_{\bs m^o}$ gives
\begin{align}
\bs h^\eps_{\bs m^o} = -2 \pi \begin{bmatrix} \bs m^o_p \\ 0 \end{bmatrix} + \frac{16\eps}{3} \begin{bmatrix} \bs m^o_p \\ -2m^o_3 \end{bmatrix}+ \pi \eps^2 \begin{bmatrix} \bs m^o_p \\-2 m^o_3 \end{bmatrix}. \nonumber
\end{align}
Note $\nabla_p \bs m^\eps = \nabla_p (\bs m^o + \bs M^\eps)=\nabla_p \bs M^\eps$ as $\bs m^o $ is constant. Lemma \ref{lemmaA.1.1} gives along with Young's inequality gives
\begin{align}
\frac{d}{\eps^2} \big\Vert \nabla_p \bs m^\eps \big\Vert^2_{L^2(\Omega)} + \mathcal E^\eps_d(\bs M^\eps) -\mathcal E^\eps_d(\widehat{\bs M}^\eps) &\ge \frac{d}{\eps^2} \big\Vert \nabla_p \bs M^\eps \big\Vert^2_{L^2(\Omega)} -D_0 \big\Vert \bs M^\eps \big\Vert_{L^2(\Omega)} \ \big\Vert \nabla_p \bs M^\eps \big\Vert_{L^2(\Omega)} \nonumber \\
&\ge \frac{d}{2\eps^2} \big\Vert \nabla_p \bs M^\eps \big\Vert^2_{L^2(\Omega)} - \frac{D_0\eps^2}{2d} \big\Vert \bs M^\eps \big\Vert^2_{L^2(\Omega)}.  \nonumber \end{align}
Using propositions \ref{proposition3.1.1} , \ref{proposition3.1.2} , \ref{proposition3.1.3} , \ref{proposition3.1.4} , \ref{proposition3.1.5} and equations \eqref{J22form} and \eqref{J23form} we get
\begin{align} \hspace{-6mm}
\mathcal E^\eps_d(\widehat{\bs M}^\eps)-\pi \int_{\Omega} \big| \bs M^\eps_p \big|^2 &\ge \mathcal E^\eps_d(\widehat{\bs M}^\eps)-\pi \int_{\Omega} \big| \widehat{ \bs M \,}^\eps_p \big|^2 \ge -\big( D_{19}\eps+D_{20} \eps^{3/4} \big) \Big( \big\Vert \bs M^\eps \big\Vert^2_{L^2}+\big\Vert \partial_3 \bs M^\eps \big\Vert^2_{L^2} \Big)  \nonumber \\
&\ge -D_{20} \eps^{3/4} \Big( \big\Vert \bs M^\eps \big\Vert^2_{L^2(\Omega)}+\big\Vert \partial_3 \bs M^\eps \big\Vert^2_{L^2(\Omega)} \Big).  \nonumber
\end{align}
Adding the two together we get,
\begin{align}
\frac{d}{\eps^2}\big\Vert \nabla_p \bs m^\eps \big\Vert^2_{L^2(\Omega)} & + \mathcal E^\eps_d(\bs M^\eps) -\pi \int_{\Omega} \big| \bs M^\eps_p \big|^2 \nonumber \\
&\ge \frac{d}{2\eps^2}\big\Vert \nabla_p \bs M^\eps \big\Vert^2_{L^2(\Omega)} -D_{20} \eps^{3/4} \Big( \big\Vert \bs M^\eps \big\Vert^2_{L^2(\Omega)}+\big\Vert \partial_3 \bs M^\eps \big\Vert^2_{L^2(\Omega)} \Big). \label{A19} 
\end{align}
Note by the linearity of Maxwell's equation
\begin{align}
\frac{1}{8\pi} \int_{\mathbb R^3} \big| \widetilde{\bs h}^\eps_{\widetilde{\bs m \,}^\eps} \big|^2 \bs d \bs y = \frac{1}{8\pi} \int_{\mathbb R^3} \big| \widetilde{\bs h}^\eps_{\widetilde{\bs m \,}^o} \big|^2 \bs d \bs y+\frac{1}{8\pi}\int_{\mathbb R^3} \big| \widetilde{\bs h}^\eps_{\widetilde{\bs M}^\eps} \big|^2 \bs d \bs y-\int_{\Omega_\eps} \widetilde{\bs h}^\eps_{\widetilde{\bs m \,}^o} \cdot \widetilde{\bs M}^\eps \bs d \bs y.
\nonumber \end{align}
Dividing by $\eps^2,$ noting that $\mathcal E^\eps_d(\bs m^\eps)=\frac{1}{\eps^2}\frac{1}{8\pi} \int_{\mathbb R^3} \big| \widetilde{\bs h}^\eps_{\widetilde{\bs m \,}^\eps} \big|^2 \bs d \bs y $, and rescaling $\widetilde{\bs m \,}^\eps $ \& $ \widetilde{\bs M \,}^\eps$, 
\begin{align} \hspace{-5mm}
\mathcal E^\eps_d(\bs m^\eps) & = \mathcal E^\eps_d(\bs m^o)+\mathcal E^\eps_d(\bs M^\eps)-\int_{\Omega} \bs h^\eps_{\bs m^o} \cdot \bs M^\eps \bs d \bs x \nonumber \\
&= \mathcal E^\eps_d(\bs m^o)+\mathcal E^\eps_d(\bs M^\eps)+2\pi \int_{\Omega}\bs m^o_p \cdot \bs M_p^\eps \bs d \bs x +\frac{8\eps}{3} \int_{\Omega} \Big( \, \bs m^o_p \cdot \bs M^\eps_p -2m^o_3 M^\eps_3 \, \Big) \bs d \bs x \nonumber \\
& \qquad \qquad -\eps^2 \pi \int_{\Omega} \big( \bs m^o_p \cdot \bs M^\eps_p -2 m^o_3 M^\eps_3 \big) \bs d \bs x \label{mepsMeps} \\
&\ge \mathcal E^\eps_d(\bs m^o)+\mathcal E^\eps_d(\bs M^\eps)+2\pi \int_{\Omega}\bs m^o_p \cdot \bs M_p^\eps \bs d \bs x -D_{21} \eps \big\Vert \bs m_p^o \cdot \bs M_p^\eps \big\Vert_{L^1} -D_{22}\eps \big\Vert m_3^o M_3^\eps \big\Vert_{L^1} \nonumber \\
&\ge \mathcal E^\eps_d(\bs m^o)+\mathcal E^\eps_d(\bs M^\eps)+2\pi \int_{\Omega}\bs m^o_p \cdot \bs M_p^\eps \bs d \bs x -\frac{\Lambda}{4} \big\Vert \bs M^\eps \big\Vert^2_{L^2(\Omega)} -D_{23} \eps^2 \big\Vert \bs m^o \big\Vert^2_{L^2(\Omega)}, \nonumber
\end{align}
where we have used Young's inequality to bound the last two terms in the final step i.e. $D_{21} \eps \big| \bs m_p^o \cdot \bs M_p^\eps \big| \le \big( \, D_{21} \eps \frac{1}{\sqrt \Lambda} \big| \bs m_p^o \big| \, \big) . \ \big( \sqrt \Lambda \ \big| \bs M_p^\eps \big| \, \big) \le \frac{D^2_{21} \eps^2}{\Lambda} \big| \bs m_p^o \big|^2+\frac{\Lambda}{2}\big|\bs M_p^\eps \big|^2$ and similarly for the $m^o_3 M^\eps_3$ term. Also $\ \pi \mbox{\fontsize{9}{8}\selectfont $\displaystyle{ \int_{\Omega}}$} \big| \bs m^\eps_p \big|^2=\pi \mbox{\fontsize{9}{8}\selectfont $\displaystyle{ \int_{\Omega}}$} \big|\bs m^o_p \big|^2+\pi \mbox{\fontsize{9}{8}\selectfont $\displaystyle{ \int_{\Omega}}$} \big| \bs M^\eps_p \big|^2+2\pi \mbox{\fontsize{9}{8}\selectfont $\displaystyle{ \int_{\Omega}}$} \bs m^o_p \cdot \bs M_p^\eps$. Subtracting the two and noting $\big\Vert \bs m^o \big\Vert^2_{L^2(\Omega)}=m_s^2|\Omega| $ since $|\bs m^o|=m_s$, we get
\begin{align} \hspace{-7mm}
\mathcal E^\eps_d(\bs m^\eps) -\pi \int_{\Omega} \big| \bs m^\eps_p \big|^2 & \ge \eps Q_1 +\eps^2 Q_2 +\mathcal E^\eps_d(\bs M^\eps)-\pi \int_{\Omega} \big| \bs M^\eps_p \big|^2-\frac{\Lambda}{4} \big\Vert \bs M^\eps \big\Vert^2_{L^2(\Omega)} -D_{23} \eps^2 m_s^2 |\Omega|. \nonumber \end{align}
Using this and eqn. \eqref{A19} we get $\big($ note $ \nabla_p \bs m^o = \nabla_p \bs M^\eps$ $\big)$
\begin{align}
\frac{d}{\eps^2} & \big\Vert \nabla_p \bs m^\eps \big\Vert_{L^2(\Omega)}+\Lambda \big( \, \big\Vert \bs M^\eps \big\Vert_{L^2(\Omega)}+\big\Vert \partial_3 \bs M^\eps \big\Vert_{L^2(\Omega)} \, \big)+\mathcal E^\eps_d(\bs m^\eps) -\pi \int_{\Omega} \big| \bs m^\eps_p \big|^2 \nonumber \\
& \ge \frac{d}{2\eps^2}\big\Vert \nabla_p \bs M^\eps \big\Vert^2_{L^2(\Omega)} +\Lambda \big( \, \big\Vert \bs M^\eps \big\Vert^2_{L^2(\Omega)}+\big\Vert \partial_3 \bs M^\eps \big\Vert^2_{L^2(\Omega)} \, \big)+\eps Q_1 +\eps^2 Q_2 \nonumber \\
& \qquad -D_{20} \eps^{3/4} \Big( \big\Vert \bs M^\eps \big\Vert^2_{L^2(\Omega)}+\big\Vert \partial_3 \bs M^\eps \big\Vert^2_{L^2(\Omega)} \Big)-\frac{\Lambda}{4} \big\Vert \bs M^\eps \big\Vert^2_{L^2(\Omega)} -D_{23} \eps^2 m_s^2 |\Omega|\nonumber \\
& \ge \eps Q_1 +\eps^2 Q_2 +\frac{d}{2\eps^2}\big\Vert \nabla_p \bs M^\eps \big\Vert^2_{L^2(\Omega)}+\frac{\Lambda}{2} \big( \, \big\Vert \bs M^\eps \big\Vert^2_{L^2(\Omega)}+\big\Vert \partial_3 \bs M^\eps \big\Vert^2_{L^2(\Omega)} \, \big)-D_{18} \eps^2 \nonumber \end{align}
for $ \eps $ small enough.
\end{proof}
\begin{RemA} \label{RemA2}
If $ \bs m^o=(0,0,m_s)$, we get a simpler estimate than in above Proposition \ref{proposition3.1.8}  . Note $ \bs m^\eps=\bs m^o+\bs M^\eps$ gives $ |\bs m^\eps|^2=m_s^2=|\bs m^o|^2+|\bs M^\eps|^2+2 \bs m^o \cdot \bs M^\eps=m_s^2+|\bs M^\eps|^2+2 m_3^\eps  M_3^\eps$ which means $ -2 m_3^o  M_3^\eps= |\bs M^\eps|^2$ and $ \bs m_p^\eps \cdot \bs M_p^\eps = \bs 0 $. Substituting these in \eqref{mepsMeps} we get
\begin{align} \hspace{-5mm}
\mathcal E^\eps_d(\bs m^\eps) &= \mathcal E^\eps_d(\bs m^o)+\mathcal E^\eps_d(\bs M^\eps)+2\pi \int_{\Omega}\bs m^o_p \cdot \bs M_p^\eps \bs d \bs x +\frac{8\eps}{3} \int_{\Omega} \Big( \, \bs m^o_p \cdot \bs M^\eps_p -2m^o_3 M^\eps_3 \, \Big) \bs d \bs x \nonumber \\
& \qquad \qquad -\eps^2 \pi \int_{\Omega} \big( \bs m^o_p \cdot \bs M^\eps_p -2 m^o_3 M^\eps_3 \big) \bs d \bs x \nonumber \\
& =\mathcal E^\eps_d(\bs m^o)+\mathcal E^\eps_d(\bs M^\eps)+2\pi \int_{\Omega}\bs m^o_p \cdot \bs M_p^\eps \bs d \bs x + \big( \frac{8\eps}{3} -\eps^2 \pi \big) \int_{\Omega} \big| \bs M^\eps \big|^2 \bs d \bs x. \nonumber
\end{align}
Using the above and proceeding with the remaining part of the estimate in Proposition \ref{proposition3.1.8} we get the following result,
\begin{align}
\frac{d}{\eps^2} & \big\Vert \nabla_p \bs m^\eps \big\Vert_{L^2(\Omega)}+\Lambda \big( \, \big\Vert \bs M^\eps \big\Vert_{L^2(\Omega)}+\big\Vert \partial_3 \bs M^\eps \big\Vert_{L^2(\Omega)} \, \big)+\mathcal E^\eps_d(\bs m^\eps) -\pi \int_{\Omega} \big| \bs m^\eps_p \big|^2 \nonumber \\
& \ge \eps Q_1 +\eps^2 Q_2 +\frac{d}{2\eps^2}\big\Vert \nabla_p \bs M^\eps \big\Vert^2_{L^2(\Omega)}+\frac{\Lambda}{2} \big( \, \big\Vert \bs M^\eps \big\Vert^2_{L^2(\Omega)}+\big\Vert \partial_3 \bs M^\eps \big\Vert^2_{L^2(\Omega)} \, \big) \label{finalremark} \end{align}
for $ \eps $ small enough.
\end{RemA}

\bibliography{refs}

\begin{thebibliography}{}

\bibitem[Acerbi et~al., 1991]{acerbi1991variational}
Acerbi, E.; Buttazzo, G.; and Percivale, D. (1991).
\newblock ''{A variational definition of the strain energy for an elastic
  string}''.
\newblock {\em Journal of Elasticity}, 25(2), pp. 137--148.

\bibitem[Adams and Fournier, 2009]{adams2009sobolev}
Adams, R. and Fournier, J. (2009).
\newblock {\em {Sobolev spaces}}.
\newblock Pure and applied mathematics.

\bibitem[Anzellotti and Baldo, 1993]{anzellotti1993asymptotic}
Anzellotti, G. and Baldo, S. (1993).
\newblock ''{Asymptotic development by $\Gamma$-convergence}''.
\newblock {\em Applied Mathematics and Optimization}, 27(2), pp. 105--123.

\bibitem[Anzellotti et~al., 1994]{anzellotti1994dimension}
Anzellotti, G.; Baldo, S.; and Percivale, D. (1994).
\newblock ''{Dimension reduction in variational problems, asymptotic
  development in $\Gamma$-convergence and thin structures in elasticity}''.
\newblock {\em Asymptotic Analysis}, 9(1), pp. 61--100.

\bibitem[Bhattacharya and James, 1999]{bhattacharya1999theory}
Bhattacharya, K. and James, R. (1999).
\newblock ''{A theory of thin films of martensitic materials withapplications
  to microactuators}''.
\newblock {\em Journal of the Mechanics and Physics of Solids}, 47(3), pp.
  531--576.

\bibitem[Braides, 2002]{braides2002gamma}
Braides, A. (2002).
\newblock {\em {Gamma-convergence for Beginners}}.
\newblock Oxford University Press, USA.

\bibitem[Br{\'e}zis, 2002]{brézis2002recognize}
Br{\'e}zis, H. (2002).
\newblock ''How to recognize constant functions. Connections with Sobolev
  spaces''.
\newblock {\em Russian Mathematical Surveys}, 57, pp. 693.

\bibitem[Brown, 1963]{brown1963}
Brown, W. (1963).
\newblock {\em {Micromagnetics}}.
\newblock Interscience Publishers New York.

\bibitem[Brown, 1966]{brown1966}
Brown, W. (1966).
\newblock {\em {Magnetoelastic Interactions}}.
\newblock Springer-Verlag.

\bibitem[Carbou, 2001]{carbou2001thin}
Carbou, G. (2001).
\newblock ''{Thin layers in micromagnetism}''.
\newblock {\em Mathematical Models and Methods in Applied Sciences}, 11(9), pp.
  1529--1546.

\bibitem[Clark et~al., 2000]{clark2000}
Clark, A.; Restorff, J.; Wun-Fogle, M.; Lograsso, T.; Schlagel, D.; Associates,
  C.; and Adelphi, M. (2000).
\newblock ''{Magnetostrictive properties of body-centered cubic Fe-Ga
  andFe-Ga-Al alloys}''.
\newblock {\em IEEE Transactions on Magnetics}, 36(5 Part 1), pp. 3238--3240.

\bibitem[Downey, 2008]{downey2008thesis}
Downey, P. (2008).
\newblock ''{Characterization of Bending Magnetostriction in Iron-Gallium
  Alloys for Nanowire Sensor Applications}''.

\bibitem[Downey et~al., 2008]{downey2008}
Downey, P.; Flatau, A.; McGary, P.; and Stadler, B. (2008).
\newblock ''{Effect of magnetic field on the mechanical properties of
  magnetostrictive iron-gallium nanowires}''.
\newblock {\em Journal of Applied Physics}, 103, pp. 07D305.

\bibitem[Gioia and James, 1997]{gioia1997micromagnetics}
Gioia, G. and James, R. (1997).
\newblock ''{Micromagnetics of very thin films}''.
\newblock {\em Proceedings: Mathematical, Physical and Engineering Sciences},
  453(1956), pp. 213--223.

\bibitem[Joseph, 1966]{joseph1966ballistic}
Joseph, R. (1966).
\newblock ''{Ballistic demagnetizing factor in uniformly magnetized cylinders:
  J. appi}''.
\newblock {\em Phys}, 37, pp. 4639--4643.

\bibitem[Kohn and Slastikov, 2005]{kohn2005effective}
Kohn, R. and Slastikov, V. (2005).
\newblock ''{Effective dynamics for ferromagnetic thin films: a rigorous
  justification}''.
\newblock {\em Proceedings of the Royal Society A: Mathematical, Physical and
  Engineering Science}, 461(2053), pp. 143.

\bibitem[Landau and Lifshitz, 1935]{landau1935}
Landau, L. and Lifshitz, E. (1935).
\newblock ''{On the theory of the dispersion of magnetic permeability in
  ferromagnetic bodies}''.
\newblock {\em Physik. Z. Sowjetunion}, 8, pp. 153--169.

\bibitem[Le~Dret and Meunier, 2005]{le2005modeling}
Le~Dret, H. and Meunier, N. (2005).
\newblock ''{Modeling heterogeneous martensitic wires}''.
\newblock {\em Mathematical Models and Methods in Applied Sciences}, 15(3), pp.
  375--406.

\bibitem[Maxwell, 1873]{james2007}
Maxwell, J.~C. (1873).
\newblock {\em {A Treatise on Electricity and Magnetism-Volume Two}}.
\newblock A Treatise on Electricity and Magnetism. Lightning Source Inc.

\bibitem[Park et~al., 2010]{park2010characterization}
Park, J.; Reddy, M.; Mudivarthi, C.; Downey, P.; Stadler, B.; and Flatau, A.
  (2010).
\newblock ''Characterization of the magnetic properties of multilayer
  magnetostrictive iron-gallium nanowires''.
\newblock {\em Journal of Applied Physics}, 107(9), pp. 09A954--09A954.

\bibitem[Trabucho and Viano, 1996]{trabucho1996mathematical}
Trabucho, L. and Viano, J. (1996).
\newblock ''{Mathematical modelling of rods}''.
\newblock {\em Handbook of numerical analysis}, 4, pp. 487--974.

\bibitem[Yang et~al., 2006]{yang2006distant}
Yang, Y.; Chen, J.; Engel, J.; Pandya, S.; Chen, N.; Tucker, C.; Coombs, S.;
  Jones, D.; and Liu, C. (2006).
\newblock ''Distant touch hydrodynamic imaging with an artificial lateral
  line''.
\newblock {\em Proceedings of the National Academy of Sciences}, 103(50), pp.
  18891.

\end{thebibliography}
\bibliographystyle{jqt1999}
\end{document}